\documentclass[letterpaper,11pt]{llncs}%
\usepackage{pdflscape}
\usepackage[cp1250]{inputenc}
\usepackage{amsfonts}
\usepackage{amsmath}
\usepackage{amssymb}
\usepackage[letterpaper]{geometry}
\usepackage{lettrine}
\usepackage{mathbbol}
\usepackage{bbding}
\usepackage{extarrows}
\usepackage{graphicx}
\usepackage{setspace}
\usepackage{soul} 
\usepackage{url}
\usepackage[usenames,dvipsnames]{color}
\usepackage{framed}
\usepackage{pstricks}
\usepackage{imakeidx}
\usepackage[framed,thmmarks,amsmath]{ntheorem}
 \allowdisplaybreaks
 \usepackage{paralist}
\usepackage{times}
\usepackage{listings}
\usepackage{etoolbox}
\appto\appendix{
\renewcommand\thechapter{\arabic{chapter}}
}

\definecolor{mygreen}{rgb}{0.66,0.7,0.7}

\def\Aut{\mathrm{Aut}}
\def\C{\mathcal{C}}

\def\D{\mathcal{D}}

\DeclareMathOperator{\Diag}{Diag}

\def\E{\mathtt{E}}

\def\ep{\vspace*{.4cm}}

\def\G{\mathcal{G}}

\def\ld{\lambda}

\def\Z{\mathbb{Z}}

\def\R{\mathbb{R}}

\def\ra{\rightarrow}

\def\WL{\mathtt{WL}}

\def\dd{\diamond}

\def\dw{\dot{\wedge}}

\DeclareMathOperator{\spp}{sp}

\def\Var{\mathrm{Var}}

\def\T{{\mathtt{T}}}

\newcommand{\myfunc}[3]{#1 \,: \, #2\ra #3}

\newcommand{\mset}[1]{\{\hspace{-1.07mm}|\,#1\,|\hspace{-1.06mm}\}}

\makeindex

\newcommand{\tcb}[1]{\textcolor{blue}{#1}}
\newcommand{\tcr}[1]{\textcolor{red}{#1}}

\def\endproof{\hfill$\Box$\vspace{2mm}}

\usepackage{hyperref}

\begin{document}

\title{Description Graphs, Matrix-Power Stabilizations\\  and Graph Isomorphism in Polynomial Time\thanks{This work is supported by National Natural Science Foundation
of China with Grant 62172405.}
}
\author{Rui Xue}
\authorrunning{Rui Xue}
\tocauthor{Rui Xue}
\institute{State Key Laboratory of Information Security,\\ Institute of Information
Engineering, CAS.\\
School of Cyber Security,\\
 University of Chinese Academy of Sciences.\\
\email{xuerui@iie.ac.cn}\\
November 6, 2022}

\date{\today}

\maketitle

 \begin{abstract}
It is  confirmed in this work  that the graph isomorphism can be tested in polynomial time, which resolves a longstanding problem in the theory of computation. The contributions are in three phases as follows.
\begin{itemize}
  \item A description graph $\tilde{A}$ to a given graph $A$ is introduced so that labels to vertices and edges of $\tilde{A}$ indicate the identical or different amounts of walks of any sort in any length  between vertices in  $A$. Three processes are then developed to obtain description graphs. They reveal relations among matrix power, spectral decomposition and adjoint matrices, which is of independent interest.

   \item We show that the stabilization of description graphs can be implemented via matrix-power stabilization, a new approach to distinguish vertices and edges to graphs.  The  approach is proven to be equivalent in the partition of vertices to Weisfeiler-Lehman (WL for short) process. The specific Square-and-Substitution (SaS) process is more succinct than WL process.

      The vertex partitions to our stable graphs are proven to be \emph{strongly} equitable partitions, which is important in the proofs of our main conclusion.  Some properties on stable graphs are also explored.

  \item A class of graphs named binding graphs is proposed and proven to be graph-isomorphism complete. The vertex partition to the stable graph of a binding graph is the automorphism partition, which allows us to confirm graph-isomorphism problem is in complexity class $\mathtt{P}$. Since the binding graph to a graph is so simple in construction, our approach can be readily applied in practice.%
\end{itemize}
Some examples are supplied as illustrations to the contexts, and a brief suggestion of implementation is also given in the appendix.%
\end{abstract}
\clearpage

\section{Introduction}
In this work, a  graph  is an undirected graph with labels assigned to vertices and edges. The labels are all independent variables. Two graphs are isomorphic iff there is a bijection between their vertices that respects adjacency and labels of vertices and edges. The graph isomorphism problem is a computational problem of deciding whether any two given graphs are isomorphic. When the labels on edges are all set to $1$ and the labels to vertices or non-edges are all set to $0$, the graph isomorphism problem here is just that in convention.

 Apart from its importance in practical applications such as in Chemistry, Biology and many other areas, graph isomorphism problem attracts so much attention in the theory of computation due to its specific placement in computational complexity.
 As is well known, whether two complexity classes $\mathtt{P}$ and $\mathtt{NP}$ are equal is an important open problem in the theory of computer science. The importance of exploiting an efficient solution, or refuting  its possibility, to graph isomorphism problem comes from the fact that it is one of the two natural problems (the other is integer factorization problem) potentially with intermediate complexity, possibly neither in $\mathtt{P}$ nor $\mathtt{NP}$-complete (cf. Karp \cite{Karp72}, Garey and Johnson \cite{GareyJo79}).

 Evidence in the literatures \cite{Mathon79,GoldreichMiWi86,GoldwasserSi86,Babai16} hints that graph isomorphism problem may not be $\mathtt{NP}$-hard. The most advanced result \cite{Babai16} presents a procedure in quasipolynomial time. A publicly admitted, provably efficient algorithm for graph isomorphism problem lacks  till now. This work affirms that it is indeed in $\mathtt{P}$.

Graph isomorphism problem to some restricted class of graphs, such as bipartite graphs \cite{BoothCo79}, chordal graphs and so on, is proven to be as difficult as that to all graphs. These graph classes are called graph-isomorphism complete ones. A  new class of graphs, named binding graphs, is proposed and proven to be graph-isomorphism complete in this work.

Aimed at testing of graph isomorphism in this work, three processes are presented to distinguish vertices and edges of graphs. These processes are equivalent in the sense that the final graphs (description graphs so named) obtained by them to a given graph are equivalent. A description graph to a given one will reflect, per labels in it,  the identical or different of amount of walks, of any length, between vertices in the given graph. However, the amount of walks between vertices is preserved by the automorphisms of a graph.  The processes therefore are approaches to distinguishing vertices and edges with respect to the automorphisms. A stable graph obtained by stabilization of description graphs will possess the same automorphism group as the originally given graph.

The stable graphs of binding graphs are proven to enable automorphism partitions in this work. Since graph isomorphism problem is equivalent to the automorphism partition \cite{Mathon79} and the class of binding graphs is graph-isomorphism complete, it allows us to obtain a procedure for testing of graph isomorphism. The fact that the procedure is in polynomial time confirms that graph isomorphism problem is in the complexity class $\mathtt{P}$.

In the next subsection we describe the routines of this work, including techniques involved and results obtained. The differences of methodology employed in this work from those in existing works are presented in the subsequent subsection. The related works and the arrangement of the contexts are then introduced in the last subsections of this  introduction.

\subsection{Approaches, Techniques and Results}

The labels to vertices and edges of a graph used in the whole context (except in Subsection \ref{sec:bstv}) are formal variables which are independent of each other, and a particular variable $x_0$ is reserved to represent the non-edges (named blank edges in this paper) in a graph. In this way, all graphs are labeled complete graphs, which are expressed as matrices of labels. That could extremely simplify the proofs, and also avoid the faulty stability that may happen when numeral labels are used as labels (cf. the first example on page \pageref{NuSaS}). By the way of labeling $x_0$ on blank edges, the notions of simple graphs, connected graphs and bipartite graphs, and so on can be defined in a consistence way with those in convention (cf. e.g. Godsil and  Royle \cite{GosilRo2001}).

A technique frequently used is the imbedding relation which is introduced by Lehman \cite{Weisfeiler76} (we cite the Lehman \cite{Weisfeiler76} rather than Weisfeiler or others, following the suggestion there).
A matrix  $A:=(a_{ij})$ is said to be imbedded in a matrix $B:=(b_{ij})$, denoted as $A\rightarrowtail B$, if $b_{ij}=b_{rs}$ implies $a_{ij}=a_{rs}$. If $A$ and $B$ are both graphs, it means $B$ is a refinement of $A$ whenever $A\rightarrowtail B$.  Two graphs $A$ and $B$ are \emph{equivalent} if and only if they are imbedded in each other, denoted as $A\approx B$ in that case.  Apparently, $X\rightarrowtail Y$ implies  $\Aut(Y)\subseteq\Aut(X)$.

In this work graph isomorphism problem is investigated with vertex partition methods. The article is composed of three components: processes to description graphs, matrix-power stabilization and its relation to Weisfeiler-Lehman (WL for short) process, and binding graphs and the automorphism  partitions. We detail them one by one in the following.
\subsubsection{Processes to description graphs:} A simple observation is well known regarding to graph isomorphism: If an automorphism $\sigma$ of a graph $G$ sends a vertex $u$ to vertex $u^\sigma$, the amount of walks of any length $t$ from $u$ to all vertices of $G$ should be correspondingly the same as the amount of walks of length $t$ from vertex $u^\sigma$ to all vertices. The matrix power $A^t$ of adjacent matrix $A$ to graph $G$ records the number of walks of length $t$ between any pair of vertices in $G$. We use the adjacent matrix $A$ as graph $G$ hereafter.

Since an automorphism preserves the amount of walks of any length, the amount of walks of all lengths should be counted. That leads to the $\ld$-matrix $S(\ld)=\sum_{k=0}^\infty\ld^kA^k$, which is proven to be equivalent, in the sense as introduced above, to the adjoint matrix $\mathrm{adj}(A)$ of $\ld I-A$, where $I$ is the identity matrix. If $\mathrm{adj}(A):=(a_{ij})$, then it is possible that $u^\sigma=v$ only if $a_{uu}=a_{vv}$ for all vertices $u,v$ and any automorphism $\sigma$ of $A$.

We then substitute entries of $\mathrm{adj}(A)$ with labels so that the identical entries have the same label, and unidentical entries have different labels. Such a substitution is named as \emph{an equivalent variable substitution} and   frequently used in the context.

The graph obtained by equivalent variable substitution is  a new graph, and named as \emph{a description graph} $\tilde{A}$ of $A$. For description graphs, one of approaches is to compute the adjoint matrix of $\ld I-A$ and then perform an equivalent variable substitution to obtain the description graph $\tilde{A}$ (cf. Subsection \ref{sec:adj}).

One may notice that $S(\ld)$ is an infinite sum, which is proven not to be necessary. We have shown that if $m$ is the degree of $A$'s minimum polynomial, then $\Gamma(A,m-1):=\sum_{k=0}^{m-1}\ld^kA^k$ is equivalent to $\tilde{A}$. To obtain a description graph of $A$ from  $\Gamma(A,m-1)$ is the second process proposed (cf. Subsection \ref{sec:mp}).

The third process is based on the spectral decomposition of $A$. If $A$ has $\alpha$ distinct eigenvalues $(\mu_1,\ldots,\mu_\alpha)$ and  $A=\sum_{k=1}^{\alpha}\mu_k\,\E_k$ is the spectral decomposition, we are able to claim that $\sum_{k=1}^\alpha\ld^k\,\E_k\approx\tilde{A}$ for $\mu_k\ne0$. Which gives the third process to a description graph of a graph $A$ (cf. Subsection \ref{sec:bstv}).

The description graph to a given graph is unique up to equivalents in the sense as stated previously. It reveals the differences of amount of walks, of any sort, of arbitrary length among vertices in the given graph.

The processes hence reveal the relationship among adjoint matrices, the power of adjacent matrix and spectral decomposition together. That could be of independent interest.

The description graph of a description graph can be further worked out, and this process proceeds until no difference is further distinguished for some description graph.  The graph finally obtained is the \emph{stable graph} $\hat{A}$ to  $A$.

The stabilization process will get a stable graph $\hat{A}$, which  is an undirected graph provided $A$ is. That should be contrasted to the stable graph from the well-known WL process, where a stable graph can be a directed or partially oriented graph to an undirect graph \cite{Weisfeiler76}.

It is not hard to see that the automorphism group $\Aut(\tilde{A})$ of the description graph $\tilde{A}$ coincides with the automorphism group $\Aut(A)$ of $A$. That leads to the conclusion that $\Aut(A)=\Aut(\tilde{A})=\Aut(\hat{A})$, where graph $\hat{A}$ is the stable graph of $A$.

We are then able to show that the vertex partition of a stable graph is a \emph{strongly} equitable partition, a notion posed here as an enhancement of equitable partitions in  literature. The vertex partition to a labeled graph is a partition of vertices such that all vertices with the same label are in the same cell. Such a partition is equitable if the labels on the edges from one vertex of a given cell to all vertices of another given cell as a whole are independent of the vertex chosen from the first given cell, for each pair of cells in the partition, even if two cells are the same one.

The strongly equitable partitions require further for equitable partitions that the labels on the edges connecting any two cells appear only on edges between the given pair of cells, which means the labels on edges between different pairs of cells do not overlap in this case. Again it should be true when two given cells are the same one.

The vertex partitions to stable graphs are called stable partitions and are proven to be strongly equitable partitions. That is important in the proofs of conclusions explained later.  A vertex partition to a graph is an automorphism partition if the cells in the partition are exactly the orbits of its automorphism group.   Some other properties on stable graphs are also presented in this work (cf. Section \ref{sec:st&ppro}).

\subsubsection{Matrix-power stabilization and its relation to WL process:}

Weisfeiler-Lehman process is powerful in distinguishing vertices and edges with respect to graph isomorphism, and mostly employed in works of graph isomorphism. The stable graphs obtained by our processes are equivalent to the stable graphs obtained by WL process in the partition of vertices. The stable graphs obtained by WL process are finer in the partition of  edges.

A graph $A$ is said to recognize vertices if the labels on vertices do not overlap with the labels on edges, which means the labels at the diagonal of $A$ will not appear as non-diagonal entries in $A$.

For a graph $A$, we show that $A\rightarrowtail \cdots\rightarrowtail A^k\rightarrowtail \cdots$ for all $k>0$, provided $A$ recognizes  vertices.  That implies in this setting that $\tilde{A}\approx A^{n}$ for a graph of order $n$ (cf. Corollary \ref{coro:power}).

The result above allows us to produce a stabilization process for all $k>1$ as follows: Make given graph $A$  be a graph $A_1$ that recognizes vertices by an equivalent variable substitution to the diagonal of $A$ at first. Then proceed recursively to compute and produce equivalent variable substitution to  $A_i^k$ and obtain $A_{i+1}$ for $i>1$, until $A_t\approx A_{t+1}$ for some $t>0$. The graph $A_t$ is then the stable graph to $A$. That process is called \emph{matrix-power $k$-stabilization}.

 The matrix-power $2$-stabilization is just a procedure of \emph{square and substitution} (SaS). We prove SaS process is  equivalent to WL process in the partition of vertices in a stable graph, and with the same number of iteration (Theorem \ref{thm:saswl}).

The matrix-power $k$-stabilization process gains the same similarity to $k$-walk-refinement in \cite{LichterPoSc19} by Lichter,  Ponomarenko, and Schweitzer, just as SaS process to WL process. An upper bound $O(n\log n)$ to the iteration number of WL process is shown there. The iteration number of SaS process is hence in $O(n\log n)$ by Theorem \ref{thm:saswl}.

\subsubsection{Binding graphs and the automorphism partitions:}

A key point of this work is to apply the properties of stable partitions to a class of so-called binding graphs. Roughly speaking, given a graph $A$ of order $n$,  for each pair of vertices $u,v$ in $A$, we add a unique binding vertex $p:=u\dot{\wedge}v$ adjacent only to $u$ and $v$. The graph obtained is \emph{the binding graph of $A$} which now has $n_1:=n(n+1)/2$ vertices. The new added binding vertex $p$ is one of binding vertices of $u$, and each vertex $u$ from $A$ has $n-1$ binding vertices in the binding graph.  The new added edges are called binding edges. The given graph $A$ is then called the basic graph of the binding graph, and the vertices and edges of $A$ are basic vertices and basic edges, respectively.
 A binding graph is uniquely determined by the basic graph, up to the renaming of binding vertices.

This kind of uniqueness leads to the claim that the class of binding graphs is graph-isomorphism complete (Theorem \ref{thm:complete}), in the sense that two binding graphs are isomorphic if and only if their basic graphs are.

Let $B$ be the binding graph to a basic graph $A$ of order  $n$. If all the basic edges in $B$ are removed and only binding edges remain, it will become to be a bipartite graph $\Phi$ with all basic vertices as one part and all binding vertices as the other. Each basic vertex now is of degree $n-1$ and each binding vertex is of degree $2$ in $\Phi$ for $n>2$.

If such a bipartite $\Phi$ inherits the labels on vertices and binding edges from the stable graph $\hat{B}$, the vertex partition to labeled $\Phi$ is an equitable partition already, thanks to the strongly equitable stable partition to $\hat{B}$. The graph $\Phi$ in the following is such a labeled graph. We are able to show that $\hat{\Phi}\approx\hat{B}$ and $\Aut(\hat{\Phi})=\Aut(\hat{B})$  (Theorem \ref{thm:xphi}).

The label on a binding edge is then determined by the pair of their end labels, and the labels to basic vertices in $\Phi$ are in fact completely  determined by the labels of binding vertices. The labels to $\Phi$ can then  be further relaxed by removing the labels to basic vertices and updating the labels to all binding edges into a new same one. A new  bipartite graph $\Theta$  is then obtained in this way from $\Phi$ so that only binding vertices of $\Theta$ are with labels inherited from $\hat{B}$. We show that $\hat{\Theta}\approx\hat{\Phi}$ and $\WL(\Theta)\approx\WL(\Phi)$, and hence $\Aut(\Theta)=\Aut(\Phi)$ (cf. Theorem \ref{thm:theta}).

We are able to show that the stable partition to $\Theta$ is the automorphism partition (Theorem \ref{thm:main}) and so is to $\Phi$.
Since $\Aut(\Phi)=\Aut(\hat{B})=\Aut(B)$, it leads to the conclusion that the stable partition to stable graph $\hat{B}$ is the automorphism partition.

That achieves the goal of this work, since the well known claim that testing graph isomorphism is polynomial-time equivalent to computing orbits of automorphism groups (\cite{Mathon79},\cite{BoothCo79}). We are then able to form a polynomial-time procedure $\mathtt{GI}$ for testing of graph isomorphism.

The procedure $\mathtt{GI}$ is intuitively simple: For two graphs of the same order to be tested, treat their disjoint unions as basic graphs to get a binding graph $B$ and compute stable graph $\hat{B}$ obtained by SaS process. The cells of the stable partition are then checked if  each cell of them is shared by the vertices of two graphs. If they are, then two graphs are isomorphic, otherwise two graphs are not (cf. Section \ref{sec:gid}).

In the proof of $\Aut(\Phi)=\Aut(\hat{B})$, two ingredients are essentially applied. One is the stability of the stable graphs $\hat{B}$, which leads to a strongly equitable vertex partition and hence a equitable vertex partition of  $\Phi$. Another one is the fact that for each pair of basic vertices $u,v$  and the binding vertex $p=u\dw v$ in the binding graph $B$, the labels on two binding edges $(p,u)$ and $(p,v)$ in $\hat{B}$ recognize whether $(u,v)$ is a basic edge or basic blank edge in $B$ (Lemma \ref{lem:shbysb}).

Two interesting results reveal the power of binding graphs. One is the fact that, in the stable graph $\hat{B}$ obtained by SaS process,  the labels on binding vertices are in fact equivalent to the labels on the basic edges  in graph $\hat{B}$ (cf. Lemma \ref{lem:bv}).
The other result is that, \emph{still} in $\hat{B}$, the labels on binding edges are in fact equivalent to the labels on the basic edges of $\WL(B)$ (cf. Lemma \ref{lem:wl}), where $\WL(B)$ is the stable graph of $B$ obtained by WL process, although  $\hat{B}$ is obtained by SaS process.

These results mean that one can evaluate a stable graph $\hat{B}$ of a binding graph by SaS, and obtains simultaneously the stable basic stable graph by WL process (one stone two birds). That allows us freely to use the properties of stable basic graphs obtained by WL process in  the discussion about the stable graph $\hat{B}$. That happens only in proof of Lemma \ref{lem:tech}.

Due to the equivalent capacities of SaS process and WL process in the partition of vertices,  WL process, instead of SaS process, can be used in procedure $\mathtt{GI}$  to obtain the stable partition.

%
%
%

%
%

In summary, we propose three processes to obtain description graphs, which reveal the relations among matrix power, spectral decomposition and adjoint matrix. Based on those methods, matrix-power stabilization is then introduced to stable graphs. The properties of stable graphs, especially the strongly equitable vertex partitions, are discussed.   A class of binding graphs is proposed and proven to be graph-isomorphism complete. Stable partitions to binding graphs are shown to be automorphism partitions, which leads to a polynomial-time testing procedure to graph isomorphism.

\subsection{Methodology Clarifications}\label{sec:argument}
As stated in the last subsection that the decision procedure $\mathtt{GI}$ posed in this work can also be implemented with WL process.  Given that WL process has appeared for a long time and has been extensively exploited in a great deal of works, the author would like to point the differences of the methods here from those in existing works.

\subsubsection{Both vertices and edges}

The first one we would like to point out is that although the WL process is adopted extensively in many articles on graph isomorphism, the focuses in these works were mainly on the distinguishing of vertices. The distinguishing of edges with WL process is not treated or explored fairly enough. This is witnessed by the fact of, to author's knowledge, the lacking of a notion like \emph{strongly} equitable partition, one of essential ingredients in this work, by stable graphs appeared yet in the previous works, although the notion of equitable partitions (colored or not) were commonly mentioned.
Also the discussions about graph $\Phi$ and $\Theta$ depend heavily on the information on edges.(cf. Lemma \ref{lem:shbysb}, Lemma \ref{lem:tech}).

That might partially explain the result by Neuen and Schweitzer \cite{NeuenSc18}, where an exponential lower bound for individualization-refinement (of vertices) algorithms for graph isomorphism is presented,


\subsubsection{Distinguishing and identification}
The  procedure $\mathtt{GI}$ in this work is to distinguish and partition the vertices and edges to the binding graph of two target graphs to be tested. It treats the two given graphs simultaneously in one connected graph, which is the main difference in methodology to existing works in literatures, where two graphs are usually  treated separately with WL process and then to try to identify them. In other words, identifying two given graphs only by separate stabilizations are the main methodology in existing literatures, while ours is to identify them in one shot in a connected graph (of them).

In the seminal article by Cai, F{\"u}rer and Immerman \cite{CaiFuIm92}, it is shown that WL process can not, in the way of separate treatments, distinguish some graphs that are not isomorphic, even with the extended WL process in some higher dimension.

The power of WL-like processes, however, lies at distinguishing vertices as well as edges within ONE connected graph. It should not be expected to distinguish  graphs in a separate way with WL-like process.

With regard to isomorphism, however, two target graphs can be  related as a whole as one connected graph. In that case, the power of WL-like processes will play a role simultaneously to two graphs. Similar arguments appeared in \cite{Douglas11} too.
The author of \cite{Douglas11} (and of this article) finds, in fact, that it already obtains the automorphism partitions for some pairs of graphs proposed by Cai, F{\"u}rer, and Immerman \cite{CaiFuIm92} simply to their union graphs, and decides their nonisomorphism with WL process.

A problem faces us immediately: Is a procedure of identifying or distinguishing a graph from ALL those unisomorphic to it, once for all, polynomial-time equivalent to testing of graph isomorphism? Given so many years of endeavors without success and the conclusion in this work, it is doubted so.

Theoretically, the identification of a graph from those unisomorphic ones might not be polynomial time equivalent to the testing of graph isomorphism. That is true at least with respect to WL-like processes. Here is a simple argument: It is known that testing of isomorphism is successful in the trees with WL process in polynomial time. However, given two unisomorphic trees of  order $n$, if both possess only unit automorphism groups of its own, WL process will produce discrete partitions to each of them, since the stable graph to each given tree in this case is labeled with $n^2$ distinct elements. Hence one cannot distinguish (or identify) one from the other with WL-like processes.

As a common observation, the result by Cai, F{\"u}rer, and Immerman \cite{CaiFuIm92} does not eliminate the possibility with WL process to successfully test the isomorphism in the restricted class of graphs, like trees \cite{ImmermanLa90}, cographs and interval graphs (cf. \cite{EvdokimovPo00},\cite{Laubner10}, \cite{KoblerKuLaVe11}). The class of binding graphs posed in this paper is proven another such class which is proven to be graph-isomorphism complete.

\subsubsection{Why do binding graphs work?} A binding graph of two target graphs will bind them as one connected graph, although we technically employ wing graphs in our decision procedure which is not essential. The distinguishing with WL-like processes will be produced simultaneously to both target graphs.

In that way, the labels to binding vertices and edges in stable graphs will record or indicate not only the local deviations but the global deviations during the stabilizations. That happens since each basic vertex in a binding graph has a connection to each of other basic vertex via a binding vertex separately.

We have shown that the labels of binding edges in stable graph will recognize the adjacency of two basic vertices which are binding (see Lemma \ref{lem:shbysb} in Section \ref{sec:bdt} for details). Since the stabilization procedure is iterations of computing the description graphs, any changes between basic vertices and edges will be recorded by the binding vertices and edges and spread to all other vertices during the stabilization.

The local deviations tested will be directly transferred via binding vertices and edges to all the (basic) vertices and edges during the stabilization process. In this way, the global distinguished deviations will reflect all local deviations and vice versa in stable graphs. The differences of two graphs are then tested in their binding graph.

\subsection{Related Work}\label{sec:rw}

Graph isomorphism problem is extensively studied in literatures. Some of them aim at practical applications and the others at theoretical explorations. This part will not intend to survey literatures on graph isomorphism problem. Only the most advanced or the intimated works will be mentioned according to the knowledge of the author. Some works may not be fairly treated or cited here due to author's restricted knowledge. When that happens, please kindly remind me via email.

A great success has been made in practical applications with algorithms like Bliss \cite{JunttilaKa07}, Conauto \cite{Lopez-PresaChAn14}, Nauty \cite{McKayPi14}, Saucy \cite{DargaLiSaMa04}, Traces \cite{McKayPi14}, VF2 \cite{CordellaFoSaVe99} and Vsep \cite{Stoichev19}, to name a few. Since the intention of this work is a theoretical investigation, their practical advances will not be further addressed.

The most advanced result till now is due to Babai \cite{Babai15,Babai16}, where  a procedure for testing of graph isomorphism in quasipolynomial time was presented. The procedure in Babai \cite{Babai15,Babai16} exploits and composes the techniques in group theory and combinatorics ingeniously together, hence achieves a result  that pushes graph isomorphism problem close to the board line  of $\mathtt{NP}$ and $\mathtt{P}$. Grohe, Neuen and Schweitzer \cite{GroheNeSc18} gave more efficient algorithm for graphs of bounded degrees.

 The application of group theory to tackle graph isomorphism problem has made a great success. It emerged in  Babai \cite{Babai79}, and the seminal work by Luks \cite{Luks80} makes it more popular in the community, where profound results are explored and then frequently adopted later on. We refer readers to Grohe, Neuen and Schweitzer \cite{GroheSc20,GroheNe21} and to Babai \cite{Babai18}, and the references there for details.

The classification of vertices by the degrees, paths, and so on, is a natural way in the decision of graph isomorphism \cite{ReadCo77}. It is involved, more or less, in most works subject to graph isomorphism, and initiated by Morgan \cite{Morgan65} as reported in literature. The most popular  approach belongs to Weisfeiler and Lehman \cite{Weisfeiler76,WeisfeilerLe68}, so called  Weisfeiler-Lehman (WL) process. Original WL process is commonly referred as 2-dimensional WL process and extended to $k$-dimensional WL \cite{BabaiMa80,ImmermanLa90,CaiFuIm92} processes for positive integer $k>2$. Readers may refer F{\"u}rer \cite{Furer17} and Kiefer \cite{Kiefer20} and the references there for  references about WL process.

As previously stated, the description graphs proposed in this paper are inspired by the simple observation that an isomorphism preserves not only adjacency but also the paths and walks. The numbers of walks between vertices are naturally used as a handle to distinguish vertices and edges, and the matrix power are then natural tools. The application of this handle with regard to graph isomorphism emerged  explicitly in Morgan \cite{Morgan65} and R{\"u}cker and R{\"u}cker \cite{RuckerRu90,RuckerRu91}, all in the field of computational chemistry. R{\"u}cker and R{\"u}cker adopted matrix power to distinguish vertices of graphs, and neglected the information of edges produced. Tinhofer and Klin \cite{TinhoferKl99} extensively discussed the stabilization procedures, especially the total degree partition \cite{Tinhofer91} was developed.

Powers and Sulaiman  \cite{PowersSu82} applied the number of walks to partition the vertices of graphs and related the partition to graphs spectra. These are all documents the author finds that explicitly employ the matrix power as a handle directly to individualize the vertices with regard to graph isomorphism. The walks in graphs, however, are commonly treated in graph isomorphism and homomorphism \cite{Godsil93}, \cite{HellNe04}.

Since WL process emerged in 1960's, it is not surprised that a process like SaS was noticed in the community. The author find recently that the square-and-substitution process was already mentioned by F{\"{u}}rer in \cite{Furer01} (page 323 in \cite{Furer01}), although it was not explored in details. F{\"{u}}rer showed a lower bound to the iteration number of WL process in \cite{Furer01} and investigated the relation of labels of edges to spectral properties of graphs in \cite{Furer95}. The process to description graphs in Section \ref{sec:bstv} relate the labels of graphs with spectral decompositions.

Recently, Lichter, Ponomarenko and Schweitzer extended WL process and introduced walk-refinement approach by counting any length of walks in refinement instead of just length $2$ as in WL process \cite{LichterPoSc19}. It is  proven to be equivalent to WL process in stabilization. They also proved the iteration number of walk-refinement approach is $\Theta(n)$ to a graph of order $n$, and claimed an iteration upper bound $O(n\log n)$ to WL process. As an extension of WL process, walk-refinement process will generally produce directed (or mixed graphs meaning some edges are oriented and some not) stable graphs for undirected graphs rather than symmetric ones.
%

Inspecting eigenvalues and eigenvectors of adjacent matrices to characterize graphs is another way to classify vertices of graphs with regard to graph isomorphism. Profound results are obtained, and, for example, star partitions are introduced by Cvetkovi\'c, Rowlinson and Simi{\'c} \cite{CvetkovicRoSi93}.

The equitable partitions are frequently pursued with partitions in the literatures \linebreak (cf. e.g. \cite{McKay76}, \cite{McKay81}, \cite{Godsil93}, \cite{Kiefer20}), since the automorphism partitions  are equitable ones. Although the (colored) equitable partitions induced by WL process are exploited in the literatures (cf. e.g. Kiefer  \cite{Kiefer20}), the notion like strongly equitable partitions are not proposed yet, to the knowledge of the author.

Again, this author may not aware of some works that should be cited, due to the author's lack of knowledge. I beg readers to inform me if that happens.
\subsection{The Structure of this Work}
The notions and notations  are given in Section \ref{sec:bst1}. The notion of description graphs is proposed , and three processes to description graphs are developed in Section \ref{sec:decripgraph}. Stabilization of description graphs and relevant properties are discussed in Section \ref{sec:st&ppro}. We propose matrix-power stabilization approach to stable graphs in Section \ref{sec:mpr}, and the equivalence of stable graphs obtained by our processes to those obtained by WL process is verified.
 In Section \ref{sec:wdtxz}, we prove strongly equitable partitions to stable graphs and some other properties.

 In Section \ref{sec:bdt} the class of binding graphs are introduced and proven to be graph-isomorphism complete. The stable graphs, and the bipartite graphs $\Phi$, of binding graphs are discussed in Section \ref{sec:phi}. A further bipartite graph $\Theta$ is constructed in Section \ref{sec:theta}, and used  in Section \ref{sec:aut} to show the key result about the automorphism partitions to binding graphs. A testing procedure $\mathtt{GI}$ for graph isomorphism is then presented and shown in polynomial time in Section \ref{sec:gid}. A brief discussion is given in the last Section \ref{sec:discuss}.

\section{Preliminaries}\label{sec:bst1}

Following the convention, the set of integers, positive integers and reals are denoted as $\Z,\Z^+$ and  $\R$, respectively. $[m..n]$ is the set $\{m,m+1,\ldots,n\}$ of integers from  $m$ to $n$ with special $[n]:=[1..n]$. We always assume  $n$ a positive integer.  For any permutation $\sigma$, $j=i^\sigma$ means $\sigma$ sends $i$ to $j$.

A multiset is a set that allows an element to appear multiple times in it. While $\{a,b,c,\ldots\}$ is a general set, $\mset{a,b,c,\ldots}$ will indicate a multiset. Two multisets  $S_1$ and $S_2$ are equal iff they have the same element when counting their multiplicities, denoted as  $S_1\equiv S_2$.
\subsubsection{Definition of graphs}
Graphs considered in this work are all undirected with labels to vertices and edges, following Weisfeiler and Lehman  \cite{WeisfeilerLe68,Weisfeiler76}, where a set of independent real variables $x_0,x_1,x_2,\ldots$ is adopted to label vertices and edges of graphs. We reserve  $x_0$ as a special variable which will signal a ``non-edge'' in a graph.  That will make it convenient to describe, say, simple graphs and connected graphs and so on in our cases. The operations, like commutativity of multiplication, will be conformed by these variables. According to the properties of formal invariants, we know that, for all $i,j,r,s,u,v\in\Z^+$,  $x_ix_j+x_r=x_ux_v+x_s$ iff $\mset{\mset{x_i,x_j},x_r}\equiv\mset{\mset{x_u,x_v},x_s}$, etcetera.
Let $\Var:=\{x_1,x_2,\ldots\}$ and $x_0\notin\Var$. The notion of graphs is formally given as follows.

\begin{definition}[Graphs]\label{def:graph}
 Let $V\subseteq\Z^+$ be a nonempty set.  A graph $\G$ over $V$ of order $|V|$ is a function  $\myfunc{g}{V\times V}{\{x_0\}\cup\Var}$ satisfying $g(u,v)=g(v,u)$ for all $u,v\in V$. Elements $u,v\in V$ are called vertices of $\G$.  $(u,v)$ is an edge with label  $g(u,v)$  if $g(u,v)\ne x_0$, and $(u,v)$ a blank edge with label $x_0$ if $g(u,v)=x_0$.  A vertex $v$ is said to be a neighbor of, or adjacent to $u$ iff $u\ne v$ and $g(u,v)\ne x_0$.  Specially,  $g(u,u)$ is the label to vertex  $u\in V$.%
\end{definition}
A graph $\mathcal{G}$ over $[n]$ can be conveniently formed as a symmetric matrix $G:=(g_{ij})$ of order $n$, where  $g_{ij}:=g(i,j)$ for  $i,j\in[n]$. We will often refer to $G$ as the graph of order $n$, instead of $\G$ over $[n]$. It should be stressed that edges are those with labels other than $x_0$,  blank edges those with labels $x_0$ in a graph.

The degree $\deg(u)$ of a vertex $u$ in $G$ is the number of its neighbors. The dimension $\dim(G)$ of graph $G$ is the number of distinct entries in $G$. It is  straightforward that $\dim(G)\le n(n+1)/2$ for any graph $G$ of order $n$. We will abuse the notation $x\in G$ whenever $x$ is some entry in $G$. Similarly, to $x\in G_1\cap G_2$ or $x\in G_1\cup G_2$.
 All graphs involved in this work are undirected graphs, and hence the matrices in this paper are all symmetric ones, except when WL process is discussed and used, which will be explicit in the context. Performing a \emph{row-column permutation} to a matrix is to multiply a permutation both on the left and on the right of the matrix.

For a graph $G$ of order $n$ and $\emptyset\ne V\subseteq[n]$, the submatrix $A$ obtained from $G$  by removing those rows and columns not in $V$ is a subgraph of $G$ induced via vertices in $V$.

A simple graph $G=(g_{ij})$ of order $n$ is a graph with $\dim(G)\le 2$ and $g_{ii}=x_0$ for all $i\in[n]$. If the vertices set $V$ of graph $G$ is split into two nonempty subsets $S_1,S_2$ such that the edges between vertices in $S_i$  are all blank edges for $i=1,2$, then $G$ is called a bipartite graph.
\subsubsection{Walk and sort}
A walk of length $t$ in a graph $G=(g_{ij})$ is a sequence of vertices $W:=\langle i_0,i_1,\ldots, i_t\rangle$ with
 $g_{i_{k-1}i_{k}}\ne x_0$ for all $k\in[t]$. The \emph{ordered  sort} of walk $W$ is a sequence of labels
$$\langle g_{i_0i_1},g_{i_1i_2},\ldots,g_{i_{t-1}i_t}\rangle\,.$$
 There could be many walks with the same sort between a  pair of vertices. However, the \emph{unordered sort} of $W$, or just \emph{sort} of $W$,  is the multiset
$$\mset{g_{i_0i_1},g_{i_1i_2},\ldots,g_{i_{t-1}i_t}}\,.$$
WL process applies ordered sorts of 2-walk to any pair of vertices in a graph. The process in this work will apply unordered sorts of walks in a graph.

\subsubsection{Graph isomorphism}

Given two graphs $A=(a_{ij})$ and $B=(b_{ij})$ of order $n$,
if there is an 1-1 map $\sigma$ on $[n]$ such that  $a_{i^\sigma j^{\sigma}}=b_{ij}$ for all $i,j\in[n]$, we say that $A$ is isomorphic to  $B$ and denote as $A\cong B$. This $\sigma$ is then an isomorphism from $A$ to $B$.

Equivalently,  there is an isomorphism $\sigma$ for graphs $A$ and $B$ iff there is a permutation matrix $P$ of order $n$ such that $PAP^\T=B$, where $P^\T$ is the transpose of $P$. Both of the two forms  will be used later in the context.

The isomorphisms from $A$ to itself are automorphisms of $A$. The collection of all automorphisms is then denoted as $\Aut(A)$, which is a permutation group as the automorphism group of graph $A$. An orbit of $\Aut(A)$ is a set of vertices in $A$ satisfying not only that each vertex in it is mapped to a vertex in the set by any automorphism of $A$, but also that any two vertices in this set will be mapped one to the other by some automorphism of $A$.

For a graph $A$ on $[n]$, the partition $\mathcal{C}:=( C_1,\ldots,C_s)$ of $[n]$ consisting of all orbits $C_i$ of $\Aut(A)$ is called the automorphism partition of $A$.

\subsubsection{Substitution and Imbedding}
The imbedded graphs and equivalent graphs will be the predominate tools in this paper. They are introduced in \cite{Weisfeiler76}. It should be pointed out that the following notions are for matrices, not for graphs  only.
 \begin{definition}[Imbedding and equivalence]\label{def:equiva} Let  $A:=(a_{ij})$ and $B:=(b_{ij})$ be two matrices of order $n$.
\begin{itemize}
  \item If $b_{ij}=b_{st}$ implies $a_{ij}=a_{st}$ for all $i,j,s,t\in[n]$, then  $A$ is said to be imbedded in $B$ and denoted as $A\rightarrowtail B$.
  \item If $A\rightarrowtail B$ and $B\rightarrowtail A$, then $A$ is said to be equivalent to  $B$ and denoted as  $A\approx B$.%
\end{itemize}
\end{definition}

In the case of graphs, the following properties are easily to obtain from definitions.
\begin{proposition}[\cite{Weisfeiler76}] \label{prop:lehman} For any graphs $A$, $B$ and $X$ of order $n$, we have
\begin{itemize}
\item  $A\rightarrowtail B$ and $B\rightarrowtail X$ imply $A\rightarrowtail X$.
  \item $A\rightarrowtail B$ implies $\dim(A)\le\dim(B)$ and $\Aut(B)\subseteq\Aut(A)$.
  \item  $A\approx B$ implies $\dim(A)=\dim(B)$ and $\Aut(A)=\Aut(B)$.%
\end{itemize}
\end{proposition}
It is often to relabel or replace the entries of a matrix with variables in $\Var$ to get a graph.

\begin{definition}[Equivalent variable substitution]
Given a matrix $X$ of order $n$, substitute entries in $X$ with  variables in  $\Var$ in a way that the identical entries with the same variables and unidentical entries with different variables. The resulting matrix $Y$ will be a graph equivalent to $X$.
This procedure is named as \emph{an equivalent variable substitution} to $X$.%
\end{definition}
 Since $x_0\notin\Var$, the graph obtained by an equivalent variable substitution will be a labeled complete graph.

\noindent\textbf{Remark.} We point out that, by Definition \ref{def:graph}, a graph can be expressed equivalently in many different ways as various of variables. All these expressions are equivalent in the sense of Definition \ref{def:equiva}. Our definition to graph isomorphism only concerns with two graphs with the same set of labels (called restrict isomorphism  in some of literatures). This does not harm to decide the isomorphism of any two graphs with different set of labels, since they can, if they are isomorphic, be relabeled into the same set of labels by an equivalent variable substitution.

For two simple graphs, this can be done simply by replacing  labels to edges in both graphs with the same variable, e.g. $x$, and then to decide their isomorphism. In our decision procedure $\mathtt{GI}$, the isomorphism of any two simple graphs is considered only.
\section{Description Graphs}\label{sec:decripgraph}
In this section, we introduce the notion of description graphs and develop  three approaches to evaluate them.
The motivation for defining description graphs comes from the observation that, in a graph $A$, a vertex  $u$ is possibly carried to a vertex $v$  by an automorphism of $A$ only if the number of walks of any sort that $u$ and $v$ to all vertices in $A$ as a whole coincide. The same ideas are applied in WL process\cite{Weisfeiler76}, where ordered sorts of walks of length 2 are counted. However, in contrast to WL process, the number of unordered sorts of walks are counted in description graphs. That significantly simplifies the computations and expressions.

Intuitively, the description graph to a graph $A$ will distinguish vertices (and edges) by the number of walks, with multiplicity counted, of the same sort of arbitrary length in $A$.

\begin{definition}[Description graphs]\label{def:descgraph}
For any graph $A$ of order $n$, a graph $\tilde{A}=(\tilde{a}_{ij})$ is called a description graph of $A$, if the following conditions hold for any integer $t\ge 0$:

For any two pairs of vertices $u_1,v_1$ and $u_2,v_2$ of $\tilde{A}$, $\tilde{a}_{u_1v_1}=\tilde{a}_{u_2v_2}$ if and only if for any sort of length $t$, the number of walks from $u_1$ to $v_1 $  in graph $A$ equals to the number of walks, of the same sort, from $u_2$ to $v_2$ in $A$.%
\end{definition}
Notice that a walk of length $0$ can only occur from a vertex $u$ to itself, and default only $1$ such walk to any vertex. It is not hard to see that any two description graphs to a graph $A$ will be equivalent as in Definition \ref{def:equiva}.

Description graphs are defined with respect to unordered sorts of walks, while WL process counts the ordered sorts of walk (of length $2$).  The description graphs will be formally weaker in distinguishing vertices and edges than WL process and its variants. For example, a walk of length $2$ with ordered sort $\langle 3,7\rangle$ is distinguished from the walk of sort $\langle 7,3\rangle$ in WL process, but they are counted as the same sort of walk in our case.

One may easily see that any entry $a^{(k)}_{uv}$ in $A^k:=(a^{(k)}_{ij})$ is a multiset of walks in length $k$ for all possible sorts between vertices $u$ and $v$ in a given graph $A$. If the amount of walks, of length $k$, of some sort between any pair of vertices $u'$ and $v'$ is not identical to that between $u$ and $v$, then $a^{(k)}_{u'v'}\ne a^{(k)}_{uv}$.

That is the handle to our processes to description graphs introduced in subsequent subsections.  They are, arguably,  more natural in partition of vertices and edges with regards to graph isomorphism. We will start with a process convenient in applications.
\subsection{Description Graphs Based on  Matrix Power}\label{sec:mp}
Given a graph $A=(a_{ij})$, let  $A^k:=(a^{(k)}_{ij})$ be the $k$-th power of matrix $A$. We stress that the entries in $A$ are variables and conform commutative rule in multiplications. That guarantees the ordering of walk sorts are not counted. This fact allows the following formulation of description graphs.

 For a graph $A$ over $[n]$ and any positive integer $t$, we set $\ld$ matrix with $\ld\notin A$, as
\begin{align}
 \Gamma(A,t)&:=\sum_{k=0}^{t}\ld^kA^k=(\gamma^{(t)}_{ij}),& S(\ld):=\sum_{k=0}^\infty\ld^kA^k=(s_{ij})\,.\label{eq:gamma}%
\end{align}
 Where $A^0:=I$ is identity matrix of order $n$. For more about the walk generation function related to $S(\ld)$, please refer to \cite{Godsil93}.

Careful inspections will find that, from the definition,  a description graph $\tilde{A}:=(\tilde{a}_{ij})$ of graph $A$ indicates that $\tilde{a}_{ij}=\tilde{a}_{uv}$ if and only if $s_{ij}(\ld)=s_{uv}(\ld)$ as polynomials in $\ld$.

Since the characteristic polynomial $\Delta(\ld):=\det(\ld I-A)$ for $\ld\notin A$ satisfies that $\Delta(A)=\mathbf{0}$ as zero matrix, where $I$ is the identity matrix, we assume  $$\bar{p}(\ld)=\ld^m-p_{m-1}\ld^{m-1}-\cdots-p_1\ld-p_0$$
 as the minimum polynomial of symmetric matrix $A$ with $m\le n$, where $p_i\in\R$. It holds that $\bar{p}(A)=A^m-p_{m-1}A^{m-1}-\cdots-p_0I=\mathbf{0}$ as zero matrix.

 Each entry $\gamma^{(t)}_{ij}$ in $\Gamma(A,t)$ is a polynomial of degree $t$ in $\ld$. We have the following conclusion for $\Gamma(A,t):=(\gamma^{(t)}_{ij})$.

\begin{proposition}\label{prop:gamacishu}
  Let $A$ be a graph of order $n$. Using the notations as above, for any $t\ge m$, it holds that
 $\gamma^{(m-1)}_{ij}=\gamma^{(m-1)}_{rs}$ if and only if  $\gamma^{(t)}_{ij}=\gamma^{(t)}_{rs}$ for all $i,j,r,s\in[n]$.\end{proposition}
\begin{proof}
 By the definition, the entries of $\Gamma$ are polynomials in $\ld$. Two polynomials are equal iff the corresponding coefficients are equal.
That guarantees that, for $t>m-1$, $\gamma^{(t)}_{ij}=\gamma^{(t)}_{rs}$ implies  $\gamma^{(m-1)}_{ij}=\gamma^{(m-1)}_{rs}$ as polynomials in $\lambda$ for all $i,j,r,s\in [n]$.

On the other side,  if $\gamma^{(m-1)}_{ij}=\gamma^{(m-1)}_{rs}$, we will only show that $\gamma^{(m)}_{ij}=\gamma^{(m)}_{rs}$. For $t>m$, it is easy to show by induction.

Since $\bar{p}(A)=\mathbf{0}$ implies $A^m=p_{m-1}A^{m-1}+\cdots+p_1A+p_0I$, we have
\begin{align}
  \Gamma(A,m)&=\sum_{k=0}^{m-1}\ld^kA^k+\ld^mA^m\notag\\
  &=\sum_{k=0}^{m-1}\ld^kA^k+\ld^m\sum_{k=0}^{m-1}p_kA^k
  =\sum_{k=0}^{m-1}\left(\ld^k+p_k\ld^m\right)A^k.
\end{align}
Let  $\gamma^{(m-1)}_{ij}=\sum_{k=0}^{m-1}\alpha_k\ld^k=\gamma^{(m-1)}_{st}$. We then have
\begin{align}
\gamma^{(m)}_{ij}&=\sum_{k=0}^{m-1}\alpha_k(\ld^k+p_k\ld^m)=\sum_{k=0}^{m-1}\alpha_k\ld^k+\left(\sum_{k=0}^{m-1}\alpha_kp_k\right)\ld^m\notag{}\\
&=\gamma^{(m-1)}_{ij}+\left(\sum_{k=0}^{m-1}\alpha_kp_k\right)\ld^m
=\gamma^{(m-1)}_{rs}+\left(\sum_{k=0}^{m-1}\alpha_kp_k\right)\ld^m
=\gamma^{(m)}_{rs}.%
\end{align}
That completes the proof.%
\end{proof}

Let $\Gamma(A):=\Gamma(A,n-1)$. Proposition \ref{prop:gamacishu} assures that   $\Gamma(A,m-1)\approx\Gamma(A)$ due to $m\le n$.
For a graph $A$ of order $n$, after an equivalent variable substitution to $\Gamma(A)$, we then obtain a graph $\tilde{A}$ as a description graph to $A$.

\subsection{Description Graphs Based on Spectral Decomposition}\label{sec:bstv}

We  assume the given graph $A$ of order $n$ is a real matrix ONLY in this subsection, which means the label on each edge is a real number. In this setting, the real symmetric matrix $A$ has $n$ real eigenvalues.

Let the spectra of $A$ be formed as $\spp(G)=\spp(A)=(\mu_1^{m_1},\mu_2^{m_2},\ldots,\mu_d^{m_d})$, where  $\mu_1,\cdots,\mu_d$ are all distinct eigenvalues of $A$, and $m_x$ be algebraic multiplicity of  $\mu_x$. It is well known (cf. e.g. \cite{Godsil93}, \cite{CvetkovicRoSl97}) that there will be $d$
real matrices $\E_x=\big(e^{[x]}_{i\,j}\big)$, ($x\in[d]$) satisfying the followings:
\begin{enumerate}
  \item $\E_x\,\E_y=\delta_{xy}\,\E_x$. Where $\delta_{xy}$ is Kronecker symbol.
  \item $A\,\E_x=\mu_x\,\E_x$.
  \item $A=\sum_{x=1}^d\mu_x\,\E_x$\,.%
\end{enumerate}
That gives, for any $k\ge 0$,  $A^k:=(a_{i\,j}^{(k)})$ and $\E_x:=(e_{i\,j}^{[x]})$,
\begin{align}
 A^k&=\sum_{x=1}^d\mu_x^k\,\E_x\,,&a_{i\,j}^{(k)}&:=\sum_{x=1}^d\mu_x^k\, e_{i\,j}^{[x]}\,.\label{eq:idemDec}%
\end{align}

We hence can conclude the following.
 \begin{theorem}\label{thm:idemgx}
  For a real symmetric matrix  $A=(a_{i\,j})_{n\times n}$, let $\Gamma(A)=(\gamma_{i\,j})$ as defined previously, it holds that $\gamma_{i\,j}=\gamma_{rs}$ if and only if $e^{[x]}_{i j}=e^{[x]}_{rs}$ for each $x\in[\,d\,]$,  $\mu_x\ne0$, $i,j,r,s\in[n]$.%
\end{theorem}
\begin{proof}
We have, by \eqref{eq:idemDec},
\begin{align}
 \Gamma(A)=\sum_{k=0}^{n-1} \ld^k A^k=\sum_{k=0}^{n-1} \ld^k\sum_{x=1}^d\mu_x^k\, \E_x=\sum_{x=1}^d\left(\sum_{k=0}^{n-1}\mu_x^k\, \ld^k\right)\E_x=\sum_{x=1}^d f(x,\ld)\,\E_x\,.
\end{align}
 Where $f(x,\ld):=\sum_{k=0}^{n-1}\mu_x^k\, \ld^k,\,\forall x\in[\,d\,]$. Hence $\gamma_{ij}=\sum_{x=1}^d f(x,\ld)\,e^{[x]}_{i\,j}\,.$
Then,  $ \gamma_{ij}=\gamma_{rs}$ if, and only if
\begin{align}
 \sum_{x=1}^d f(x,\ld)\,e^{[x]}_{ij}&=\sum_{x=1}^d f(x,\ld)\,e^{[x]}_{rs}\notag\\
 &\Longleftrightarrow\quad  \sum_{x=1}^d f(x,\ld)\,\beta_{x}=0 \quad\Longleftrightarrow\quad   \sum_{k=0}^{n-1}\left(\sum_{x=1}^d \mu_x^k\,\beta_{x}\right)\ld^k=0.
\end{align}
 Where $\beta_x:=e^{[x]}_{ij}-e^{[x]}_{rs}$. We thus obtain  equations system, for all $0\le k\le n-1$,
$$\sum_{x=1}^d \mu_x^k\,\beta_{x}=0.$$
Since all $\mu_x$ are distinct, for all nonzero $\mu_x$,  it holds $\beta_x=0$ for all $x\in[\,d\,]$ (according to Vandermonde determinant). That concludes for all $x\in[\,d\,]$,  if  $\mu_x\ne0$ then  $e^{[x]}_{i\,j}=e^{[x]}_{rs}$.
\end{proof}

From the conclusion above we obtain
$$\Gamma(A)\approx\sum_{x\in[d],\mu_x\ne 0}\ld^x\,\E_x.$$
That gives our second process to compute the description graph for a graph $A$.

This process reveals the relation of labels of description graphs to spectral decompositions of graphs. That might be compared with results in \cite{Furer95}, where the relation of labels of edges to spectral properties of graphs was investigated.


\subsection{Description Graphs Based on Adjoint Matrices}\label{sec:adj}

For any graph $A=(a_{ij})$, let $\ld I-A=(\delta_{ij}\,\ld-a_{ij})$ be the characteristic matrix of $A$ and $\Delta(\ld):=\det(\ld I-A)$ the characteristic polynomial of $A$ with $\ld\notin A$, where Kronecker $\delta_{ij}=1$ iff $i=j$ and $\delta_{ij}=0$ otherwise.

The matrix $\textrm{adj}(A):=(\bar{a}_{ij})$ is called the adjoint matrix of $A$ if $\bar{a}_{ij}(\ld)$ is the algebraic complement of $\delta_{ij} \ld-a_{ij}$ in $\Delta(\ld)$. It holds that  $(\ld I-A)\cdot\mathrm{adj}(A)=\Delta(\ld)\cdot I$ (cf. \cite{Gantmacher00}, pp82-83).

Proposition \ref{prop:gamacishu} tells that  $S(\ld)\approx\Gamma(A)$ for $S(\ld)=\sum_{i=0}^\infty\ld^i A^i$, where  $A^0=I$ is the identity matrix (\cite{Godsil93}.).  Since $I=S(\ld)(I-\ld A)$ and
$$S(\ld)=(I-\ld A)^{-1}\approx (\ld I-A)^{-1}=\Delta(\ld)\,\mathrm{adj}(A),$$
where $\Delta(\ld)$ is a polynomial in $\ld$, we then have
$$\tilde{A}\approx\Gamma(A)\approx S(\ld)\approx \textrm{adj}(A)\,.$$

That is the third process for computation of a description graph to $A$ by computing $\textrm{adj}(A)$, and followed by an equivalent variable substitution.

\section{Stabilizations and some Properties}\label{sec:st&ppro}
A description graph $\tilde{X}$ of $X$ is called \emph{stable graph} if it is equivalent to $X$. That is $\tilde{X}\approx X$.
For any graph $A$, the entries of $\Gamma(A)$  are polynomials in $\ld$, and the coefficients of $\ld$ are the entries of $A$. That means the number of distinct entries in $\Gamma(A)$ can not be fewer than that in $A$. We thus have $\dim(A)\le\dim(\tilde{A})$.

Given a graph $A_0:=A$ of order $n$, one may proceed recursively as follows: Evaluate and obtain description graph $A_{i+1}:=\tilde{A}_{i}$. In this way, a sequence of description graphs is obtained as:
\begin{align}
 A_0:=A,\,A_1:=\tilde{A}_0,\,A_2:=\tilde{A}_1,\,\ldots, A_k:=\tilde{A}_{k-1},\ldots\,,\label{eq:seq}%
 \end{align}
satisfying
 $$\dim(A_0)\le\dim(A_1)\le\dim(A_2)\le\cdots.$$
Since $\dim(X)\le n_1$ for any graph of order $n$ where $n_1:=n(n+1)/2$, the sequence in \eqref{eq:seq} will reach to an $A_t$ such that $\dim(A_t)=\dim(A_{t+1})$ and $t\le n_1$. In this case, one may easily verify that $A_t\approx A_{t+1}=\tilde{A}_t$ and hence $A_t$ is a stable graph. The graph $A_t$ is called the stable graph of  $A$, and denoted as $\hat{A}$.

Readers, who are familiar with the works of Mogan \cite{Morgan65}, R{\"u}ker and R{\"u}ker \cite{RuckerRu90,RuckerRu91}, will find that $\Gamma(A)$ is an extension and reform of approaches employed there. The total degree partitions posed by Tinhofer and Klin \cite{TinhoferKl99} is  the vertices partitions  of stable graphs.


All of the works mentioned above do not pay enough attention to the partitions to edges, and the information of edges is ignored. The stable graphs are hence not explicitly introduced there.


 As a ready example for stable graph, one may show that for any strongly regular graph $A$, it holds that $\tilde{A}\approx\hat{A}$. (We refer readers to, e.g., Brouwer and van Maldeghem \cite{Brouwerv22} for strongly regular graphs.)
The following results will be cited later in the context.
\begin{lemma}\label{lem:bst}
 For any graphs  $A$ and $B$ over $[n]$, we have the following properties.
  \begin{enumerate}
      \item $A^k\rightarrowtail \Gamma(A)\rightarrowtail \tilde{A}\rightarrowtail\hat{A}$ for all $k>0$.
    \item If $A\rightarrowtail B$ then $A^k\rightarrowtail B^k$. That implies  $\tilde{A}\rightarrowtail\tilde{B}$ and $\hat{A}\rightarrowtail\hat{B}$ in this case.
    \item $\Aut(A)=\Aut(\tilde{A})=\Aut(\hat{A})$.\end{enumerate}\end{lemma}
\begin{proof}
  Denote $A:=(a_{ij}), B:=(b_{ij}),A^k:=(a^{(k)}_{ij}),B^k:=(b_{ij}^{(k)})$ for $k>1$, these properties are shown separately as follows.
\begin{enumerate}
\item By definition of $\Gamma(A)$ and Proposition \ref{prop:gamacishu}, if $\gamma^{(n)}_{ij}=\gamma^{(n)}_{st}$ as polynomials in $\lambda$, then  $a^{(k)}_{ij}=a^{(k)}_{st}$ in $A^k$ for all $k\ge 0$. That implies $A^k\rightarrowtail\Gamma(A)\rightarrowtail\tilde{A}$. Specifically, we get $A\rightarrowtail \tilde{A}$. That concludes $\tilde{A}\rightarrowtail\hat{A}$ from Proposition \ref{prop:lehman} and induction.

\item From $A\rightarrowtail B $ we know that $b_{ij}=b_{st}$ implies $a_{ij}=a_{st}$ for all $i,j,s,t\in[n]$. We argue $A^2\rightarrowtail B^2$ by showing that  $b^{(2)}_{ij}=b^{(2)}_{st}$ implies  $a^{(2)}_{ij}=a^{(2)}_{st}$ for all $i,j,r,s\in[n]$. Notice that
    \begin{align}
      a^{(2)}_{ij}&=\sum_{\ell\in[n]}a_{i\ell}a_{\ell j},& a^{(2)}_{rs}&=\sum_{\ell\in[n]}a_{r\ell}a_{\ell s},& b^{(2)}_{ij}&=\sum_{\ell\in[n]}b_{i\ell}b_{\ell j},&       b^{(2)}_{rs}&=\sum_{\ell\in[n]}b_{r\ell}b_{\ell s}.
    \end{align}
    Since all  $a_{ij}$ and $b_{ij}$ are independent variables,  $b_{ij}^{(2)}=b_{rs}^{(2)}$ is hence equivalent to
    \begin{align}\label{eq:2equ}
    \mset{\mset{b_{i1},b_{1j}},\ldots,\mset{b_{in},b_{nj}}}
    \equiv\mset{\mset{b_{r1},b_{1s}},\ldots,\mset{b_{rn},b_{ns}}}\,
    \end{align}
  By $A\rightarrowtail B$, we have
  $$\mset{b_{i\ell_1},b_{\ell_1 j}}\equiv\mset{b_{r\ell_2},b_{\ell_2 s}}\quad\Longrightarrow\quad\mset{a_{i\ell_1},a_{\ell_1j}}\equiv\mset{a_{r\ell_2},a_{\ell_2s}}$$
for all $\ell_1,\ell_2\in[n]$. Equation \eqref{eq:2equ} implies
    $$\mset{\mset{a_{i1},a_{1j}},\ldots,\mset{a_{in},a_{nj}}}
    \equiv\mset{\mset{a_{r1},a_{1s}},\ldots,\mset{a_{rn},a_{ns}}}\,.$$
   We then have $a_{ij}^{(2)}=a_{rs}^{(2)}$.

   Similarly, we show that $A^k\rightarrowtail B^k$ for any  $k>2$ by induction.
   If $b^{(k+1)}_{ij}=b^{(k+1)}_{rs}$, then
    $$\mset{\mset{b^{(k)}_{i1},b_{1j}},\ldots,\mset{b^{(k)}_{in},b_{nj}}}
    \equiv\mset{\mset{b^{(k)}_{r1},b_{1s}},\ldots,\mset{b^{(k)}_{rn},b_{ns}}}\,.$$
   Remember that all entries in $A$ and $B$ are invariants, by inductive assumption that $A\rightarrowtail B$ and $A^k\rightarrowtail B^k$ for $k>1$, we get
   $$\mset{\mset{a^{(k)}_{i1},a_{1j}},\ldots,\mset{a^{(k)}_{in},a_{nj}}}
    \equiv\mset{\mset{a^{(k)}_{r1},a_{1s}},\ldots,\mset{a^{(k)}_{rn},a_{ns}}}\,.$$
    That states $A^{k+1}\rightarrowtail B^{k+1}$.
   These together implies $\Gamma(A)\rightarrowtail \Gamma(B)$ by the definition of $\Gamma$. That claims  $\tilde{A}\rightarrowtail\tilde{B}$. That in turn implies $\hat{A}\rightarrowtail\hat{B}$ by induction.

 \item We will only show $\Aut(A)=\Aut(\tilde{A})$, which will give $\Aut(A)=\Aut(\hat{A})$ by induction.

     Let $P$ be a permutation matrix such that $PAP^\T=A$. We then have, by definitions,
  \begin{align}
   P\Gamma(A)P^\T&=P\left(\sum_{k=0}^{n-1}\ld^kA^k\right)P^\T=\sum_{k=0}^{n-1}\ld^kPA^kP^\T=\sum_{k=0}^{n-1}\ld^kA^k=\Gamma(A).\label{eq:no1}%
  \end{align}
   It means $P\tilde{A}P^\T=\tilde{A}$, since $\Gamma(A)\approx\tilde{A}$.

  On the other hand, for a permutation matrix $P$ such that $P\tilde{A}P^\T=\tilde{A}$, it is equivalent to   $P\Gamma(A)P^\T=\Gamma(A)$ by the definition of description graphs. By the definition of $\Gamma(A)$ we have $PA^kP^\T=A^k$ for $k=0,1,\ldots,n-1$. Specifically, $PAP^\T=A$.%
  \end{enumerate}
  That completes the proof of the lemma.
\end{proof}

We say that a graph recognizes vertices (respectively, edges) if the multiset of labels of vertices (respectively, edges) does not  overlap with the remaining labels in the graph.

Lemma \ref{lem:bst} tells that if some vertices  or edges are recognized at some steps in computing stable graphs, they are recognized during the evaluation of stable graphs too. With this property one may get Recognizable Property $\mathbf{R}$ in Proposition \ref{prop:xzhRU} at once.
With this terminology, we thus have the following corollary of Lemma \ref{lem:bst}.
\begin{corollary}\label{coro:power}
  If a graph $A$ of order $n$ recognizes vertices, then  $\tilde{A}\approx A^n$.%
\end{corollary}
\begin{proof}
  Let $A:=(a_{ij})$ and $A^k:=(a^{(k)}_{ij})$ for $k>1$. Since it recognizes vertices, we have $a_{uu}\ne a_{rs}$ for all $r,s,u\in[n]$ and $r\ne s$.
  We show that $A\rightarrowtail A^2$ in this case. Notice that
  $$a^{(2)}_{uv}=\sum_{t\in[n]}a_{ut}a_{t v}=a_{uu}a_{uv}+a_{uv}a_{vv}+\sum_{t\in[n]\backslash\{u,v\}}a_{ut}a_{t v}\,.$$

Since $A$ recognizes vertices and all entries in $A$ are variables,  it should hold
  $a_{uu}a_{uv}+a_{uv}a_{vv}=a_{rr}a_{rs}+a_{rs}a_{ss}\,,$ if $a^{(2)}_{uv}=a^{(2)}_{rs}$ for some  $u,v,r,s\in[n]$.

  That shows $a_{uv}=a_{rs}$. That is $A\rightarrowtail A^2$. Similarly, we obtain $A^k\rightarrowtail A^{k+1}$ by induction.

  That gives $A\rightarrowtail A^2\rightarrowtail\cdots\rightarrowtail A^n$. That in turn implies $$\tilde{A}\approx\Gamma(A)=\sum_{k=0}^{n-1}\ld^k A^k\approx A^n$$
   by Proposition \ref{prop:gamacishu}.
\end{proof}

We present here several properties of stable graphs, and more properties are discussed in Section \ref{sec:wdtxz}.
The following properties are easy to be obtained the conclusions above.

\begin{proposition}\label{prop:xzhRU}
The stable graphs possess the following properties.
\begin{itemize}
  \item \emph{Recognizable Property~$\mathbf{R}$.} If some kinds of vertices or edges are recognized in the description graph, then they are also recognized in the stable graph.
  \item \emph{Undistinguishable Property $\mathbf{U}$.} For any stable graph $X=(x_{ij})$ and matrices $X^\ell:=(x^{(\ell)}_{ij})$ for $\ell>1$, it holds that $x_{ij}=x_{rs}$ implies $x^{(\ell)}_{ij}=x^{(\ell)}_{rs}$ for all $i,j,r,s\in[n]$ and $\ell>1$.\end{itemize}
 \end{proposition}

 As applications of Proposition \ref{prop:xzhRU}, we show the followings.
\begin{proposition}\label{prop:shibie}
For a graph $A=(a_{ij})$ of order $n$ and its stable graph $\hat{A}=(\hat{a}_{ij})$, it holds that
\begin{itemize}
  \item The stable graph $\hat{A}$ recognizes vertices. That is, $\hat{a}_{uu}\ne\hat{a}_{uv}$ if $u\ne v$ for any $u,v\in[n]$.
  \item The stable graph $\hat{A}$ recognizes edges of $A$. That is, if $a_{uv}\ne x_0$ and $a_{rs}=x_0$, then $\hat{a}_{uv}\ne\hat{a}_{rs}$  for any $u,v,r,s\in[n]$.%
  \end{itemize}
\end{proposition}
\begin{proof}
  By definitions, entries in $\Gamma(A)$ for any graph $A$ are polynomials in $\ld$. The constant term of an entry is $1$ if and only if it is on main diagonal of $\Gamma(A)$ since $A^0=I$. That means entries on diagonal of $\tilde{A}$ cannot be the same as the remaining entries. By recognizable property $\mathbf{R}$, we claim that stable graphs recognize vertices.

  Similarly, considering the coefficients of $\ld$ among entries of $\Gamma(A)$, we shall get that the labels in $\tilde{A}$ to unblank edges of $A$ do not overlap with those to blank edges of $A$ in $\tilde{A}$, and hence in $\hat{A}$.

\end{proof}

We note that a graph $\hat{A}$ recognizes vertices does not mean that all vertices are necessarily with the same label.

Proposition \ref{prop:shibie} tells that in $\hat{A}$, the set of labels to edges of $A$ will not overlap with the set of labels to blank edges of $A$. That allows us to change all labels in $\hat{A}$ to blank edges of $A$ into, say, $x_0$ and to obtain a graph $B$. In this case, if any two entries in $B$ are not identical, the corresponding entries cannot be identical in $\hat{A}$. That means $B\rightarrowtail\hat{A}$. This trick will be frequently used later in the contexts. We will show,  in the next section, that the stable graphs can also be obtained in a succinct way via matrix power.

The following conclusion reveals the essential property of stable graphs and enhances the Undistinguishable Property $\mathbf{U}$.

\begin{theorem}\label{thm:equivsg}
  For any graph $X$ of order $n$ that recognizes vertices, the followings are equivalent.
\begin{enumerate}
  \item Graph $X$ is a stable graph.
  \item $X\approx X^2$.
  \item $X\approx X^\ell$ for all $\ell>1$.%
\end{enumerate}
\end{theorem}
\begin{proof}
  The proof will rely on the fact that all entries of $X$ are independent variables. Denote  $X:=(x_{ij})$ and $X^\ell:=(x^{(\ell)}_{ij})$, ($\ell>1$).

  If $X$ is a stable graph, it holds that $X\approx\tilde{X}\approx \Gamma(X)$. For all $u,v,r,s\in[n]$ and $\ell>1$, the equation  $x_{uv}=x_{rs}$ then implies  $x^{(\ell)}_{uv}=x^{(\ell)}_{rs}$   by the definition of $\Gamma(X)$ and Proposition \ref{prop:gamacishu}. That is equivalent to $X^\ell\rightarrowtail X$ for all $\ell>1$.

  We now show $X\rightarrowtail X^2$. Assume that $x^{(2)}_{uv}=x_{rs}^{(2)}$ in $X^2$.
 Since $X$ recognizes vertices, we have $x_{uu}x_{uv}+x_{uv}x_{vv}=x_{rr}x_{rs}+x_{rs}x_{ss}$. That further shows $x_{uv}=x_{rs}$. That is $X\rightarrowtail X^2$ and hence $X^2\approx X$.

Assume that $X\approx X^2$, to show $X\approx X^\ell$ ($\ell>1$) by induction. We now have $X\approx X^k$ and $x^{(k+1)}_{uv}=\sum_{r=1}^nx^{(k)}_{ur}x_{rv}$ in $X^{(k+1)}$. Since entries of $X$ are all variables and $x^{(k)}_{ur}$ is just a summation of multiplications of variables, and $X\approx X^k$ means that $x_{uv}=x_{u'v'}$ if and only if $x^{(k)}_{uv}=x^{(k)}_{u'v'}$ for all $u,v,u',v'\in[n]$. That gives
\begin{align}%
 x^{(k+1)}_{uv}=x^{(k+1)}_{u'v'}\quad&\Longleftrightarrow\quad \sum_{r=1}^nx^{(k)}_{ur}x_{rv}= \sum_{r=1}^nx^{(k)}_{u'r}x_{rv'}\notag\\
 &\Longleftrightarrow\quad \sum_{r=1}^nx_{ur}x_{rv}= \sum_{r=1}^nx_{u'r}x_{rv'}\,.
\end{align}
That is, for all $u,v,u',v'\in[n]$,
\begin{align}
 x^{(k+1)}_{uv}= x^{(k+1)}_{u'v'}\quad&\Longleftrightarrow\quad x_{uv}^{(2)}=x_{u'v'}^{(2)}\quad\Longleftrightarrow\quad x_{uv}=x_{u'v'}\,.\label{ff}%
\end{align}
This proves $X\approx X^{k+1}$.

We have shown that if $X$ is a stable graph, then both $X\approx X^2$ and $X\approx X^\ell$ ($\ell>1$) hold and are equivalent. Given that $X\approx X^\ell$ for all $\ell>1$, the stability of $X$  is straightforward by the definition of $\Gamma(X)$.

That finishes the proof.
\end{proof}

Theorem \ref{thm:equivsg} can be restated as follows.
\begin{corollary}\label{coro:sort}
 In  any stable graph $X=(m_{ij})$ of order $n$, for any vertices $u,v,u',v'$ and any $t>0$, the followings hold.

   $m_{uv}=m_{u'v'}$ if and only if the amounts of walks (and paths) of any given sort, of length 2, between $u,v$ and $u',v'$ respectively, are identical;

  $m_{uv}=m_{u'v'}$  if and only if the amounts of walks (and paths) of any given sort, of length $t$, between $u,v$ and $u',v'$ respectively, are identical.%
\end{corollary}

\section{Matrix-Power Stabilization}\label{sec:mpr}

We will propose matrix-power stabilization process to evaluate stable graphs. For $k=2$, it is called Square and Substitution (SaS) process which is similar to WL process but more succinct. The stable graphs obtained by SaS process is proven to be equivalent to stable graphs obtained by WL process in partition vertices.

\subsection{Matrix-Power Stabilization to Stable Graphs}\label{sec:mlm}
In the proof of Corollary \ref{coro:power}, we see that $A\rightarrowtail A^2\rightarrowtail\cdots\rightarrowtail A^k\rightarrowtail\cdots\,$ for any graph  $A$ that recognizes vertices. That allows us to pose a matrix-power $k$-stabilization process to stable graph $\hat{A}$ for any $k>1$ as follows.\ep

\noindent\underline{Matrix-power $k$-stabilization}: Given a graph $A$, perform an equivalent variable substitution to the diagonal entries of $A$ to obtain $A_1$ such that $A_1$ recognizes vertices. Then proceed recursively to evaluate  and perform an equivalent variable substitution to $A^k_i$ to obtain graph $A_{i+1}$ until a stable graph is  obtained.\ep

When $k=2$, the stabilization process looks like
\begin{align}
 A\rightarrowtail A_1,\quad A_2\approx A_1^2,\quad A_3\approx A_2^2,\quad\ldots,\quad A_{i+1}\approx A_i^2,\quad\ldots\,.\label{eq:wlseq}%
\end{align}

We refer to this process as Square-and-Substitution (SaS, for short). The conclusions in the following theorem are valid for matrix-power $k$-stabilization with any $k>2$.

\begin{theorem}\label{thm:equa}
For any graph $A$ of order $n$ and graph $A_i$ in \eqref{eq:wlseq}, the followings are true.
  \begin{itemize}
  \item Each graph $A_i$ recognizes vertices.
  \item We have $A_i\rightarrowtail A_{i+1}$ and $\dim(A_i)\le\dim(A_{i+1})$ for each integer $i>0$.
  \item If $\dim(A_t)=\dim(A_{t+1})$ for some $t\in\Z^+$, then $A_t\approx\hat{A}$.%
  \end{itemize}
\end{theorem}
\begin{proof}
 The first step in SaS process is  an equivalent substitution (if necessary) to $A$ to obtain graph $A_1$ such that $A_1$ recognizes the vertices. The proof in Theorem \ref{thm:equivsg} indicates the entries in the diagonal of $A_1^2$ will not appear at any other undiagonal place in $A_1^2$. That is, $A_2$ recognizes vertices. Similarly, $A_i$ recognizes vertices for any $i>1$.

 Since $A_{i+1}\approx A_{i}^2$, one can easily obtain $A_{i}\rightarrowtail A_{i+1}$ with essentially the same argument as that in the proof of Theorem \ref{thm:equivsg}.

To show the last conclusion, we notice that $A_1$ is a result obtained by an equivalent variable substitution to the diagonal labels of $A$ at the first step of SaS process. That tells $A\rightarrowtail A_1\approx I+\ld A$ for $\ld\notin A$, and hence
    $$A\rightarrowtail A_1\approx I+\ld A\rightarrowtail \Gamma(A)\approx\tilde{A}$$
 by Lemma \ref{lem:bst}. That shows $\hat{A}\rightarrowtail\hat{A}_1\rightarrowtail\hat{\tilde{A}}$. That means $\hat{A}\approx\hat{A}_1$ since $\hat{\tilde{A}}\approx\hat{A}$.

From the definition of description graphs and Lemma \ref{lem:bst}, the followings hold
 \begin{align}
A\rightarrowtail A_1\rightarrowtail\tilde{A},\quad
A_2\approx A_1^2 \rightarrowtail\tilde{\tilde{A}},\quad A_3\approx A_2^2 \rightarrowtail\tilde{\tilde{\tilde{A}}},\quad\ldots\,.\label{eq:wlequa}%
\end{align}
Together with the result above, we thus get
$$A\rightarrowtail A_1\rightarrowtail A_2\rightarrowtail\cdots\rightarrowtail A_t\rightarrowtail \cdots\rightarrowtail\hat{A}\,.$$
It shows $\hat{A}_t\approx\hat{A}$ for any $t>0$.

For $\dim(A_t)=\dim(A_{t+1})$, however, it holds that $\hat{A}_t\approx A_t$ by Theorem \ref{thm:equivsg}. We have shown $A_t\approx \hat{A_t}\approx\hat{A}$.

That ends the proof.
\end{proof}

The result states, one may employ SaS process, rather than description graphs, to compute the stable graph to any graph. From $A\rightarrowtail A_1\rightarrowtail A^k\rightarrowtail \hat{A}$, we know matrix-power stabilization produce the equivalent stable graph to any graph $A$.
\begin{corollary}
  The stable graph obtained by matrix-power stabilizations and that obtained by stabilizations of description graphs to a graph are equivalent.%
\end{corollary}

We note that the square-and-substitution process was already mentioned by F{\"{u}}rer in \cite{Furer01} (page 323 in \cite{Furer01}), although there it was not explored in details.

\subsection{Equivalence of SaS and WL in the Partition of Vertices}
The process SaS is formally the same as WL process. We will now confirm that two processes are equivalent in vertices partition. In fact, SaS process is a symmetrization of WL process.

For convenience in presentations, let us introduce unsymmetric  products. For any variables $x,y$, we denote $x\diamond y$ as the uncommutative multiplication such that $x\diamond y=x'\diamond y'$ iff $x=x'$ and $y=y'$. In this way, the commutative dot production and the uncommutative diamond product for two vectors $\alpha:=(x_1,\ldots,x_n)$ and $\beta:=(y_1,\ldots,y_n)$ are, respectively, defined as
\begin{align}
  \alpha\cdot \beta:=\sum_{k=1}^nx_ky_k,\quad\alpha\diamond\beta:=\sum_{k=1}^nx_k\diamond y_k\,.\label{eq:diamond}%
\end{align}

 Let $X$ and $Y$ be two graphs of order $n$, where $\alpha_k$ is the  $k$-th row vectors of $X$ and $\beta_k$ is the $k$'th column vector in $Y$. The matrices multiplication and diamond multiplication are defined as  $XY=(\alpha_i\cdot\beta_j)$ and $X\diamond Y:=(\alpha_i\diamond\beta_j)$, respectively.

A graph $X:=(x_{uv})$  respects \emph{converse equivalent}, if it holds that $x_{uv}=x_{rs}$ iff $x_{vu}=x_{sr}$ in $X$ (cf. \cite{LichterPoSc19}).

The following result reveals the relation of multiplication and diamond product.
\begin{proposition}\label{prop:sqdiomnd}
    Let $X$ be a graph of order $n$. The followings hold.
  \begin{enumerate}
    \item If $X$ respects converse equivalent, the matrix $X\diamond X$ respects converse equivalent.
    \item If $X$ is symmetric, we have $X^2\approx (X\diamond X)+(X\diamond X)^\T$. In fact, if graphs $A,B,C$ satisfy $A\approx X^2$, $B\approx X\diamond X$ and $C\approx B+B^\T$, then $A\approx C$.%
\end{enumerate}
\end{proposition}
\begin{proof}
  Set $X:=(x_{ij})$,   $X^2:=(x^{(2)}_{ij})$ and $X\diamond X:=(z_{ij})$.
  \begin{enumerate}
    \item For any $u,v,r,s\in [n]$, if $z_{uv}=z_{rs}$ then
    $z_{uv}=\sum_{k\in[n]}x_{uk}\diamond x_{kv}=\sum_{k\in[n]}x_{rk}\diamond x_{ks}=z_{rs}$. Equivalently, \begin{align}\label{eq:conv1}
      \mset{x_{uk}\diamond x_{kv}\mid k\in[n]}\equiv\mset{x_{rk}\diamond x_{ks}\mid k\in[n]}\,.
    \end{align}
    Since $X$ respects converse equivalent, it holds that $x_{kv}=x_{ks}$ iff $x_{vk}=x_{sk}$. We have by equation \eqref{eq:conv1} that
    \begin{align}\label{eq:conv2}
     \mset{x_{uk}\diamond x_{vk}\mid k\in[n]}\equiv\mset{x_{rk}\diamond x_{sk}\mid k\in[n]}\,.
    \end{align}
   With the same reason, we further have
    \begin{align}\label{eq:conv3}
     \mset{x_{ku}\diamond x_{vk}\mid k\in[n]}\equiv\mset{x_{kr}\diamond x_{sk}\mid k\in[n]}\,.
    \end{align}
    By the definition of diamond product, that shows  $\mset{x_{vk}\diamond x_{ku}\mid k\in[n]}\equiv\mset{x_{sk}\diamond x_{kr}\mid k\in[n]}$. Hence, $z_{vu}=z_{sr}$.
    \item For any $u,v,r,s\in[n]$, if $z_{uv}+z_{vu}=z_{rs}+z_{sr}$, e.g.,  $$\sum_{k\in[n]}(x_{uk}\diamond x_{kv}+x_{vk}\diamond x_{ku})=\sum_{k\in[n]}(x_{rk}\diamond x_{ks}+x_{sk}\diamond x_{kr})$$
         In multiset form, it is
         \begin{align}\label{eq:conv4}
          \mset{\mset{x_{uk}\diamond x_{kv},x_{vk}\diamond x_{ku}}\mid k\in[n]}\equiv\mset{\mset{x_{rk}\diamond x_{ks},x_{sk}\diamond x_{kr}}\mid k\in[n]}\,.
         \end{align}
        By symmetry of $X$, equation \eqref{eq:conv4} holds iff
        $\mset{x_{uk}x_{kv}\mid k\in[n]}\equiv\mset{x_{rk}x_{ks}\mid k\in[n]}$.
       That is $x_{uv}^{(2)}=x_{rs}^{(2)}$.
       We hence have $X^2\approx (X\diamond X)+(X\diamond X)^\T$ .%

       Let $A, B$ and $C$ be the graphs obtained by equivalent variable substitution to $X^2, X\diamond X$ and $B+B^\T$, respectively. We have $A\approx C$ by \eqref{eq:conv4}.%
  \end{enumerate}
\end{proof}

With these notations, the process of WL may be stated as follows. For any graph $A$ of order $n$, perform an equivalent variable substitution only to the diagonal entries of $A$ to obtain a graph $X_1$, such that $X_1$ recognizes vertices, just as does in SaS process. We then let $X_2$ be the graph obtained by an equivalent variable substitution to $X_1\diamond X_1$.  Proceeding this multiplication and substitution procedure recursively, we get a sequence as follows.
\begin{align}\label{eq:prodsub}
A\rightarrowtail X_1,\quad X_2\approx X_1\diamond X_1,\quad X_3\approx X_2\diamond X_2,\quad \ldots,\quad  X_{i+1}\approx X_i\diamond X_i,\quad\ldots\,.
\end{align}
As shown in \cite{Weisfeiler76}, $X_i$ may be a (partially or totally) oriented graph satisfying $X_i\rightarrowtail X_{i+1}$ and  $X_i\rightarrowtail X_{i+1}^\T$. Remember that $X^\T$ is the transpose of $X$ for any graph $X$.

 We thus know that $\dim(X_i)\le\dim(X_{i+1})$.  A graph $X_t$ with $\dim(X_t)=\dim(X_{t+1})$ in \eqref{eq:prodsub} is called the stable graph of $A$ obtained by WL process, denoted as $\mathtt{WL}(A)$.

By Proposition \ref{prop:sqdiomnd}, it is easy to see from the diamond product that the graph $X_i:=(x_{uv})$ respects \emph{converse equivalent}.

In sequences \eqref{eq:wlseq} and \eqref{eq:prodsub}, it holds that $A_1\approx X_1$ and all $A_i$, $X_i$ recognize vertices for $i>0$.
One can show that $A_i\rightarrowtail X_i$ and $\Diag(A_i)\approx\Diag(X_i)$ for each $i>0$ by Proposition \ref{prop:sqdiomnd} and induction.
Where $\Diag(A_t)$ and $\Diag(X_t)$  denote the diagonals of $A_i$ and $X_i$, respectively. For stable graph $A_t$ obtained  by SaS and stable graph $X_t$ obtained by WL process, we have $\Diag(\hat{A})\approx\Diag(A_t)\approx\Diag(X_t)\approx\Diag(\WL(A))$.

The minimum number of $t$ such that $X_t\approx\mathtt{WL}(A)$ is usually called the iteration number of WL process (similarly to SaS process).
It is claimed by Lehman that $X_t\approx X_t^2$ and $X_t\approx X_t^\T$ for a stable graph $X_t$ (cf. 9.1 on page 17 in \cite{Weisfeiler76}).
 That means $A_t\approx A_t^2$ in this case since $A_t\rightarrowtail X_t$. That indicates whenever $X_t$ is stable graph in \eqref{eq:prodsub}, then  $A_t$ is a stable graph in \eqref{eq:wlseq} by Theorem \ref{thm:equivsg}. We thus have the followings.
 \begin{theorem}\label{thm:saswl} The followings hold.
 \begin{itemize}
   \item The iteration number of SaS process is the same as the iteration number of WL process.
   \item SaS process  and WL process have the same capacity in the partition of vertices to a graph.%
   \item The iteration number of SaS process is in $O(n\log n)$.%
 \end{itemize}
 \end{theorem}
The last conclusion comes from \cite{LichterPoSc19}, where the upper bound of iteration number of WL process is given as $O(n\log n)$.

It is not difficult to gain the similar conclusions as in Theorem \ref{thm:saswl} for matrix-power $k$-stabilization in \eqref{eq:wlequa} and $k$-walk-refinement process posed in  \cite{LichterPoSc19}.

From Theorem \ref{thm:saswl} the stable graphs obtained by the processes proposed in this work are equivalent in the partitions of vertices to those obtained by WL process.

\section{Some Properties of Stable Graphs}
\label{sec:wdtxz}

We will propose the notion of a strongly equitable partition to a graph at first, and then show the vertex partition of a stable graph is indeed a strongly equitable one. From now on, a stable graph refers the stable graph obtained by SaS process without explicit explanation. Whenever a stable graph obtained by WL process is used, we will mention it clearly.

A partition of a set $[n]$ is a collection of disjoint nonempty subsets of $[n]$ that as a whole exactly covers $[n]$. A vertex partition of a labeled graph $A$ over $[n]$ is a partition $\mathcal{C}:=(C_1,C_2,\ldots,C_p)$ of $[n]$ such that all the vertices that possess the same label are in the same subset.  Each subset $C_u$ is called a cell of $\mathcal{C}$.

A vertex partition of $A$ is equitable if the set of labels on all edges between a vertex $i$ in a cell $C_u$ with all vertices from a cell $C_v$ is independent of the vertex $i$ from $C_u$. A strongly equitable partition  requires further that the labels between different pairs of cells are completely different.
The strongly equitable partition of a graph is an enhancement of the equitable partition appeared in the literatures. It is formally defined as follows.

\begin{definition}[Strongly equitable partitions]
 Let $\mathcal{C}=(C_1,C_2,\ldots,C_p)$ be a vertex partition to a graph $A$ of order $n$. Assume $C_u=\{u_1,\ldots, u_r\}$, $C_v=\{v_1,\ldots, v_s\}$, and $C_w=\{w_1,\ldots,w_t\}$ for  $u,v,w\in[p]$.
\begin{itemize}
  \item If for all $u,v\in[p]$, it holds that $\sum_{k=1}^sa_{u_iv_k}=\sum_{k=1}^sa_{u_jv_k}$ and $\sum_{k=1}^ra_{u_kv_i}=\sum_{k=1}^ra_{u_kv_j}$. Equivalently,
      $$\mset{a_{u_iv_k}\mid k\in[s]}\equiv\mset{a_{u_jv_k}\mid k\in[s]}\quad \text{and}\quad \mset{a_{u_kv_i}\mid k\in[r]}\equiv\mset{a_{u_kv_j}\mid k\in[r]},$$
      for all $u_i,u_j\in C_u, v_i,v_j\in C_v$, then $\mathcal{C}$ is said to be \emph{an equitable partition}.
  \item If $\C$ is an equitable partition of $A$ and moreover,  for all $u,v,w\in[p]$, $v\ne w$, it holds that
      $$\mset{a_{u_1v_k}\mid k\in[s]}\cap\mset{a_{u_1w_k}\mid k\in[t]}=\mset{a_{v_ku_1}\mid k\in[s]}\cap\mset{a_{w_ku_1}\mid k\in[t]}=\emptyset\,.$$
      then $\mathcal{C}$ is said to be \emph{a strongly equitable partition}.%
\end{itemize}
\end{definition}

The following conclusions are similar to those appeared in Lehman \cite{Weisfeiler76}, and are shown here only for stable graphs obtained by SaS process. They are also valid to the stable graphs obtained by WL process.

From now on in this section, all stable graphs are obtained by SaS process except explicitly mentioned in the context. The properties about stable graphs obtained by WL process are only used once in the proof of Lemma \ref{lem:tech}. In appendix, we give an example as illustration (cf. page \pageref{gX}.)
\begin{proposition}\label{prop:wdtxzhi}
  Let $X=(m_{ij})$ be a stable graph obtained by SaS process over $[n]$. It holds that
  \begin{enumerate}
  \item\label{prop:wdtxzhi2} The labels on two vertices are equal if and only if the two rows are equivalent. That is,\\  $m_{uu}=m_{vv}\Leftrightarrow\mset{m_{uk}\mid k\in[n] }\equiv\mset{m_{vk}\mid k\in[n]}$ for all $u,v\in[n]$.
    \item\label{prop:wdtxzhi4}  If $m_{su}=m_{tv}$,  then  $\mset{m_{uu},m_{ss}}\equiv\mset{m_{vv},m_{tt}}$.
        Specially,  $m_{su}=m_{sv}$ implies $m_{uu}=m_{vv}$  for all  $u,v,s,t\in[n]$.%
  \end{enumerate}
\end{proposition}
\begin{proof} We will show them one by one. Remember that all entries in $X$ are variables in $\Var$.

\begin{enumerate}
\item If $\mset{m_{uk}\mid k\in[n]}\equiv\mset{m_{vk}\mid k\in[n]}$, then  $m_{uu},m_{vv}\in \mset{m_{vk}\mid k\in[n]}$. The labels on vertices will only appear at the diagonal of $\hat{A}$ by Proposition \ref{prop:shibie}, that means $m_{uu}=m_{vv}$.

On the other side, if $m_{uu}=m_{vv}$, the Undistinguishable Property $\mathbf{U}$ tells that $m_{uu}^{(2)}=m_{vv}^{(2)}$ in $X^2:=(m^{(2)}_{ij})$. Since $X$ is symmetric, we then get
\begin{align}\label{eq:hxd}
 \sum_{k=1}^n m_{uk}^2=\sum_{k=1}^n m_{uk}m_{ku}=m^{(2)}_{uu}=m^{(2)}_{vv}=\sum_{k=1}^n m^2_{vk}\,.%
\end{align}
Equivalently,
$$\mset{m_{uk}^2\mid k\in[n]}\equiv\mset{ m^2_{vk}\mid k\in[n]}\,.$$
It forces $\mset{m_{uk}\mid k\in[n]}\equiv\mset{m_{vk}\mid k\in[n]}$, since  $m_{uk}, m_{vk}$ are all variables. We thus show that  $m_{uu}=m_{vv}$ iff the rows $u$ and $v$ are equivalent.

\item For all $u,v,s,t\in[n]$, $m_{su}=m_{tv}$ implies $m^{(2)}_{su}=m^{(2)}_{tv}$ in $X^2:=(m^{(2)}_{ij})$ again by property $\mathbf{U}$. That is
    $$\sum_{k=1}^nm_{sk}m_{ku}=m_{su}^{(2)}=m_{tv}^{(2)}=\sum_{k=1}^nm_{tk}m_{kv}.$$
  To write those terms that include vertices labels separately from above, it becomes
  $$m_{ss}m_{su}+m_{su}m_{uu}+\sum_{\stackrel{k\ne s,\,u}{k\in[n]}}m_{sk}m_{ku}=m_{tt}m_{tv}+m_{tv}m_{vv}+\sum_{\stackrel{k\ne t,\, v}{k\in[n]}}m_{tk}m_{kv}.$$
  Since vertices labels only appear on the diagonal of $X$ and all $m_{ij}$ are variables, we have
  $$m_{ss}m_{su}+m_{su}m_{uu}=m_{tt}m_{tv}+m_{tv}m_{vv}.$$
  Since $m_{su}=m_{tv}$, it must be $\mset{m_{uu},m_{ss}}\equiv\mset{m_{vv},m_{tt}}$\,.%
\end{enumerate}
That ends the proof. \end{proof}

The vertex partition of a stable graph $\hat{A}$ is called \emph{the stable partition} of $A$. Since all stable graphs for a graph are all equivalent, the stable partition for a graph is well-defined. The first conclusion in the following properties is very important for our explorations.

\begin{theorem}\label{thm:jhhf}  The vertex partition $\mathcal{C}$ of a stable graph $X$ obtained by SaS process has the following properties.
\begin{enumerate}
  \item The stable partition is a strongly equitable partition.
  \item If a cell $C=\{u\}$ is a singleton, then it holds that  $m_{uv_1}=\cdots=m_{uv_s}$ for any cell $C_v=\{v_1,\ldots, v_s\}$ in the partition.
  \item\label{coro:ysb2} For  two cells $C_u,C_v$, in graph $A$, the number of neighbors in $C_v$ of a vertex from  $C_u$ is independent of the vertex chosen from $C_u$. That is, the stable partition induces an equitable partition of $A$.%
\end{enumerate}
\end{theorem}
\begin{proof} Let the stable graph $X:=(m_{ij})$ be over $[n]$. To show the stable partition is a strongly equitable partition, we notice that given a cell $C$ from $\mathcal{C}$, it should be  $m_{uu}=m_{vv}$ for all $u,v\in C$. For all vertices $s,t$ from different cells, it holds that  $m_{ss}\ne m_{tt}$. It implies $\mset{m_{uu},m_{ss}}\not\equiv\mset{m_{vv},m_{tt}}$. That in turn implies $m_{us}\ne m_{vt}$ by Proposition \ref{prop:wdtxzhi}. This fact indicates that the edges of vertices in $C$ connecting vertices from two different cells   possess completely different labels.

Formally, let $C_u:=\{u_1,\ldots,u_r\},C_v:=\{v_1,\ldots,v_s\}$ and $C_w:=\{w_1,\ldots,w_t\}$ be cells in  $\mathcal{C}$, and  $C_v\ne C_w$. The argument as above claims  $m_{u_iv_j}\ne m_{u_kw_\ell}$ for all $i,k\in[r],j\in[s],\ell\in[t]$. That tells
\begin{align}\label{eq:jhk}
  \mset{m_{u_iv_j}\mid j\in[s]}\cap\mset{m_{u_k w_\ell}\mid \ell\in[t]}=\emptyset,
\end{align}
for all $i,k\in[r]$.

The Proposition \ref{prop:wdtxzhi} tells
$\mset{m_{u_ix}\mid x\in[n]}\equiv\mset{m_{u_kx}\mid x\in[n]}$
for all $i,k\in[r]$.  The facts together give
\begin{align}\label{eq:cellxd}
 \mset{m_{u_iv_j}\mid j\in[s]}\equiv\mset{m_{u_kv_j}\mid j\in[s]}.
\end{align}
From \eqref{eq:jhk} and \eqref{eq:cellxd}, the stable partition is a strongly equitable one.

The remaining two conclusions come from the results of equation \eqref{eq:cellxd}, Proposition \ref{prop:shibie} and the fact that stable graphs recognize edges. The details of arguments are omitted.

\end{proof}

It should be stressed that the conclusions in Theorem
\ref{thm:jhhf} might not be valid for graphs other than stable graphs. Hence, the main result concluded later in the context will essentially depend on the properties of stable graphs.

 It is intuitively true from the construction of the stable graphs that only if the vertices possess the same label, they might be  sent one to another by an automorphism of the graph. That means that a cell in a stable partition is a union of some orbits of the automorphism group of the graph. The next theorem makes it explicit.

Let  $Y$ be a stable graph of order $n$. We assume the vertices in each cell appear consecutively in $[n]$ and $Y$ is partitioned into $t$ blocks on diagonal, where $t$ is the number of cells in the stable partition of $Y$.  Formally, let $Y:=(Y_{ij})$ be a $t\times t$ block matrix such that each $Y_{ii}$ has the same entries on the diagonal and different entries for different $i$. With such an assumption, we claim the followings.
\begin{theorem}\label{thm:setwisest}
  For a stable graph $Y$ described as above and any real block  matrix $X$ with the same block partition as $Y$, if $XY=YX$, then $X$ will be a diagonal block matrix satisfying: for all $i,j\in[t]$,
  \begin{align}\label{eq:commute}
  X_{i,j}&=\mathbf{0}\quad \text{for $i\ne j$}\,,& X_{ii}Y_{ij}&=Y_{ij}X_{jj}\,.%
  \end{align}
\end{theorem}
\begin{proof} Let  $Y:=(Y_{ij}), X:=(X_{ij})$ both be $t\times t$ block matrices such that $Y_{ij}$ and $X_{ij}$ are $n_i\times n_j$ blocks. We have
\begin{align}
 XY&=(Z_{ij})\,,& Z_{ij}&=\sum_{k=1}^tX_{ik}Y_{kj}\,,&
 YX&=(\bar{Z}_{ij})\,,&\bar{Z}_{ij}&=\sum_{k=1}^tY_{ik}X_{kj}\,.
\end{align}
 $XY=YX$ implies  $Z_{ij}=\bar{Z}_{ij}$ for all $i,j\in [t]$. That is
\begin{align}\label{eq:setwise}
 \sum_{k=1}^tX_{ik}Y_{kj}=\sum_{k=1}^tY_{ik}X_{kj}.
\end{align}
The properties of stable graph $Y$ guarantee that $Y_{ij}\cap Y_{uv}=\emptyset$ for all $i,j,u,v\in[t]$ and $\mset{i,j}\cap\mset{u,v}=\emptyset$. That, together with the fact that all entries in $Y$ are independent variables in $\Var$, concludes $X_{ik}=\mathbf{0}$ and $X_{kj}=\mathbf{0}$ for  $i\ne j$ and $k\notin\{i,j\}$ from \eqref{eq:setwise} by the strongly equitable partition of $Y$. The equation \eqref{eq:setwise} hence becomes
\begin{align}
  X_{ii}Y_{ij}=Y_{ij}X_{jj}.\notag{}%
\end{align}
The equation also holds when $i=j$. That completes the proof of the theorem.  \end{proof}

 If $X$ is a permutation matrix in Theorem \ref{thm:setwisest}, the relation $XY=YX$ tells that $X$ is an automorphism of $Y$. The conclusion above tells that each cell of stable partition is setwise fixed by automorphisms. It seems that the last set of equations in \eqref{eq:commute} is frequently ignored or not applied in previous works.  The following conclusion is straightforward from Theorem \ref{thm:setwisest}.

\begin{corollary}
  If all cells in the stable partition of a graph are singletons, then the automorphism group of the graph is the unit group.%
\end{corollary}

\section{Binding Graphs and Graph-Isomorphism Completeness}\label{sec:bdt}

We propose a new class of graphs in which any graph  has unique correspondence up to isomorphism. For a graph $A$ of order $n$, we add a new vertex for each pair of vertices of $A$ and make it adjacent only to the two vertices. Totally  $n(n-1)/2$ new vertices will be added and the result graph will have $n_1:=n+n(n-1)/2=n(n+1)/2$ vertices. That graph is called a binding graph of $A$. Formal definition is given as follows.
\begin{definition}[Binding graphs]
  A simple graph $B$ over $[n_1]$ is called \emph{a binding graph} if for each pair of vertices $u,v\in [n]$, there exists unique $p\in [n+1..n_1]$ of degree $2$  adjacent to both $u$ and $v$. In latter case, we write $p:=u\dot{\wedge}v$ and say  vertex $p$ binds vertices $u$ and $v$. Both $(u,p)$ and $(v,p)$ are called binding edges.

  The subgraph $A$ of $B$ induced by vertices $[n]$ is called the \emph{basic graph} of $B$. We also say $B$ is a binding graph of $A$. The vertices in $[n]$ (and the edges between them) of $B$ are called basic vertices (and basic edges), while the vertices in $[n+1..n_1]$ the binding vertices.%
\end{definition}

In a binding graph,  all binding vertices are of degree $2$, and different binding vertices will bind different pairs of basic vertices. In any binding graph, a basic vertex will have a degree of more than $2$ if $n>3$, and that is true for a connected simple basic graph if $n>2$.

We  presume here in binding graphs that the first $n$ vertices are basic vertices, which is for the purpose of simplification in presentation and not necessary in general. Also, the binding graphs can be defined for general graphs rather than simple graphs here. However, the definition given here is enough for our purpose.

When constructing the binding graph from a simple graph $A$, it needs to assure that the label to binding edges is the same as the label to edges of $A$ in order to obtain a simple binding graph. Also, it is easy to find that the binding graph for a simple graph $A$ of order $n$ can be uniquely constructed up to the renaming of binding vertices, provided $n>2$.

For the simplicity of descriptions, we sometimes mention, in a risk of abusing, the binding vertices and binding edges in graphs (such as in $\WL(B)$) related to binding graph $B$  refer to the vertices (and edges) corresponding to binding vertices (and binding edges) of $B$. Similarly for basic vertices and basic edges.
We claim that the class of binding graphs is graph-isomorphism complete.

\begin{theorem}\label{thm:complete}
  For two connected simple graphs $A_1$ and $A_2$ of order $n$ $(n>2)$ with the same label on edges in two graphs, let $B_1$ and $B_2$ be their binding graphs, respectively, with the same label on edges as in $A_1$ and $A_2$. We have the following conclusions.
  \begin{enumerate}
    \item\label{thm:complete-1}    $A_1\cong A_2$ if and only if $B_1\cong B_2$.
    \item\label{thm:complete-0} $\Aut(A_1)\cong\Aut(B_1)$.
    \item\label{thm:complete-2}  Any basic vertex in $B_1$ will never share the same orbit with any binding vertex in $\Aut(B_1)$. We thus name an orbit consisting of basic vertices as a basic orbit, and that consisting  of binding vertices as a binding orbit.
    \item\label{thm:complete-3} A basic orbit in  $\Aut(B_1)$ is an orbit in $\Aut(A_1)$, and vice versa.%
  \end{enumerate}
\end{theorem}
\begin{proof}
Let $\sigma\,:\,A_1\rightarrow A_2$ be an isomorphism. Construct a map $\tau\,:\,B_1\rightarrow B_2$ as follows:
\begin{align}
x^\tau:=\left\{
  \begin{array}{ll}
    x^\sigma, & \hbox{if $x$ is a basic vertex in $B_1$,} \\
    u^\sigma\dot{\wedge}v^\sigma, & \hbox{if $x:=u\dot{\wedge}v$ in $B_1$.}%
  \end{array}
\right.%
\end{align}
Given an isomorphism $\sigma$, it is easy to argue that $\tau$ is a 1-1 map from vertices of $B_1$ to vertices of $B_2$ since the uniqueness of binding vertex for each pair of basic vertices. The map $\tau$ is  well defined.

Since both $B_1$ and $B_2$ are simple graphs, it is a routine  to verify that $\tau$ is now an isomorphism under the condition that $\sigma$ is an isomorphism.

On the other hand, let $\myfunc{\tau}{B_1}{B_2}$ be an isomorphism. Since basic graph $A_1$ of $B_1$ is a connected simple graph, the degree of any basic vertex in $B_1$ is thus not less than  $1+n-1=n>2$, while any binding vertex has degree $2$. That guarantees that  $\tau$  sends any basic vertex in $B_1$ to a basic vertex in $B_2$. The restriction of $\tau$ to basic vertices of $B_1$ is then  an isomorphism of $A_1$ to $A_2$.

That shows:  $A_1\cong A_2\Longleftrightarrow B_1\cong B_2$\,. The rest of conclusions in the Theorem are not hard to be obtained from the arguments as above with $n>2$. The details are omitted.

\end{proof}

For simplicity of descriptions, when  $p=u\dot{\wedge}v$, we say vertex $p$ (respectively, the binding edges $(p,u)$ and $(p,v)$) binds a (unblank) edge if $(u,v)$ is  a (unblank) basic edge in the binding graph.
We now give a technique lemma used later. Since the WL stable graphs will be used in the proof of our main result, we will discuss with WL stable graphs for the consistence. We also assume all the graphs are vertices recognizable (cf. the WL processes).
In this setting, a binding graph will be a simple graph with label to vertices as a variable different from $x_0$ and the variable to edges.

\begin{lemma}\label{lem:shbysb}
  Let $B:=(b_{ij})$ be a binding graph, $D:=(d_{ij})$ be the description graph of $B$ and $\WL(B):=(m_{ij})$ be the stable graph to $B$.  We have the following conclusions.
  \begin{enumerate}
    \item\label{lem:shbysb-1} For all basic vertices $u,v\in[n]$, $u\ne v$ and  $p:=u\dot{\wedge}v$ in $B$, the labels $m_{pu}$ and $m_{pv}$ in $\WL(B)$ witness the blankness (and non-blankness) of the edge $(u,v)$ in $B$.

       In general, a WL stable graph recognizes the binding edges that bind basic edges, and those which bind basic blank edges in the basic graph.
    \item\label{lem:shbysb-2} The stable graph $\WL(B)$ recognizes basic edges and binding edges of $B$, respectively.%
  \end{enumerate}
\end{lemma}
\begin{proof} For  $B=(b_{ij})$ over $[n]$ and $B\dd B:=(b^{(2)}_{ij})$,
  we assume, without loss of generality, the edges are labeled with $x_1$, and the vertices and blank edges with $x_0$ in the graph $B$.
  One knows that $b^{(2)}_{ij}=\sum_{k=1}^{n_1}b_{ik}\dd b_{kj}$, where $n_1=n(n+1)/2$.
 \begin{enumerate}
\item For all basic vertices  $u,v\in[n]$, $u\ne v$ and  $p:=u\dot{\wedge}v$ in $B$, we will show that if $(u,v)$ is an edge in $B$, the labels $m_{pu}$ and $m_{pv}$ in $\WL(B)$ cannot overlap with those to blank edges.

     With the assumption as above, $b_{pu}=b_{pv}=x_1, b_{pk}=x_0$ for all $k\in[n_1]\backslash\{u,v\}$. And also $B$ is a vertices recognizable graph. It thus holds that
\begin{align}
 b^{(2)}_{pu}&=\sum_{k\in[n_1]}b_{pk}\dd b_{ku}=
 b_{pu}\dd b_{uu}+b_{pv}\dd b_{vu}+\sum_{k\in[n_1]\backslash\{u,v\}}b_{pk}\dd b_{ku}\notag\\
 &=x_1\dd b_{uu}+x_1\dd b_{vu}+\Big(\sum_{k\in[n_1]\backslash\{u,v\}}x_0\dd b_{ku}\Big)\,.\label{eq:bdb}
\end{align}
If  $(u,v)$ is a blank edge in $B$, then  $b_{vu}=x_0$, and $b_{vu}=x_1$ otherwise. Since $B$ recognizes vertices, it means $b_{uu}\ne x_1$. We thus have by \eqref{eq:bdb} that  $b^{(2)}_{pu}$ has $x_1\dd x_1$ as a term iff $(u,v)$ is a basic (unblank) edge of $B$. That shows  the label $m_{pu}$ in $\WL(B)$ will be recognized if $(u,v)$ is a blank edge in $B$.  The same conclusion holds to label $m_{qv}$.

In fact, the arguments above show that the labels in $\WL(B)$ on binding edges (of $B$) that bind basic edges cannot overlap with those on binding edges that bind blank edges of $B$, due to the arbitrary of $u,v$ from $B$.

\item
Since stable graph distinguishes edges and blank edges of $B$, we now show that it also distinguishes basic edges and binding edges.

For any $u,v\in[n], u\ne v$, we have
  \begin{align}
   b^{(2)}_{uv}&=b_{uu}\dd b_{uv}+b_{uv}\dd b_{vv}+\sum_{k\in []\backslash\{u,v\}}b_{uk}\dd b_{kv}\,. \label{eq:yshb}
  \end{align}
There are two terms $b_{uu}\dd b_{uv}$ and $b_{uv}\dd b_{vv}$ in $b^{(2)}_{uv}$ that contain vertices labels $b_{uu},b_{vv}$, while there are only one term containing vertex label in $b^{(2)}_{wq}$ or $b^{(2)}_{qw}$ for any binding edge $(w,q)$ (cf. \eqref{eq:bdb}). That shows the label $m_{uv}$ to a basic edge $(u,v)$ cannot be the same as the label $m_{wq}$ to any binding edge $(w,q)$ or $(q,w)$.

That claims the stable graph $\WL(B)$ recognizes the basic edges and binding edges, respectively.%
\end{enumerate}
That is the end of proof.
\end{proof}
%

\section{Graph $\Phi$  from a Stable Binding Graph}\label{sec:phi}
We now turn to define the bipartite graph $\Phi$ as a subgraph of $\WL(B)$.
Let $A$ be a connected simple graph over $[n]$ with $n>2, n_1:=n(n+1)/2$,  $B:=(b_{ij})$ the binding graph of $A$ with vertices recognized and  $\WL(B):=(m_{ij})$ the stable graph of $B$. We will construct a connected graph $\Psi$ and a bipartite graph $\Phi$, respectively.

Intuitively, the graph $\Psi$ is a relabeling to vertices and unblank edges of graph $B$ with labels inherited from $\WL(B)$. Since $B$ is a connected graph, the graph $\Psi$ is a connected graph too.

The graph $\Phi$ is constructed by changing the labels to basic edges in $\Psi$ into $x_0$. That obtains $\Phi$ as a bipartite graph with binding edges and vertices labeled as in $\WL(B)$.

They are formally constructed as follows. Remember that the set of vertices in $A$ is $[n]$, and the first $n$ vertices in $B:=(b_{ij})$ and $\WL(B):=(m_{ij})$ are basic vertices.
 \begin{itemize}
   \item Constructing the graph $\Psi:=(\psi_{ij})$ : For all $i,j\in[n_1]$,
   \begin{align}
    \psi_{ij}:=\left\{
      \begin{array}{ll}
       x_0, & \hbox{if $i\ne j$ and $b_{ij}=x_0$,} \\
        m_{ij}, & \hbox{otherwise.}
      \end{array}
    \right.
   \end{align}

   \item Constructing the graph $\Phi:=(\phi_{ij})$ : For all  $i,j\in[n_1]$,
   \begin{align}
\phi_{ij}:=\begin{cases}
      x_0, & \hbox{if $i,j\in[n]$ and $i\ne j$,}\\
        \psi_{ij}, & \hbox{Otherwise.}
     \end{cases}
\end{align}
 \end{itemize}%
Since stable graphs recognize edges, blank edges, basic edges and binding edges, and stable partitions are strongly equitable partitions, it is easy to prove the following proposition.

\begin{proposition}\label{prop:equit}
 With the graphs defined as above, the followings hold.
 \begin{enumerate}
  \item $B\rightarrowtail \Psi\rightarrowtail\WL(B)$.
   \item $\Phi\rightarrowtail \Psi\rightarrowtail\WL(B)$.
   \item The vertex partitions of $\Phi$ and $\Psi$ are, respectively, equitable partitions.%
  \end{enumerate}
\end{proposition}

We will give the proof of following properties about graphs just constructed, in which the fact that there exists a binding vertex for each pair of basic vertices is  essentially used.

\begin{theorem}\label{thm:xphi}
  For a connected simple graph $A$ over $[n]$ with $n>2$, we have
$$\Aut(B)=\Aut(\WL(B))=\Aut(\WL(\Psi))=\Aut(\WL(\Phi))=\Aut(\Phi).$$
\end{theorem}
\begin{proof} The equations  $\Aut(B)=\Aut(\WL(B))$ and $\Aut(\Phi)=\Aut(\WL(\Phi))$ are from Lemma \ref{lem:bst}.    We prove $\Aut(\WL(B))=\Aut(\WL(\Psi))=\Aut(\WL(\Phi))$ by showing  $\WL(B)\approx \WL(\Psi)\approx \WL(\Phi)$.
\begin{itemize}
  \item Since $B\rightarrowtail \Psi\rightarrowtail\WL(B)$ by Proposition \ref{prop:equit}, it gives at once $\WL(B)\rightarrowtail\WL(\Psi)\rightarrowtail\WL(\WL(B))\approx\WL(B)$
by Lemma \ref{lem:bst}. It hence has $\WL(\Psi)\approx\WL(B)$. We thus proved  $\Aut(\WL(B))=\Aut(\WL(\Psi))$ by Proposition  \ref{prop:lehman}.

  \item
 First of all, since we have $\Phi\rightarrowtail \Psi\rightarrowtail\WL(B)$ in Proposition \ref{prop:equit}, by using Lemma \ref{lem:bst}, we get
\begin{align}\label{eq:phiinx}
 \WL(\Phi)\rightarrowtail \WL(B).
\end{align}

Next, let $\Phi\dd\Phi:=(\phi^{(2)}_{ij})$,  $\Phi:=(\phi_{ij})$. We have
$$
\phi_{uv}^{(2)}=
\phi_{up}\dd\phi_{pv}+\sum_{k\in[n_1]\backslash\{p\}}
\phi_{uk}\dd\phi_{kv}=m_{up}\dd m_{pv}+\sum_{k\in[n_1]\backslash\{p\}}\phi_{uk}\dd\phi_{kv}\,.
$$
By Lemma \ref{lem:shbysb}, the terms  $m_{pu}\dd m_{pv}$ in $\phi_{uv}^{(2)}$ will force $\phi_{uv}^{(2)}$ to distinguish a basic edge $(u,v)$ from a blank edge of $B$, since $m_{pu},m_{vp}\in \WL(B)$ from the definition of $\Phi$. (It is here the fact that there exists a binding vertex between any pair of basic vertices is implicitly used.)
Hence, the entry $\phi^{(2)}_{uv}$ in $\Phi\dd\Phi$ to an edge $(u,v)$ of $B$ cannot be identical to any entry $\phi^{(2)}_{u'v'}$ when $(u',v')$ is a blank edge in $B$. That implies, $B\rightarrowtail \Phi\dd\Phi$. Since $\Phi\dd\Phi\approx\Phi_2\rightarrowtail\WL(\Phi)$ by definition of $\WL(\Phi)$, we then have
\begin{align}\label{eq:xinphi}
  \WL(B)\rightarrowtail\WL(\Phi)\,.
\end{align}
Combining the equations \eqref{eq:phiinx} and \eqref{eq:xinphi} together, we obtain $\WL(B)\approx\hat{\Phi}$.%
\end{itemize}

That ends the proof.
\end{proof}

We have proposed the notion of binding graphs, and also constructed a bipartite graph $\Phi$ with respect to a connected simple graph $A$ over $[n]$ for $n>2$. Our purpose is to show that the vertex partition of $\Phi$ is the automorphism partition. For this purpose, we will relax graph $\Phi$ in next section.

\section{Graph $\Theta$ from Graph $\Phi$}\label{sec:theta}
 All notations and definitions in the last section are inherited in this section.
For a connected simple graph $A$ over $[n]$ with $n>2$, its binding graph $B$ with vertices recognizable, and a stable graph $\WL(B)$ of $B$, we have shown $\WL(B)\approx \WL(\Psi)\approx \WL(\Phi)$ in the proof of Theorem \ref{thm:xphi}. That fact allows us to make an assumption in the rest of contexts that
$\WL(B)= \WL(\Psi)=\WL(\Phi)$, obtained by equivalent variable substitution if necessary.

By the assumption above and Theorem \ref{thm:complete}, the vertex partition $\mathcal{C}$ of $\Phi$ consists of basic cells and binding cells and is an equitable partition by Proposition  \ref{prop:equit}.

Since the partition $\mathcal{C}$ is equitable, any binding cell $C_p$ connects only per (unblank) binding edges with at most two basic cells $C_u$ and $C_v$ due to the degree $2$ of any binding vertex, and the labels on binding edges between $C_p$ and $C_u$ are all identical.

Further, the strongly equitable property of stable partition of $\WL(\Phi)$ implies that the label on edges between $C_p$ and $C_u$ is different from the label on edges between $C_p$ and $C_v$ if $C_u\ne C_v$.

If a binding cell binds the same cell $C_u$, the labels on edges between $C_p$ and $C_u$ can either be all identical, or have only two distinct labels as stated in the last paragraph. These are summarized in the following proposition.

%
%

\begin{proposition}\label{prop:dbqd}
  Let
$u_1,u_2,v_1,v_2\in[n]$ be basic vertices in $\Phi$ with $u_1\ne u_2, v_1\ne v_2$, and $p=u_1\dot{\wedge}u_2$, $q=v_1\dot{\wedge}v_2$ be their binding vertices. The followings hold in graph $\Phi:=(\phi_{ij})$.
    \begin{enumerate}
    \item The followings are equivalent:
\begin{itemize}
    \item $\phi_{pp}=\phi_{qq}$;
    \item $\mset{\phi_{u_1p},\phi_{u_2p}}\equiv\mset{\phi_{v_1q},\phi_{v_2q}}$;

    \item $\mset{\phi_{u_1p},\phi_{u_2p}}\cap\mset{\phi_{v_1q},\phi_{v_2q}}\ne\emptyset$.%
  \end{itemize}
\item If $\phi_{pp}=\phi_{qq}$, then $\mset{\phi_{u_1u_1},\phi_{u_2u_2}}\equiv\mset{\phi_{v_1v_1},\phi_{v_2v_2}}$. \item If $\phi_{u_1u_1}\ne\phi_{u_2u_2}$, then  $\phi_{u_1p}\ne\phi_{u_2p}$\,.
     \item\label{prop:dbqd3}  If  $\phi_{pu_1}$ is given, then the labels $\phi_{u_1u_1},\phi_{pp},\phi_{pu_2},\phi_{u_2u_2}$ are all determined in $\Phi$.%
    \end{enumerate}
  \end{proposition}
\begin{proof}
 Note that $\Phi$ inherits all the labels on binding edges and those of vertices from $\WL(\Phi)$. We can easily obtain
  the first three conclusions by the strongly equitable partition of $\WL(\Phi)$. We stress that there exist only binding edges (and blank edges) in $\Phi$.

Let $\WL(\Phi):=(m_{ij})$. We now prove  the last conclusion.
 Since the labels on binding vertices are inherited from $\hat{\Phi}$, the label $\phi_{pu_1}=m_{pu_1}$ implicitly determines $m^{(2)}_{pu_1}$ in $\WL(\Phi)\dd\WL(\Phi):=(m^{(2)}_{ij})$. While $m_{pu_1}^{(2)}=\sum_{k\in[n_1]}m_{pk}\dd m_{ku_1}$, then $m_{u_1,u_1}, m_{u_2u_2}, m_{pu_2},m_{pp}$  are determined from $m_{pu_1}^{(2)}$ since the stable graph $\WL(\Phi)$ recognizes vertices and binding edges. However, we have $m_{u_1,u_1}=\phi_{u_1,u_1}, m_{u_2u_2}=\phi_{u_2u_2}, m_{pu_2}=\phi_{pu_2},m_{pp}=\phi_{pp}$ by the definition of $\Phi=(\phi_{ij})$.
 \end{proof}

It should be emphasized that the conclusions in Proposition \ref{prop:dbqd} heavily depend on the \emph{strongly} equitable partition of $\hat{\Phi}$, which may be faulty for unstable graphs.

The following conclusions can be obtained by the strongly equitable partition of stable graph to a binding graph.
\begin{lemma}\label{lem:basicl}
  In  the graph $\Phi:=(\phi_{ij})$, for  two basic vertices $u$ and $v$, if $\{p_1,\ldots,p_{n-1}\}$ and $\{q_1,\ldots, q_{n-1}\}$ are all binding vertices of $u$ and  of $v$, respectively, the followings  are equivalent.
\begin{enumerate}
  \item $\phi_{uu}=\phi_{vv}$. That is, vertices $u$ and  $v$ are in the same cell.
  \item $\mset{\phi_{up_1},\ldots,\phi_{up_{n-1}}}\equiv \mset{\phi_{vq_1},\ldots,\phi_{vq_{n-1}}}$. That is,  the collection of all labels on $u$'s binding edges coincides with the collection of all labels on $v$'s binding edges.
\item $\mset{\phi_{up_1},\ldots,\phi_{up_{n-1}}}\cap \mset{\phi_{vq_1},\ldots,\phi_{vq_{n-1}}}\ne\emptyset$. That is, the collection of all labels on $u$'s binding edges is not disjoint with the collection of all labels on $v$'s binding edges.
  \item $\mset{\phi_{p_1p_1},\ldots,\phi_{p_{n-1}p_{n-1}}}\equiv \mset{\phi_{q_1q_1},\ldots,\phi_{q_{n-1}q_{n-1}}}$. That is,  the collection of all labels on $u$'s binding vertices coincides with the collection of all labels on $v$'s binding vertices.%
\end{enumerate}
\end{lemma}
\begin{proof}
We now show that the first three conclusions are equivalent.
 If two basic vertices $u,v$ are in the same cell, then the entries of two rows in $\WL(\Phi)$ are equal as multisets. Since a stable graph recognizes binding edges,  the collection of all labels on binding edges of $u$ in $\Phi$ is the same as the collection of all labels on binding edges of $v$.

On the other hand, if basic vertices $u$ and $v$ are in different cells of $\WL(B)$, then  $\phi_{uu}\ne \phi_{vv}$. Let $p_i:=u\dot{\wedge}u_i$ and $q_j:=v\dot{\wedge}v_j$  be any two binding vertices of $u$ and $v$, respectively. We then have  $\phi_{u p_i}\ne \phi_{vq_j}$ from the definition of $\Phi$ and the strongly equitable partition of $\hat{\Phi}$. That is, in this case
$$\mset{\phi_{up_1},\ldots,\phi_{up_{n-1}}}\cap \mset{\phi_{vq_1},\ldots,\phi_{vq_{n-1}}}=\emptyset\,.$$
 That shows
\begin{align}
  \phi_{uu}=\phi_{vv}&\quad\Longleftrightarrow\quad\mset{\phi_{up_1},\ldots,\phi_{up_{n-1}}}\equiv \mset{\phi_{vq_1},\ldots,\phi_{vq_{n-1}}}\notag\\
  &\quad\Longleftrightarrow\quad \mset{\phi_{up_1},\ldots,\phi_{up_{n-1}}}\cap \mset{\phi_{vq_1},\ldots,\phi_{vq_{n-1}}}\ne\emptyset\,.%
\end{align}%
That shows the first three conclusions are equivalent to each other.
The last conclusion is equivalent to the first three ones by the last result in Proposition \ref{prop:dbqd}.

That completes the proof.
\end{proof}

Further, we have an important fact: The stable labels on the edges of any pair of basic vertices are uniquely determined by the label to their binding vertex.

\begin{lemma}\label{lem:bv}
  In the stable graph $\WL(\Phi):=(m_{ij})$, for basic vertices $u,v,u',v'$ and $p=u\dw v, q=u'\dw v'$, we have
  \begin{align}\label{eq:bveq}
    m_{uv}=m_{u'v'}\quad\Longleftrightarrow\quad m_{pp}=m_{qq}\quad\Longleftrightarrow\quad m_{up}=m_{u'q}\quad \text{and}\quad m_{pv}=m_{qv'} \,.
  \end{align}
\end{lemma}
\begin{proof}
Let $\WL(\Phi)\dd\WL(\Phi):=(m^{(2)}_{ij})$, $p:=u\dot{\wedge}v$ and $q:=u'\dot{\wedge}v'$. We have, for $n_1=n(n+1)/2$,
\begin{align}\label{eq:uv}
 m^{(2)}_{uv}&=\sum_{k\in[n_1]}m_{uk}\dd m_{kv}=
m_{up}\dd m_{pv}+\sum_{k\in[n_1]\backslash\{p\}}m_{uk}\dd m_{kv}\,,\\
m^{(2)}_{u'v'}&=\sum_{k\in[n_1]}m_{u'k}\dd m_{kv'}=
m_{u'q}\dd m_{qv'}+\sum_{k\in[n_1]\backslash\{q\}}m_{u'k}\dd m_{kv'}\,.
\end{align}
Since a stable graph recognizes binding edges, we have,  by uniqueness of binding vertex to $u,v$ or $u',v'$,
$$\mset{m_{up}\dd m_{pv},m_{u'q}\dd m_{qv'}}\cap\mset{ m_{uk}\dd m_{kv}\mid k\in[n_1]\backslash\{p\}}=\emptyset$$
 and
 $$\mset{m_{up}\dd m_{pv},m_{u'q}m_{qv'}}\cap\mset{ m_{u'k}\dd m_{kv'}\mid k\in[n_1]\backslash\{q\}}=\emptyset\,.$$
If $m_{uv}=m_{u'v'}$, then $m^{(2)}_{uv}=m^{(2)}_{u'v'}$ by stableness of $\WL(\Phi)$. It further implies $m_{up}\dd m_{pv}=m_{u'q}\dd m_{qv'}$ from facts above. It means $m_{up}=m_{u'q}$ and $m_{pv}=m_{qv'}$. By Proposition \ref{prop:dbqd}, we get $m_{pp}=m_{qq}$. (Remember that we assumed the labels of binding edges and vertices in $\Phi$ are inherited from $\WL(\Phi)$.)

On the other side, if $m_{pp}=m_{qq}$, then $m_{up}\dd m_{pv}=m_{u'q}\dd m_{qv'}$ by Proposition \ref{prop:dbqd}.
It means   $m_{up}=m_{u'q}$ and $m_{pv}=m_{qv'}$.   We then have $m_{up}^{(2)}=m_{u'q}^{(2)}$ in $\WL(\Phi)\dd\WL(\Phi)$, and
\begin{align}
m_{up}^{(2)}&= m_{up}\dd m_{pp}+m_{uv}\dd m_{vp}+m_{uu}\dd m_{up}+\sum_{k\in[n_1]\backslash\{p,u,v\}}m_{uk}\dd m_{kp}\,,\label{eq:up}\\
m_{u'q}^{(2)}&= m_{u'q}\dd m_{qq}+m_{u'v'}\dd m_{v'q}+m_{u'u'}\dd m_{u'q}+\sum_{k\in[n_1]\backslash\{q,u',v'\}}m_{u'k}\dd m_{kq}\,.\label{eq:u'q}
\end{align}
Note that $m_{up}, m_{pv}$ are only labels on $p$'s binding edges appeared in \eqref{eq:up}, and $m_{u'q}, m_{qv'}$ are only labels on  $q$'s binding edges appeared in \eqref{eq:u'q}. Since a stable graph recognizes binding edges, $m_{up}^{(2)}=m_{u'q}^{(2)}$ implies that
$$m_{up}\dd m_{pp}+m_{uv}\dd m_{vp}+m_{uu}\dd m_{up}=m_{u'q}\dd m_{qq}+m_{u'v'}\dd m_{v'q}+m_{u'u'}\dd m_{u'q}\,.$$
That further implies $
m_{up}\dd m_{pp}=m_{u'q}\dd m_{qq},
m_{uu}\dd m_{up}=m_{u'u'}\dd m_{u'q}$ and $m_{uv}\dd m_{vp}=m_{u'v'}\dd m_{v'q}$. That shows $m_{up}=m_{u'q}$, $m_{vp}=m_{v'q}$ and $m_{uv}=m_{u'v'}$. By the converse equivalent property, we have $m_{pu}=m_{qu'}$ and $m_{pv}=m_{qv'}$.

The remaining part is similar to prove.
That ends the proof.
\end{proof}

For a binding graph $B$, Lemma \ref{lem:bv} states that the stable graph $\WL(B)$ is completely determined by the labels on binding vertices.
The result in Lemma \ref{lem:bv} can be interpreted as such: For a stable graph $\WL(B):=(m_{ij})$ of a binding graph $B$ of order $n$, a complete labelling graph $X:=(x_{ij})$ of order $n$ satisfies that the labels $x_{uv}=x_{rs}$ iff $m_{up}=m_{rq}$  for $u,v,r,s\in[n]$ with $u<v,r<s$, $p=u\dw v$ and $q=r\dw s$, then the  graph $X$ is equivalent to the subgraph of $\WL(B)$ induced by  $[n]$.

We now consider a graph $\Theta$  as the result from $\Phi$ by substituting the labels on all basic vertices of $\Phi$ with $x_0$,  and substituting the labels on all binding edges of $\Phi$ with $x\in\Var$ and $x\notin\Phi$. Formally, the graph $\Theta:=(\theta_{ij})$ is defined as follows, for all  $i,j\in[n_1]$,
   \begin{align}\label{eq:theta}
\theta_{ij}:=\begin{cases}
     \phi_{ii}, & \hbox{if $i=j,\,i\in[n+1..n_1]$,}\\
      x, & \hbox{if $(i,j)$ is a binding edge,}\\
      x_0, &\hbox{otherwise.}
     \end{cases}
\end{align}
In other words, the graph $\Theta$ is a bipartite graph with all binding edges labeled with $x$, and all binding vertices inherit the corresponding labels from $\WL(\Phi)$.

Lemma \ref{lem:basicl} claims that the labels on basic vertices in $\Phi$ are completely determined by all the labels on their binding vertices, and the labels on binding edges in $\Phi$ are completely determined by the labels on binding vertices too. That implies by Lemma \ref{lem:bv} that the graph $\WL(\Phi)$ are determined by labels on their binding vertices, and by
Lemma \ref{lem:bv} that the graph $\WL(\Phi)$ is determined by labels on binding vertices too. By the construction of $\Theta$, we thus have the following conclusion. (A direct proof is sketched in appendix Section \ref{sec:proof1} and \ref{sec:proof}.)

\begin{theorem}\label{thm:theta}
 $\WL(\Phi)\approx\WL(\Theta)$ and $\hat{\Phi}\approx\hat{\Theta}$.%
\end{theorem}

This result in fact shows the labels of a stable graph of any binding graph is determined by the labels on binding vertices, provided the labeling induces equitable partition. The equitable partition is guaranteed by the regularity in bipartite like $\Theta$. We introduce a general notion of labeling to binding vertices in the next section.

\section{The Stable Partition of a Binding Graph is the Automorphism Partition}\label{sec:aut}

A \emph{plain binding graph} is a binding graph with an empty graph (a graph with only blank edges) as the basic graph.  It is unique (up to isomorphism) and denoted as $\Pi_n$ if the basic empty graph is of order $n$. There are totally $n_1:=n(n+1)/2$ vertices and $n(n-1)/2$ binding vertices in $\Pi_n$.

A binding vertex (b.v. for short) labeling  $\pi$ to  $\Pi_n:=(m_{ij})$ is an assignment to binding vertices in $[n+1..n_1]$ with variables in $\Var$ such that $\pi(i)\notin\{x_0,x\}$ for $\forall i\in [n+1..n_1]$ where $x$ is the label to binding edges. The graph after assignment by $\pi$ is denoted as $\Pi_n^\pi$.

The b.v. labeling $\pi$ will form a partition $\D:=(D_1,D_2,\ldots, D_d)$ to binding vertices such that each cell $D_i$ contains all binding vertices with the same label.

For each basic vertex $u$ of $\Pi_n^\pi$, the collection (multiset) $T_u$ of the labels on all its binding vertices is called the \emph{labeling type} by $\pi$, or just \emph{$\pi$-type} of $u$.  A collection $C$ of all basic vertices that share the identical $\pi$-type is a basic cell induced by $\pi$. In this way, $\pi$ will induce a partition of basic vertices denoted as $\C:=(C_1,\ldots, C_c)$.

The \emph{degree} of a basic vertex $u$ to a cell $D\in\D$ is the number of $u$'s neighbours in $D$, denoted as $\deg(u,D)$. If the degrees of all vertices from a basic cell $C\in\C$ to the binding  cell $D$ are identical, we say  $C$ is \emph{regular to} $D$, and denote this degree as $\deg(C,D)$. Similarly, we define $D$ is \emph{regular} to $C$,  and denote $\deg(p,C)$ for $p\in D$ and $\deg(D,C)$.

A b.v. labeling $\pi$ is said to be a\emph{ stable b.v. labeling} if any induced  binding cell $D\in \D$ is regular to any induced basic cell $C\in\C$.

\begin{theorem}
For a b.v. stable labeling $\pi$ to plain binding graph $\Pi_n$ and the partitions $\C:=(C_1,\ldots,C_c)$ and $\D:=(D_1,\ldots,D_d)$ induced by $\pi$,  it holds that
\begin{enumerate}
  \item Any basic cell from $\C$ and binding cell from $\D$ are regular to each other.
  \item $(\C,\D)$ is the stable partition to $\Pi_n^\pi$.%
\end{enumerate}
\end{theorem}
\begin{proof}
For a basic cell $C\in\C$ and $u,v\in C$, the $\pi$-type $T_u$ of $u$ and $\pi$-type $T_v$ of $v$ satisfy $T_u\equiv T_v$ by the definition. For a cell $D\in\D$ determined by some label $\mathtt{d}\in\Var$, the degree $\deg(u,D)$ is the number of $u$'s binding vertices in $D$, and also the number of $\mathtt{d}$ in $T_u$. That indicates $\deg(u,D)=\deg(v,D)$ by $T_u\equiv T_v$ for all $u,v\in C$. Hence, $C$ is regular to $D$. On the other side, $D$ is regular to $C$ by the definition of stable b.v. labelings.

The regularity shown above implies that the partition $(\C,\D)$ to $\Pi_n^\pi$ is an equitable partition.  It is straightforward from Lemma \ref{lem:basicl}, \ref{lem:bv} and Theorem \ref{thm:theta} to obtain that $(\C,\D)$ is the stable partition to $\Pi_n^\pi$.

\end{proof}

If $|C_i|=1$ for all $i\in[c]$, it is called \emph{a discrete partition} to basic vertices.
 If a stable b.v. labeling $\pi$ induces a discrete partition to basic vertices, then $|C_i|=1$ for all $i\in[c]$. Let $|D_1|:=d_1$ and $p=u\dot{\wedge}v$ be a binding vertex in $D_1$ for two different basic vertices $u,v$. The singleton of basic cells implies vertex $u\in C_{i_1}=\{u\}$ and $v\in C_{i_2}=\{v\}$ with $i_1\ne i_2$. Since $D_1$ is regular to both $C_{i_1}$ and $C_{i_2}$, all vertices in $D_1$ are binding vertices between $u$ and $v$. That forces $d_1=1$ due to the uniqueness of binding vertex between any pair of basic vertices. That shows that $|D_j|=1$ for all binding cells in this setting.

On the other side, if $\pi$ assigns distinct labels to binding vertices, it is easy to verify that $\pi$ induces a discrete partition to basic vertices. If  a labeling $\pi$ induces a discrete partition to binding vertices, it is  then called \emph{a discrete labeling}. We thus have the following lemma.

\begin{lemma}
  For $n>2$, a stable b.v. labeling to $\Pi_n$ induces a discrete partition if and only if it is a discrete b.v.  labeling.%
\end{lemma}

If there exist edges connecting two cells, it is said that the two cells are neighbor. By the regularity between cells induced by a stable b.v. labeling, it is easy to see that each binding cell $D\in\D$ will have at most two basic cells as neighbors, since the binding vertices are all of degree $2$. We say that $D$ binds the two cells in this setting. Whenever the two cells are identical as $C\in\C$,  it should be $\deg(D,C)=2$. The binding cell $D$ is then said to be an exclusive binding cell to $C$. It is easy to see that there is at least one exclusive binding cell for each basic cell $C$, provided $|C|>1$. Note that an exclusive binding cell to $C$ consists of binding vertices among vertices of $C$.


We now consider two stable b.v. labelings $\pi_1$ and $\pi_2$ to plain binding graph $\Pi_n$.

\begin{definition}[Similar stable b.v. labelings]\label{def:similar}
  Let $\pi_1$ and $\pi_2$ be two stable b.v. labelings to plain binding graph $\Pi_n$. Let the partitions $\C:=(C_1,\ldots, C_c)$ and $\D:=(D_1,\ldots,D_d)$ be basic vertices partition and binding vertices partition induced by $\pi_1$, respectively, and the partitions $\mathcal{E}:=(E_1,\ldots, E_e)$ and $\mathcal{F}:=(F_1,\ldots, F_f)$ be basic vertices partition and binding vertices partition induced by $\pi_2$, respectively. If the followings are satisfied up to reordering of cells induced by $\pi_2$, labelings $\pi_1$ and $\pi_2$ are said to be equivalent.
\begin{enumerate}
  \item $c=e$ and $d=f$.
  \item $|C_i|=|E_i|$ and $|D_j|=|F_j|$ for all $i\in[c],j\in[d]$.
  \item $\deg(C_i,D_j)=\deg(E_i,F_j)$ and $\deg(D_j,C_i)=\deg(F_j,E_i)$ for all $i\in[c],j\in[d]$.%
\end{enumerate}
The label to vertices in $D_j$ and the label to vertices in $F_j$ are now called \emph{corresponding labels}.

If two equivalent stable b.v. labelings have the same corresponding labels, we say they are \emph{similar labelings}.%
\end{definition}

We will show that two graphs $\Pi_n^{\pi_1}$ and $\Pi_n^{\pi_2}$ are isomorphic for  two similar stable b.v. labelings $\pi_1$ and $\pi_2$.
 This conclusion is proved by induction on $n$. To use the inductive assumption, we have to construct a stable b.v. labeling $\pi'$ as a refinement of $\pi$, where the stable graph obtained by WL process will be used. For that purpose, we stress here some results about the stable graph obtained by WL process. The proofs to these results are similar to the previous proofs of the corresponding results about stable graphs obtained by SaS process (for example, Proposition \ref{prop:wdtxzhi} and Lemma \ref{lem:shbysb}), and omitted.

 \begin{proposition}
   The stable graph $\WL(A)$ obtained by WL process has the following properties.
   \begin{itemize}
     \item $\WL(A)$ recognizes vertices and edges. In case $A$ is a binding graph, $\WL(A)$ recognizes binding vertices and binding edges.
     \item The vertex partition of $\WL(A)$ is a strongly equitable partition.%
   \end{itemize}
 \end{proposition}

 A stable b.v. labeling $\pi'$ is \emph{a refinement of $\pi$} when it  satisfies that  $\pi'(p)=\pi'(q)$ only if $\pi(p)=\pi(q)$ for binding vertices $p,q$.

\begin{lemma}\label{lem:tech}
  A stable b.v. labeling $\pi$ to $\Pi_n$ induces a stable b.v. labeling $\pi'$ as a refinement of $\pi$.%
\end{lemma}
\begin{proof}
  Given a stable b.v. labeling $\pi$, we will construct a stable labeling $\pi'$  from $\Pi_{n}^\pi$.

  Let $X:=(x_{ij})$ be the subgraph of $\WL(\Pi_n^{\pi}):=(w_{ij})$ induced by $[n]$. It is a stable graph obtained by WL process. Moreover,  by Lemma \ref{lem:bv},  for $u,v,r,x\in[n]$ and $u<v,r<s$, it holds that $x_{uv}=x_{rs}$ if and only if $w_{up}=w_{rq}$ (and $w_{pv}=w_{qs}$ by converse equivalent property).

  Let $\mathcal{C}:=(C_1,\ldots, C_c)$ and $\mathcal{D}:=(D_1,\ldots,D_d)$ be the partitions to basic and binding vertices, respectively, induced by $\pi$ in $\Pi_n^\pi$.
   The partition $\C$ is also the vertex partition of graph $X$.

  Assume $1\in C_1$, we will construct a graph $X_0$ from $X$ by turning  the labels on vertex $1$ and all $1$'s incident edges to $x_0$.

Assume the stable graph $X$ is in forms as follows (produce row-column permutation if necessary): The identical entries of the first row (column) locate consecutively, and in such a way, the same label are in the same block and different labels are in different blocks. Formally, let $X=(x_{ij})$ and its  block form is as follows.
\begin{align}\label{eq:xwdt}
 X&=\left(\begin{smallmatrix}
  X_{11}&X_{12}&X_{13}&\cdots&X_{1s}\\
   X_{21}&X_{22}&X_{23}&\cdots&X_{2s}\\
    X_{31}&X_{32}&X_{33}&\cdots&X_{3s}\\
     \vdots&\vdots&\vdots&\vdots&\vdots\\
      X_{s1}&X_{s2}&X_{s3}&\cdots&X_{ss}
\end{smallmatrix}\right)\,,&
 X_0&=\left(\begin{smallmatrix}
  \bar{X}_{11}&\bar{X}_{12}&\cdots&\bar{X}_{1s}\\
   \bar{X}_{21}&X_{22}&\cdots&X_{2s}\\
    \bar{X}_{31}&X_{32}&\cdots&X_{3s}\\
     \vdots&\vdots&\vdots&\vdots\\
      \bar{X}_{s1}&X_{s2}&\cdots&X_{ss}
\end{smallmatrix}\right)\,.
\end{align}
Here the first block is $X_{11}=(x_{11})$, and each block $X_{1i}$ is an $1\times \ell_i$ matrix, such that all entries in $X_{1i}$ are identical. Any entry in $X_{1i}$ is different with entries in $X_{1j}$ for $i\ne j$. A block $X_{ii}$ is a $\ell_i\times \ell_i$ matrix, and $X_{ij}$ a $\ell_i\times \ell_j$ matrix for $i,j\in[s]$.

We now change all entries in  $X_{1i}$ and $X_{i1}$ into $x_0$ and obtain a graph $X_0$ as shown in \eqref{eq:xwdt}. Remember that $x_0\notin X$ and hence $x_0$ appears only in the first row and first column of $X_0$.
 In $X_0$, all entries of $\bar{X}_{1j}$ and $\bar{X}_{j1}$ are  $x_0$, and $X_{ij}$ are all the same as in $X$.

 We now consider the stable graphs of $X_0$.
Set $X\diamond X:=(x_{ij}^{(2)})$ and $X_0\diamond X_0:=(z_{ij})$ in the following discussions, without further explanations.
Since the diagonal entries of $X$ (and of $X_{ii}$) are labels of vertices, we have the claims.\ep

\noindent\textbf{Claim} 1. For all $t\in[s]$, the vertices in each block $X_{tt}$ have the same label.\ep

\noindent \emph{Proof of Claim 1.}
 Let $u,v$ be two vertices in $X_{tt}$, the formation of $X$ indicates $x_{1u}=x_{1v}$ in $X_{1t}$. By the stability of $X$, it should hold that $x^{(2)}_{1u}=x^{(2)}_{1v}$ in $X\diamond X$. That is $\sum_{k\in[n]}x_{1k}\diamond x_{ku}=\sum_{k\in[n]}x_{1k}\diamond x_{kv}$, and we obtain $x_{11}\diamond x_{1u}+ x_{1u}\diamond x_{uu}=x_{11}\diamond x_{1v}+x_{1v}\diamond x_{vv}$ since the stable graph $X$ recognizes vertices. That means $x_{1u}\diamond x_{uu}=x_{1v}\diamond x_{vv}$, we thus have $x_{uu}=x_{vv}$.

That ends the proof of Claim 1.  $\hfill\lhd$\ep

 Claim 1 tells that the vertices in each block are in the same cell $C\in\C$. A cell of $X$ may be split into several blocks in this case. That is, some blocks may share the same vertex label.
 \ep

\noindent\textbf{Claim} 2. In the stable graph $\WL(X_0)$ of  $X_0$, vertices from different blocks $X_{aa}$ and $X_{bb}$ of $X_0$ can not have the same label for all $a,b\in[2..s]$, and $a\ne b$.\ep

\noindent \emph{Proof of Claim 2.}
For $a\ne b,\, a,b\in[2..s]$,  all vertices in $X_{aa}$ (respectively, in $X_{bb}$) have the same label by Claim 1. If the vertices labels in $X_{aa}$ and $X_{bb}$  are different, then all vertices in $X_{aa}$ cannot share the same  labels with any vertices in $X_{bb}$ in $\WL(X_0)$ according to the recognizing property $\mathbf{R}$ of a stable graph (valid to WL process).

We consider the case that the vertex label in $X_{aa}$ is the same as the vertex label in $X_{bb}$. That means $x_{uu}=x_{vv}$ for
 any vertex  $u$ in $X_{aa}$ and  any vertex  $v$ in $X_{bb}$. By the stability of $X$, we have $x^{(2)}_{uu}=x^{(2)}_{vv}$ in $X\diamond X$.

Since $a\ne b$, the formation of  $X$ implies $x_{1u}\ne x_{1v}$. In  $X_0\diamond X_0$,
\begin{align}
  z_{uu}&=\bar{x}_{u1}\diamond \bar{x}_{1u}+\sum_{k=2}^{n}x_{uk}\diamond x_{ku}=x_0\diamond x_0-x_{u1}\diamond x_{1u}+x^{(2)}_{uu}\,,\label{eq:tmp1}\\
   z_{vv}&=\bar{x}_{v1}\diamond \bar{x}_{1v}+\sum_{k=2}^{n}x_{vk}\diamond x_{kv}=x_0\diamond x_0-x_{v1}\diamond x_{1v}+ x^{(2)}_{vv}\,.\label{eq:tmp2}
   \end{align}
We thus get $z_{uu}\ne z_{vv}$ by the facts $x_{1u}\ne x_{1v}$ and $x^{(2)}_{uu}=x^{(2)}_{vv}$ in \eqref{eq:tmp1}, \eqref{eq:tmp2}.

The fact that $X_0\diamond X_0\rightarrowtail\WL(X_0)$ implies that $u$ and  $v$ have different labels in $\WL(X_0)$.

That is the end of Claim 2. $\hfill\lhd$\ep

\noindent\textbf{Claim} 3. For  $a,b\in[2..s]$ and vertices $u$ and $v$ in $X_{aa}$ and $X_{bb}$, respectively, it should be $a=b$  if  $z_{wu}=z_{wv}$ in $X_0\diamond X_0$ for all vertex $w\in[n]\backslash\{1\}$.\ep

\noindent \emph{Proof of Claim 3.}
 In $X_0\diamond X_0$,
  \begin{align}
   z_{wu}&=\bar{x}_{w1}\diamond \bar{x}_{1u}+\sum_{k=2}^{n}x_{wk}\diamond x_{ku}=x_0\diamond x_0-x_{w1}\diamond x_{1u}+x^{(2)}_{wu},\label{eq:zwu}\\
   z_{wv}&=\bar{x}_{w1}\diamond \bar{x}_{1v}+\sum_{k=2}^{n}x_{wk}\diamond x_{kv}=x_0\diamond x_0-x_{w1}\diamond x_{1v}+x^{(2)}_{wv}.\label{eq:zwv}
  \end{align}
 If $z_{wu}=z_{wv}$, then from \eqref{eq:zwu}, \eqref{eq:zwv} and the fact that $X$ recognizes vertices, we have   $x_{wu}\diamond x_{uu}+x_{ww}\diamond x_{wu}=x_{wv}\diamond x_{vv}+x_{ww}\diamond x_{wv}$. That further implies $x_{wu}\diamond x_{uu}=x_{wv}\diamond x_{vv}$ and $x_{ww}\diamond x_{wu}=x_{ww}\diamond x_{wv}$. We then have $x_{wu}=x_{wv}$. Therefore,
 $\sum_{k=1}^{n}x_{wk}\diamond x_{ku}=\sum_{k=1}^{n}x_{wk}\diamond x_{kv}$ by the stability of $X$.  That implies  $x_{w1}\diamond x_{1u}=x_{w1}\diamond x_{1v}$ from \eqref{eq:zwu}, \eqref{eq:zwv}. We thus have $x_{1u}=x_{1v}$, which means $a=b$ by the formation of $X$.

That ends the proof of Claim 3. $\hfill\lhd$\ep

We claim that the rows in each $X_{ab}$ are equivalent as multiset. \ep

\noindent\textbf{Claim} 4. For $a,b\in[2..s]$, let $U:=\{u+1,\ldots,u+m_a\}$ and $V:=\{v+1,\ldots,v+m_b\}$ be the sets containing all vertices in $X_{aa}$  and $X_{bb}$, respectively. Denote $T_{i}:=\mset{x_{(u+i)\,(v+j)}\mid j\in [m_b]}$ for all $u+i\in U$, then $T_{1}\equiv T_{2}\equiv\cdots\equiv T_{m_a}$.\ep

\noindent \emph{Proof of Claim 4.}
We only show $T_1\equiv T_2$. By the formation of $X$, for some $\alpha,\beta\in\Var$, we have
$$
X_{1a}=(x_{1\,u+1},\ldots,x_{1\,u+m_a})=(\alpha,\ldots,\alpha),\quad
X_{1b}=(x_{1\,v+1},\ldots,x_{1\,v+m_b})=(\beta,\ldots,\beta)\,.
$$

Note that $x_{1\,u+1}=x_{1\,u+2}=\alpha$ implies $x_{1\,u+1}^{(2)}=x^{(2)}_{1\,u+2}$ in $X\diamond X$. That is
$$x_{1\,u+1}^{(2)}=\sum_{k\in [n]}x_{1k}\diamond x_{k\,u+1}=\sum_{k\in [n]}x_{1k}\diamond x_{k\, u+2}=x_{1\,u+2}^{(2)}\,.$$
Since the entries in different blocks  of $X_{1i}$ and $X_{1j}$ in $X$ are with different variables, we get, from the noncommutativity of diamond product,
$$\sum_{k\in [m_b]
}x_{1\,v+k}\diamond x_{v+k\,u+1}=\sum_{k\in[m_b]}x_{1\,v+k}\diamond x_{v+k\,u+2}\,.$$
That is, $\sum_{k\in V}\beta\diamond x_{k\,u+1}=\sum_{k\in V}\beta\diamond x_{k\,u+2}$. Which gives $$\mset{x_{v+1\,u+1},\ldots,x_{v+m_b\,u+1}}\equiv\mset{x_{v+1\,u+2},\ldots,x_{v+m_b\,u+2}}.$$ By the converse equivalent property, we obtain $T_1\equiv T_2$.
Similarly, we obtain $T_1\equiv T_j$ for $j\in[m_a]$.

That ends the proof of Claim 4. $\hfill\lhd$\ep

We now update the labels in each $X_{ij}$ in $X_0$ block by block for $i,j\in[n]\backslash\{1\}$ in the way as follows: For each $X_{ij}$ chosen, we perform an equivalent variable substitution to $X_{ij}$ with new variables from $\Var$ that have not appeared in $X,X_0$ nor in any other substituted blocks. The resulting block is denoted as $Y_{ij}$ for $i,j\in[n]\backslash\{1\}$.

After updating of all $X_{ij}$, let $Y_{1i}:=X_{1i}$ and $Y_{j1}:=X_{j1}$.
The resulting graph is denoted as $Y$ shown as follows.

\begin{align} \label{eq:y}
 Y&:=\left(\begin{smallmatrix}
Y_{11}&Y_{12}&\cdots&Y_{1s}\\
   Y_{21}&Y_{22}&\cdots&Y_{2s}\\
    Y_{31}&Y_{32}&\cdots&Y_{3s}\\
     \vdots&\vdots&\vdots&\vdots\\
     Y_{s1}&Y_{s2}&\cdots&Y_{ss}
\end{smallmatrix}\right)\,.
\end{align}

It is not hard to see the following properties  from  the updating  procedure. For any $u,v,r,s\in[n]$,
\begin{itemize}
  \item $y_{uv}=y_{st}$ only if $x_{uv}=x_{st}$, and
  \item $y_{uv}=y_{st}$ implies $y_{vu}=y_{ts}$ (the converse equivalent property).
  \item $Y$ recognizes vertices since $X$ does.
  \item rows (and columns) are equivalent in each $Y_{ij}$ guaranteed by Claim 4.%
\end{itemize}

%
%
%
%

Now, we construct a stable b.v. labeling $\pi'$ to $\Pi_n$ based on $Y$. For each pair of basic vertices $u,v$ with $u<v$ and $p=u\dw v$, we set the label on binding edge $(u,p)$ as $y_{uv}$, and the label on binding edge $(p,v)$ as $y_{vu}$. After all binding edges are set, we turn to define the labels on binding vertices as such: For every two binding vertices $p=u\dw v$ and $q=r\dw s$, the label $\pi'(p)$ on $p$ and label $\pi'(q)$ on $q$ are identical iff $\mset{y_{uv},y_{vu}}=\mset{y_{rs},y_{sr}}$. That completes the definition of $\pi'$.

From the properties stated at the very beginning of the proof and properties listed just below \eqref{eq:y}, it is easy to see that $\pi'$ is indeed a refinement of $\pi$. (or to see it by Claim 1.)

 Labeling $\pi'$ induces the vertex partition to basic vertices as $C':=(\{1\},C_2',\ldots,C_s')$. Denote the vertex partition to binding vertices in $\Pi_n^{\pi'}$ as $\D'=(D'_1,\ldots,D'_{d'})$. For each $D'\in\D'$ and $C'\in\C'$, $D'$ is regular to $C'$ by Claim 4.

The b.v. labeling $\pi'$ is therefore a stable b.v. labeling to $\Pi_n$ as a refinement of $\pi$, by Claims.
That is the end of the proof.
\end{proof}

We now are able to prove the following conclusion using Lemma \ref{lem:tech}.
\begin{theorem}\label{thm:simiso}
   If $\pi_1,\pi_2$ are two similar stable b.v. labelings to $\Pi_n$ with $n>2$, then $\Pi_n^{\pi_1}\cong\Pi_n^{\pi_2}$.%
\end{theorem}
\begin{proof}
   We prove the claim by induction on the number of basic vertices $n$.

For $n=3$ it is straightforward to verify the validity of the claim. Assuming the validity of the claim for $\Pi_{n-1}$ with $n> 2$, we now consider two similar stable b.v. labelings $\pi_1$ and $\pi_2$ to $\Pi_n$.

 Let the partitions $(\C,\D)$ and $(\mathcal{E},\mathcal{F})$ be induced, respectively, by $\pi_1$ and $\pi_2$ to $\Pi_n^{\pi_1}$ and $\Pi_n^{\pi_2}$, where $\C,\D,\mathcal{E},\mathcal{F}$ are defined as in Definition \ref{def:similar}. The label to $D_j$ is the same as the label to $F_j$ for all $j\in[d]$, since $\pi_1,\pi_2$ are similar stable labelings.

Without loss of generality, we assume $1\in C_1$ and $1\in E_1$ (in fact one may assume $C_i=E_i$ and $D_j=F_j$ if necessary by permutation). We construct $\pi_1'$ as a refinement of $\pi_1$  in the same way as in the proof of Lemma \ref{lem:tech}.

Since $\pi_2$ is similar to $\pi_1$,
it holds that  $\deg(C_i,D_j)=\deg(E_i,F_j)$ and $\deg(D_j,C_i)=\deg(F_j,E_i)$ for all $i\in[c],j\in[d]$.
Also $|C_i|=|E_i|$ and $|D_j|=|F_j|$ for all $i\in[c],j\in[d]$. All corresponding variables are the same variable.

In this setting, any  block   induced by $1\in C_1$  is with the same size  as block induced by $1\in E_1$ as in $X_0$ in \eqref{eq:xwdt}. We thus obtain two graphs $X_0^{\pi_1}$ and $X_0^{\pi_2}$ by $\pi_1$ and $\pi_2$, respectively, just as $X_0$ obtained by $\pi$ in Lemma \ref{lem:tech}. They are as follows.

\begin{align}  \label{eq:nabla}
  X^{\pi_1}&=\left(\begin{smallmatrix}
  \nabla_{11}&\nabla_{12}&\cdots&\nabla_{1s}\\
   \nabla_{21}&\nabla_{22}&\cdots&\nabla_{2s}\\
    \nabla_{31}&\nabla_{32}&\cdots&\nabla_{3s}\\
     \vdots&\vdots&\vdots&\vdots\\
     \nabla_{s1}&\nabla_{s2}&\cdots&\nabla_{ss}\\
\end{smallmatrix}\right)\,,&
 X_0^{\pi_1}&=\left(\begin{smallmatrix}
  \bar{\nabla}_{11}&\bar{\nabla}_{12}&\cdots&\bar{\nabla}_{1s}\\
   \bar{\nabla}_{21}&\nabla_{22}&\cdots&\nabla_{2s}\\
    \bar{\nabla}_{31}&\nabla_{32}&\cdots&\nabla_{3s}\\
     \vdots&\vdots&\vdots&\vdots\\
      \bar{\nabla}_{s1}&\nabla_{s2}&\cdots&\nabla_{ss}
\end{smallmatrix}\right)\,,\\[0.3cm] X^{\pi_2}&=\left(\begin{smallmatrix}
  \Delta_{11}&\Delta_{12}&\cdots&\Delta_{1s}\\
   \Delta_{21}&\Delta_{22}&\cdots&\Delta_{2s}\\
    \Delta_{31}&\Delta_{32}&\cdots&\Delta_{3s}\\
     \vdots&\vdots&\vdots&\vdots\\
      \Delta_{s1}&\Delta_{s2}&\cdots&\Delta_{ss}\\
      \end{smallmatrix}\right)\,,&
 X_0^{\pi_2}&=\left(\begin{smallmatrix}
  \bar{\Delta}_{11}&\bar{\Delta}_{12}&\cdots&\bar{\Delta}_{1s}\\
   \bar{\Delta}_{21}&\Delta_{22}&\cdots&\Delta_{2s}\\
    \bar{\Delta}_{31}&\Delta_{32}&\cdots&\Delta_{3s}\\
     \vdots&\vdots&\vdots&\vdots\\
      \bar{\Delta}_{s1}&\Delta_{s2}&\cdots&\Delta_{ss}
\end{smallmatrix}\right)\,.\label{eq:nabla2}
\end{align}
Matrix blocks $\nabla_{ij},\Delta_{ij}$ have the same dimension and the same number of correspondent variables by the similarity, and  so do the blocks $\nabla_{1k},\Delta_{1k}$ and $\bar{\nabla}_{1k},\bar{\Delta}_{1k}$ for all $i,j\in[n]\backslash\{1\}$ and $k\in[n]$.

We proceed  $X_0^{\pi_1}$ as in the proof of Lemma \ref{lem:tech} to $X_0$ and finally obtain $\pi_1'$.
Similarly, we obtain $\pi'_2$ to $\pi_2$. But, update each $\Delta_{ij}$ in the same way as $\nabla_{ij}$ such that the corresponding variables with the same new variables for both of them. This guarantees that  $\pi_2'$ will have the same number of binding labels of each kind as $\pi_1$ does.

Let $\C':=(\{1\},C_2,\ldots,C_s')$, $\D':=(D_1',\ldots,D'_{d'})$ be partitions induced by $\pi_1'$, and $\mathcal{E}':=(\{1\},E_2,\ldots,E_s')$, $\mathcal{F}':=(F_1',\ldots,F'_{d'})$ be partitions induced by $\pi_2'$. The similarity between $\pi_1$ and $\pi_2$, and the procedure of constructions give
$\deg(C_i',D_j')=\deg(E_i',F_j')$ and $\deg(D_j',C_i')=\deg(F_j',E_i')$ for all $i\in[s],j\in[d']$.
Also $|C_i'|=|E_i'|$ and $|D_j'|=|F_j'|$ for all $i\in[s],j\in[d']$.

In this way, we obtain two similar stable b.v. labelings $\pi_1',\pi_2'$. Also, whenever they are restricted to $[2..n]\cup\{u\dw v\mid u,v\in[2..n]\}$, they are similar too, by constructions (notice that vertex $1$ in a singleton cell).

By induction assumption, whenever $\pi'_1$ and $\pi'_2$ are restricted to the subgraphs $\Pi_{n-1}^{\pi_1'}$ and $\Pi_{n-1}^{\pi_2'}$, respectively,  of  $\Pi_n^{\pi'_1}$ and $\Pi_n^{\pi'_2}$ induced by $[2..n]\cup\{u\dw v\mid u,v\in[2..n]\}$, then  $\Pi_{n-1}^{\pi_1'}\cong\Pi_{n-1}^{\pi_2'}$.

%
%
%
%

That means, the subgraph induced by $[2..n]\cup\{u\dw v\mid u,v\in[2..n]\}$ of $\Pi_n^{\pi_1'}$ is isomorphic to the subgraph induced by $[2..n]\cup\{u\dw v\mid u,v\in[2..n]\}$ of $\Pi_n^{\pi_2'}$. We denote this isomorphism as $\sigma$.
We then construct map $\delta$ from $\Pi_n^{\pi_1'}$ to $\Pi_n^{\pi_2'}$ as follows.
\begin{align}\label{eq:delta}
 w^\delta:=\left\{
             \begin{array}{ll}
               1, & \hbox{If $w=1$ ,} \\
               w^\sigma, & \hbox{If $w\in [2..n]$ ,} \\
               u^\sigma\dot{\wedge}v^\sigma,& \hbox{If $w= u\dot{\wedge}v$ and $u,v\in [2..n]$ ,}\\
               1\dot{\wedge}u^\sigma, & \hbox{If $w=1\dot{\wedge}u$ and $u\in[2..n]$.}
             \end{array}
           \right.
\end{align}
  Since $\sigma$ is a 1-1 map, the map $\delta$ is a 1-1 map that preserves adjacency. To show that it also preserves labels, we only need to check that $\pi_1'(1\dot{\wedge}u))=\pi_2'(1\dot{\wedge}u^\sigma)$. That is true because  $u$ and $u^\sigma$ should, respectively, be in  $\nabla_{ii}$ and $\Delta_{ii}$ for some $i\in[2..n]$ as shown in \eqref{eq:nabla2}. However, according to the construction, the variables in $\nabla_{1i}$ and $\Delta_{1i}$ are identical in $X^{\pi_1}$ and $X^{\pi_2}$ by \eqref{eq:nabla}.

  Hence, the map $\delta$ is an isomorphism from $\Pi_n^{\pi_1'}$ to $\Pi_n^{\pi_2'}$. In fact, $\delta$ is also an isomorphism from  $\Pi_n^{\pi_1}$ to $\Pi_n^{\pi_2}$, since $\pi_1'$ and $\pi_2'$ are, respectively, the refinements of $\pi_1$ and $\pi_2$.

The claim of the theorem is hence valid for all $n>2$ by induction.
That finishes the proof.

\end{proof}

Given a stable b.v. labeling $\pi$ to $\Pi_n$, it induces partition $(\C,\D)$ with $\C:=(C_1,\ldots,C_c)$ and $\D:=(D_1,\ldots, D_d)$. Assume that $|C_1|>1$ and $1,2\in C_1$, we perform a row-column permutation to transpose vertices $1$ and $2$ in $\Pi_n^{\pi}$ obtaining a graph $\Pi_n'$. Denote the b.v. labeling to $\Pi_n'$ as $\pi_1$.

Labelings $\pi$ and $\pi_1$ are similar stable b.v. labelings.
By the construction in \eqref{eq:delta}, we obtain an isomorphism $\delta'$ from $\Pi_n^{\pi}$ to $\Pi_n^{\pi_1}$ such that $1^{\delta'}=1$. While the vertex $1$ in $\Pi_n^{\pi_1}$ is in fact the vertex $2$ of $\Pi_n^{\pi}$ and $\delta'$ is an automorphism of $\Pi_n^{\pi}$, in fact it shows that $C_1$ is an orbit  of automorphism group of $\Pi_n^{\pi}$ by the arbitrariness of vertex $1,2$ in $C_1$.  By the arbitrariness of $C_1$ (by row-column permutation if necessary), it shows all basic cells are orbits of the automorphism group of $\Pi_n^{\pi}$.
The binding cells are orbits  can be concluded from the uniqueness of binding vertices  between any pair of basic vertices and the fact that basic cells are orbits. We thus have the following.

\begin{corollary}\label{corr:storb}
  A stable b.v. labeling $\pi$ to $\Pi_n$ induces the automorphism  partition to  $\Pi_n^\pi$.%
\end{corollary}

The labels on binding vertices in graph $\Theta$, by Lemma \ref{lem:basicl}, compose a stable b.v. labeling $\pi$ to $\Pi_n$. Corollary \ref{corr:storb} and Theorem \ref{thm:theta} give the main conclusion as follows.
\begin{theorem}\label{thm:main}
  The stable partition to any binding graph is the automorphism partition.%
\end{theorem}

\section{Decision Procedure of Graph Isomorphism and the Complexity}\label{sec:gid}

It is well known that the problem of computing automorphism partitions of graphs is polynomial-time equivalent to graph isomorphism problem (cf. Mathon \cite{Mathon79}). We now give an explicit decision procedure of graph isomorphism in this section and estimate the time complexity of it. The main purpose here is to show a polynomial-time procedure and by no means to pursue the optimum time complexity.

As is well known, if the graph isomorphism problem is polynomial-time solvable to connected simple graphs, so is the  problem of graph isomorphism in general. We will therefore focus on a decision procedure to connected simple graphs.

Before introducing the decision procedure $\mathtt{GI}$, some preparations are needed. Given any two connected simple graphs $A_1,A_2$ over $[n]$ such that the labels to the edges of them are $x\in\Var$ and the labels to blank edges (and vertices) are $x_0$, in order to make directly use of the conclusion in previous sections\footnote{The usage of a wing graph in $\mathtt{GI}$ is not essential, the reason of constructing a wing graph here is that the conclusions obtained previously are subject to connected simple graphs or binding graphs of connected simple graphs.},
 we have to combine them into one connected simple graph, which is called a wing graph of $A_1$ and $A_2$.

 The wing graph $A$ over $[2n+1]$ is a connected simple graph of order $2n+1$ such that the induced subgraph of $A$ by $[n]$ is just $A_1$, and the induced graph by $[n+1..2n]$ is just $A_2$. The last vertex of $A$ is adjacent to all vertices in $[2n]$.  Formally,
$$A:=\left(\begin{smallmatrix}
 A_1&X_0&\mathbf{x}\\
X_0&A_2&\mathbf{x}\\
\mathbf{x}'&\mathbf{x}'&x_0
\end{smallmatrix}\right)\,.
$$
Here $X_0$ is a square matrix of order $n$ with all entries being $x_0$, $\mathbf{x}$ is a $n\times 1$ matrix with all entries being $x$, and $\mathbf{x}'$ is a $1\times n$ matrix with all entries being $x$.

In this case, it is easy to see that $A_1\cong A_2$ if and only if every orbit of $\Aut(A)$ consists of vertices from both sets $[n]$ and $[n+1.. 2n]$, apart from the singleton orbit $\{2n+1\}$ (since vertex $2n+1$ is  the only vertex of degree $2n$ in $A$).

Now, we are ready to present the procedure $\mathtt{GI}$ for testing graph isomorphism as follows.\ep

\noindent \underline{\textbf{Procedure $\mathtt{GI}$}:} Input two connected simple graphs $A_1,A_2$ over $[n]$.
\renewcommand{\labelenumi}{\roman{enumi}.}
\begin{enumerate}
  \item Construct the wing graph $A$ of $A_1$ and $A_2$;
  \item Construct the binding graph $B$ of $A$;
  \item Compute the stable graph $\WL(B)$ with WL process;
  \item Form the vertex partition $\mathcal{C}$ of $\WL(B)$ ;
  \item Inspect the cells of $\mathcal{C}$ except the cell $\{2n+1\}$. If every basic cell contains vertices both from $[n]$ and $[n+1..2n]$, output ``YES''; otherwise output ``NO''.%
\end{enumerate}

By Theorem \ref{thm:complete} and \ref{thm:main} and $N=2n+1>2$, the cells in $\mathcal{C}$ are all orbits of $\Aut(B)=\Aut(\WL(B))$. It assures that the procedure $\mathtt{GI}$ outputs ``YES'' iff $A_1\cong A_2$. That claims, the procedure  $\mathtt{GI}$ indeed decides if $A_1\cong A_2$.

\renewcommand{\labelenumi}{\arabic{enumi}.}
We now estimate the time cost at each step in $\mathtt{GI}$.
\begin{enumerate}
  \item The construction of $A$ over $[N]$ from $A_1$ and $A_2$ accomplishes with $O(N^2)=O(n^2)$ steps, where $N:=2n+1$.
  \item The construction of binding graph $B$ over $[N_1]$ needs at most $O(N_1^2)=O(n^4)$ steps, where $N_1:=N(N+1)/2$.
  \item The evaluation of $\WL(B)$ costs at most $O(N_1^3\cdot N_1\cdot \log N_1)=O(n^{8}\log n)$ steps by Theorem \ref{thm:saswl}, assuming the cost for  multiplication of matrices of order $N_1$ is $O(N_1^3)$.

  \item The formation of vertex partition $\mathcal{C}$ of  $\WL(B)$ uses only $O(N_1)=O(n^2)$ steps.
  \item The examination of cells from $\mathcal{C}$ can be completed in $O(N_1)=O(n^2)$ steps.%
\end{enumerate}
     Totally, the procedure $\mathtt{GI}$ consumes at most $O(n^{8}\log n)$ steps for two graphs of order $n$. That allows us to conclude the following.
      \begin{theorem}
       Graph isomorphism problem is solvable in polynomial time.%
      \end{theorem}

\section{Brief Discussions}\label{sec:discuss}
In this work we have proposed the notion of description graphs and introduced three processes to obtain description graphs. The stable graphs obtained in our processes were proven to be equivalent to stable graphs obtained by WL process in the partition of vertices.  The stable partitions were then proven to be strongly equitable partitions. We have also proposed a new graph-isomorphism complete class of graphs named binding graphs and shown that the stable partitions to binding graphs are automorphism partitions. That leads to a polynomial-time procedure for testing of graph isomorphism.

The description graphs were defined only with regard to undirected graphs. The analogous notions can be defined to directed graphs.
The fact that graph isomorphism problem is in $\mathtt{P}$ will bring a bunch of problems into $\mathtt{P}$ (cf., e.g. Booth and Colbourn \cite{BoothCo79},Mathon \cite{Mathon79}, Luks \cite{Luks93} and Babai \cite{Babai15}), and answer some open problems in computational complexity classes relevant to graph isomorphism problem (cf. e.g.  K{\"o}bler, Sch{\"o}ning and Tor{\'a}n \cite{KoblerScTo93}).

The new class of binding graphs are expected to have more applications in graph theory and even in group theory. The process for graph isomorphism here may be extended and applied to more combinatoric structures like relational structures in general and even algebraic structures.

 It is interesting to apply some other approaches of graph isomorphism to the class of binding graphs.  The other interesting work  may include exploiting the optimum complexity of graph isomorphism in theory, and implementing more efficient programs in  practice for graph isomorphism using the results of this work.
Finally, we put the problem mentioned  in Section \ref{sec:argument} as an open problem as follows.

\begin{problem} Is identifying a graph from all unisomorphic ones polynomial-time equivalent to testing of graph isomorphism?%
\end{problem}

The identification of a graph here means a characterization of the graph, such as the canonical form of the graph, such that any graph that is not isomorphic to it is different in the characterization.

\section*{Acknowledgement}
The author is grateful to Eugene Luks, who kindly clarifies issues of my concerns on graph isomorphism every now and then.
Thanks to my students: Shujiao Cao, Jinyong Chang and Tianci Peng for discussions on this topic. Thanks also to Haiying Li, Stoicho D. Stoichev and Gerta R{\"u}cker for their deliveries and communications of some articles.

\bibliographystyle{alpha}

\bibliography{rui-bib}
\clearpage
\appendix
\section{Some Examples for Illustration of Contexts}
Here are some examples for clarifying the contexts. Different numbers in the matrices indicate different variables.

1. The first example to show that \emph{the usage of numerical labels will some times to obtain a faulty stable graph.}

      During the computation of a stable graph with SaS (or WL) process, we insist, theoretically, that the graph obtained in each iteration be a labeled graph with variables as labels. The reason is that whenever  performing an equivalent number substitution to $A_i^2$ with numbers to obtain $A_{i+1}$, the ordering of numbers used in the substitution will affects the multiplication of matrices in next iteration and sometimes leads to a faulty stable graph. We give an example as follows.

      In the following, we perform an equivalent number substitution to a real matrix $A:=(a_{ij})$ in such a way of ``first come, first served''(ff for short): We list the rows of $A$ in a sequence of row one, row 2,$\ldots$, row $n$, and delete the numbers in the sequence that previously appeared. The remaining list of numbers consists of all distinct numbers in $A$. We then substitute each number in $A$ with its position in the list, and obtain the resulting graph.  We name such kinds of equivalent number substitutions as \emph{ff-equivalent number substitutions}.

      The following example shows that it will sometimes get faulty stable graphs whenever use numerical labels in graphs.

  Let the adjacency matrix $A$ of a graph as follow. We relabel the vertices as $2$ to obtain $A_1$ such that $A_1$ recognizes vertices. After evaluation of $A_1^2$, we perform a ff-equivalent number substitution to $A_1^2$ to obtain $A_2$. It is a pity that $A_2^2\approx A_2$, which is a faulty stable graph to $A$. The genuine stable graph $\hat{A}$ is as shown. The vertex partition to $\hat{A}$ is

 $$\{\{1, 14\}, \{2, 19\}, \{3, 8\}, \{4\}, \{5, 15\}, \{6, 18\}, \{7,12\}, \{9\}, \{10\}, \{11, 17\}, \{13, 20\}, \{16\}, \{21\}\}.$$
 Which is the automorphism partition of graph $A$.

  \begin{align}
  A&=\left(\begin{smallmatrix}
  0 & 1 & 1 & 1 & 1 & 1 & 1 & 1 & 1 & 0 & 0 & 0 & 0 & 0 & 0 & 0 & 0 & 0 & 0 & 0 & 0 \\
 1 & 0 & 1 & 1 & 1 & 0 & 0 & 0 & 0 & 1 & 1 & 1 & 1 & 0 & 0 & 0 & 0 & 0 & 0 & 0 & 0 \\
 1 & 1 & 0 & 0 & 0 & 1 & 1 & 0 & 0 & 1 & 0 & 0 & 0 & 1 & 1 & 1 & 0 & 0 & 0 & 0 & 0 \\
 1 & 1 & 0 & 0 & 0 & 1 & 0 & 0 & 0 & 0 & 1 & 0 & 0 & 1 & 0 & 0 & 1 & 1 & 1 & 0 & 0 \\
 1 & 1 & 0 & 0 & 0 & 0 & 0 & 1 & 1 & 0 & 0 & 0 & 0 & 0 & 0 & 0 & 1 & 1 & 0 & 1 & 1 \\
 1 & 0 & 1 & 1 & 0 & 0 & 0 & 0 & 0 & 0 & 0 & 1 & 0 & 0 & 1 & 0 & 1 & 0 & 0 & 1 & 1 \\
 1 & 0 & 1 & 0 & 0 & 0 & 0 & 0 & 0 & 0 & 1 & 0 & 0 & 0 & 0 & 1 & 0 & 1 & 1 & 1 & 1 \\
 1 & 0 & 0 & 0 & 1 & 0 & 0 & 0 & 0 & 1 & 0 & 1 & 0 & 1 & 0 & 1 & 0 & 1 & 1 & 0 & 0 \\
 1 & 0 & 0 & 0 & 1 & 0 & 0 & 0 & 0 & 1 & 0 & 0 & 1 & 1 & 1 & 0 & 0 & 0 & 0 & 1 & 1 \\
 0 & 1 & 1 & 0 & 0 & 0 & 0 & 1 & 1 & 0 & 0 & 0 & 1 & 0 & 0 & 1 & 0 & 0 & 1 & 1 & 0 \\
 0 & 1 & 0 & 1 & 0 & 0 & 1 & 0 & 0 & 0 & 0 & 0 & 1 & 0 & 1 & 1 & 0 & 1 & 0 & 1 & 0 \\
 0 & 1 & 0 & 0 & 0 & 1 & 0 & 1 & 0 & 0 & 0 & 0 & 1 & 1 & 0 & 1 & 1 & 0 & 0 & 0 & 1 \\
 0 & 1 & 0 & 0 & 0 & 0 & 0 & 0 & 1 & 1 & 1 & 1 & 0 & 0 & 1 & 0 & 1 & 1 & 0 & 0 & 0 \\
 0 & 0 & 1 & 1 & 0 & 0 & 0 & 1 & 1 & 0 & 0 & 1 & 0 & 0 & 1 & 0 & 0 & 1 & 1 & 0 & 0 \\
 0 & 0 & 1 & 0 & 0 & 1 & 0 & 0 & 1 & 0 & 1 & 0 & 1 & 1 & 0 & 0 & 0 & 0 & 1 & 0 & 1 \\
 0 & 0 & 1 & 0 & 0 & 0 & 1 & 1 & 0 & 1 & 1 & 1 & 0 & 0 & 0 & 0 & 1 & 0 & 0 & 0 & 1 \\
 0 & 0 & 0 & 1 & 1 & 1 & 0 & 0 & 0 & 0 & 0 & 1 & 1 & 0 & 0 & 1 & 0 & 0 & 1 & 1 & 0 \\
 0 & 0 & 0 & 1 & 1 & 0 & 1 & 1 & 0 & 0 & 1 & 0 & 1 & 1 & 0 & 0 & 0 & 0 & 0 & 0 & 1 \\
 0 & 0 & 0 & 1 & 0 & 0 & 1 & 1 & 0 & 1 & 0 & 0 & 0 & 1 & 1 & 0 & 1 & 0 & 0 & 1 & 0 \\
 0 & 0 & 0 & 0 & 1 & 1 & 1 & 0 & 1 & 1 & 1 & 0 & 0 & 0 & 0 & 0 & 1 & 0 & 1 & 0 & 0 \\
 0 & 0 & 0 & 0 & 1 & 1 & 1 & 0 & 1 & 0 & 0 & 1 & 0 & 0 & 1 & 1 & 0 & 1 & 0 & 0 & 0
\end{smallmatrix}\right),\
A_1=\left(\begin{smallmatrix}
 2 & 1 & 1 & 1 & 1 & 1 & 1 & 1 & 1 & 0 & 0 & 0 & 0 & 0 & 0 & 0 & 0 & 0 & 0 & 0 & 0 \\
 1 & 2 & 1 & 1 & 1 & 0 & 0 & 0 & 0 & 1 & 1 & 1 & 1 & 0 & 0 & 0 & 0 & 0 & 0 & 0 & 0 \\
 1 & 1 & 2 & 0 & 0 & 1 & 1 & 0 & 0 & 1 & 0 & 0 & 0 & 1 & 1 & 1 & 0 & 0 & 0 & 0 & 0 \\
 1 & 1 & 0 & 2 & 0 & 1 & 0 & 0 & 0 & 0 & 1 & 0 & 0 & 1 & 0 & 0 & 1 & 1 & 1 & 0 & 0 \\
 1 & 1 & 0 & 0 & 2 & 0 & 0 & 1 & 1 & 0 & 0 & 0 & 0 & 0 & 0 & 0 & 1 & 1 & 0 & 1 & 1 \\
 1 & 0 & 1 & 1 & 0 & 2 & 0 & 0 & 0 & 0 & 0 & 1 & 0 & 0 & 1 & 0 & 1 & 0 & 0 & 1 & 1 \\
 1 & 0 & 1 & 0 & 0 & 0 & 2 & 0 & 0 & 0 & 1 & 0 & 0 & 0 & 0 & 1 & 0 & 1 & 1 & 1 & 1 \\
 1 & 0 & 0 & 0 & 1 & 0 & 0 & 2 & 0 & 1 & 0 & 1 & 0 & 1 & 0 & 1 & 0 & 1 & 1 & 0 & 0 \\
 1 & 0 & 0 & 0 & 1 & 0 & 0 & 0 & 2 & 1 & 0 & 0 & 1 & 1 & 1 & 0 & 0 & 0 & 0 & 1 & 1 \\
 0 & 1 & 1 & 0 & 0 & 0 & 0 & 1 & 1 & 2 & 0 & 0 & 1 & 0 & 0 & 1 & 0 & 0 & 1 & 1 & 0 \\
 0 & 1 & 0 & 1 & 0 & 0 & 1 & 0 & 0 & 0 & 2 & 0 & 1 & 0 & 1 & 1 & 0 & 1 & 0 & 1 & 0 \\
 0 & 1 & 0 & 0 & 0 & 1 & 0 & 1 & 0 & 0 & 0 & 2 & 1 & 1 & 0 & 1 & 1 & 0 & 0 & 0 & 1 \\
 0 & 1 & 0 & 0 & 0 & 0 & 0 & 0 & 1 & 1 & 1 & 1 & 2 & 0 & 1 & 0 & 1 & 1 & 0 & 0 & 0 \\
 0 & 0 & 1 & 1 & 0 & 0 & 0 & 1 & 1 & 0 & 0 & 1 & 0 & 2 & 1 & 0 & 0 & 1 & 1 & 0 & 0 \\
 0 & 0 & 1 & 0 & 0 & 1 & 0 & 0 & 1 & 0 & 1 & 0 & 1 & 1 & 2 & 0 & 0 & 0 & 1 & 0 & 1 \\
 0 & 0 & 1 & 0 & 0 & 0 & 1 & 1 & 0 & 1 & 1 & 1 & 0 & 0 & 0 & 2 & 1 & 0 & 0 & 0 & 1 \\
 0 & 0 & 0 & 1 & 1 & 1 & 0 & 0 & 0 & 0 & 0 & 1 & 1 & 0 & 0 & 1 & 2 & 0 & 1 & 1 & 0 \\
 0 & 0 & 0 & 1 & 1 & 0 & 1 & 1 & 0 & 0 & 1 & 0 & 1 & 1 & 0 & 0 & 0 & 2 & 0 & 0 & 1 \\
 0 & 0 & 0 & 1 & 0 & 0 & 1 & 1 & 0 & 1 & 0 & 0 & 0 & 1 & 1 & 0 & 1 & 0 & 2 & 1 & 0 \\
 0 & 0 & 0 & 0 & 1 & 1 & 1 & 0 & 1 & 1 & 1 & 0 & 0 & 0 & 0 & 0 & 1 & 0 & 1 & 2 & 0 \\
 0 & 0 & 0 & 0 & 1 & 1 & 1 & 0 & 1 & 0 & 0 & 1 & 0 & 0 & 1 & 1 & 0 & 1 & 0 & 0 & 2
\end{smallmatrix}\right),\notag{}\\
 A_1^2&=\left(\begin{smallmatrix}
   12 & 7 & 7 & 6 & 7 & 6 & 5 & 5 & 5 & 4 & 3 & 3 & 2 & 4 & 3 & 3 & 3 & 4 & 3 & 4 & 4 \\
 7 & 12 & 6 & 6 & 5 & 4 & 3 & 4 & 4 & 6 & 6 & 5 & 7 & 3 & 3 & 4 & 4 & 4 & 2 & 3 & 2 \\
 7 & 6 & 12 & 4 & 2 & 6 & 6 & 4 & 4 & 6 & 4 & 4 & 3 & 5 & 6 & 6 & 2 & 2 & 4 & 3 & 4 \\
 6 & 6 & 4 & 12 & 4 & 6 & 4 & 4 & 2 & 2 & 6 & 4 & 4 & 6 & 4 & 2 & 6 & 6 & 6 & 4 & 2 \\
 7 & 5 & 2 & 4 & 12 & 4 & 4 & 6 & 7 & 4 & 3 & 4 & 4 & 3 & 2 & 3 & 5 & 6 & 3 & 6 & 6 \\
 6 & 4 & 6 & 6 & 4 & 12 & 4 & 2 & 4 & 2 & 3 & 6 & 3 & 4 & 6 & 4 & 7 & 2 & 4 & 5 & 6 \\
 5 & 3 & 6 & 4 & 4 & 4 & 12 & 4 & 3 & 4 & 7 & 2 & 2 & 3 & 4 & 7 & 3 & 6 & 5 & 6 & 6 \\
 5 & 4 & 4 & 4 & 6 & 2 & 4 & 12 & 4 & 6 & 2 & 6 & 3 & 7 & 2 & 6 & 4 & 6 & 6 & 3 & 4 \\
 5 & 4 & 4 & 2 & 7 & 4 & 3 & 4 & 12 & 6 & 3 & 3 & 6 & 5 & 7 & 2 & 3 & 4 & 4 & 6 & 6 \\
 4 & 6 & 6 & 2 & 4 & 2 & 4 & 6 & 6 & 12 & 4 & 4 & 6 & 4 & 4 & 6 & 4 & 2 & 6 & 6 & 2 \\
 3 & 6 & 4 & 6 & 3 & 3 & 7 & 2 & 3 & 4 & 12 & 3 & 7 & 3 & 5 & 5 & 4 & 7 & 4 & 5 & 4 \\
 3 & 5 & 4 & 4 & 4 & 6 & 2 & 6 & 3 & 4 & 3 & 12 & 6 & 5 & 4 & 7 & 7 & 4 & 3 & 2 & 6 \\
 2 & 7 & 3 & 4 & 4 & 3 & 2 & 3 & 6 & 6 & 7 & 6 & 12 & 4 & 6 & 4 & 5 & 5 & 3 & 4 & 4 \\
 4 & 3 & 5 & 6 & 3 & 4 & 3 & 7 & 5 & 4 & 3 & 5 & 4 & 12 & 7 & 3 & 3 & 6 & 7 & 2 & 4 \\
 3 & 3 & 6 & 4 & 2 & 6 & 4 & 2 & 7 & 4 & 5 & 4 & 6 & 7 & 12 & 3 & 3 & 4 & 5 & 4 & 6 \\
 3 & 4 & 6 & 2 & 3 & 4 & 7 & 6 & 2 & 6 & 5 & 7 & 4 & 3 & 3 & 12 & 5 & 4 & 4 & 4 & 6 \\
 3 & 4 & 2 & 6 & 5 & 7 & 3 & 4 & 3 & 4 & 4 & 7 & 5 & 3 & 3 & 5 & 12 & 3 & 6 & 7 & 4 \\
 4 & 4 & 2 & 6 & 6 & 2 & 6 & 6 & 4 & 2 & 7 & 4 & 5 & 6 & 4 & 4 & 3 & 12 & 4 & 3 & 6 \\
 3 & 2 & 4 & 6 & 3 & 4 & 5 & 6 & 4 & 6 & 4 & 3 & 3 & 7 & 5 & 4 & 6 & 4 & 12 & 7 & 2 \\
 4 & 3 & 3 & 4 & 6 & 5 & 6 & 3 & 6 & 6 & 5 & 2 & 4 & 2 & 4 & 4 & 7 & 3 & 7 & 12 & 4 \\
 4 & 2 & 4 & 2 & 6 & 6 & 6 & 4 & 6 & 2 & 4 & 6 & 4 & 4 & 6 & 6 & 4 & 6 & 2 & 4 & 12
 \end{smallmatrix}\right)\,,\notag{}\\
  A_2&=\left(\begin{smallmatrix}
  1 & 2 & 2 & 3 & 2 & 3 & 4 & 4 & 4 & 5 & 6 & 6 & 7 & 5 & 6 & 6 & 6 & 5 & 6 & 5 & 5 \\
 2 & 1 & 3 & 3 & 4 & 5 & 6 & 5 & 5 & 3 & 3 & 4 & 2 & 6 & 6 & 5 & 5 & 5 & 7 & 6 & 7 \\
 2 & 3 & 1 & 5 & 7 & 3 & 3 & 5 & 5 & 3 & 5 & 5 & 6 & 4 & 3 & 3 & 7 & 7 & 5 & 6 & 5 \\
 3 & 3 & 5 & 1 & 5 & 3 & 5 & 5 & 7 & 7 & 3 & 5 & 5 & 3 & 5 & 7 & 3 & 3 & 3 & 5 & 7 \\
 2 & 4 & 7 & 5 & 1 & 5 & 5 & 3 & 2 & 5 & 6 & 5 & 5 & 6 & 7 & 6 & 4 & 3 & 6 & 3 & 3 \\
 3 & 5 & 3 & 3 & 5 & 1 & 5 & 7 & 5 & 7 & 6 & 3 & 6 & 5 & 3 & 5 & 2 & 7 & 5 & 4 & 3 \\
 4 & 6 & 3 & 5 & 5 & 5 & 1 & 5 & 6 & 5 & 2 & 7 & 7 & 6 & 5 & 2 & 6 & 3 & 4 & 3 & 3 \\
 4 & 5 & 5 & 5 & 3 & 7 & 5 & 1 & 5 & 3 & 7 & 3 & 6 & 2 & 7 & 3 & 5 & 3 & 3 & 6 & 5 \\
 4 & 5 & 5 & 7 & 2 & 5 & 6 & 5 & 1 & 3 & 6 & 6 & 3 & 4 & 2 & 7 & 6 & 5 & 5 & 3 & 3 \\
 5 & 3 & 3 & 7 & 5 & 7 & 5 & 3 & 3 & 1 & 5 & 5 & 3 & 5 & 5 & 3 & 5 & 7 & 3 & 3 & 7 \\
 6 & 3 & 5 & 3 & 6 & 6 & 2 & 7 & 6 & 5 & 1 & 6 & 2 & 6 & 4 & 4 & 5 & 2 & 5 & 4 & 5 \\
 6 & 4 & 5 & 5 & 5 & 3 & 7 & 3 & 6 & 5 & 6 & 1 & 3 & 4 & 5 & 2 & 2 & 5 & 6 & 7 & 3 \\
 7 & 2 & 6 & 5 & 5 & 6 & 7 & 6 & 3 & 3 & 2 & 3 & 1 & 5 & 3 & 5 & 4 & 4 & 6 & 5 & 5 \\
 5 & 6 & 4 & 3 & 6 & 5 & 6 & 2 & 4 & 5 & 6 & 4 & 5 & 1 & 2 & 6 & 6 & 3 & 2 & 7 & 5 \\
 6 & 6 & 3 & 5 & 7 & 3 & 5 & 7 & 2 & 5 & 4 & 5 & 3 & 2 & 1 & 6 & 6 & 5 & 4 & 5 & 3 \\
 6 & 5 & 3 & 7 & 6 & 5 & 2 & 3 & 7 & 3 & 4 & 2 & 5 & 6 & 6 & 1 & 4 & 5 & 5 & 5 & 3 \\
 6 & 5 & 7 & 3 & 4 & 2 & 6 & 5 & 6 & 5 & 5 & 2 & 4 & 6 & 6 & 4 & 1 & 6 & 3 & 2 & 5 \\
 5 & 5 & 7 & 3 & 3 & 7 & 3 & 3 & 5 & 7 & 2 & 5 & 4 & 3 & 5 & 5 & 6 & 1 & 5 & 6 & 3 \\
 6 & 7 & 5 & 3 & 6 & 5 & 4 & 3 & 5 & 3 & 5 & 6 & 6 & 2 & 4 & 5 & 3 & 5 & 1 & 2 & 7 \\
 5 & 6 & 6 & 5 & 3 & 4 & 3 & 6 & 3 & 3 & 4 & 7 & 5 & 7 & 5 & 5 & 2 & 6 & 2 & 1 & 5 \\
 5 & 7 & 5 & 7 & 3 & 3 & 3 & 5 & 3 & 7 & 5 & 3 & 5 & 5 & 3 & 3 & 5 & 3 & 7 & 5 & 1
 \end{smallmatrix}\right)\,,\label{a2}\\
  \hat{A}&=\left(\begin{smallmatrix}
  1 & 2 & 3 & 4 & 5 & 6 & 7 & 8 & 9 & 10 & 11 & 12 & 13 & 14 & 15 & 16 & 17 & 18 & 19 & 20 & 21 \\
 2 & 22 & 23 & 24 & 25 & 26 & 27 & 28 & 29 & 30 & 31 & 32 & 33 & 19 & 34 & 35 & 36 & 37 & 38 & 39 & 40 \\
 3 & 23 & 41 & 42 & 43 & 44 & 45 & 46 & 47 & 48 & 49 & 50 & 51 & 8 & 52 & 53 & 54 & 55 & 28 & 56 & 57 \\
 4 & 24 & 42 & 58 & 59 & 60 & 61 & 42 & 62 & 63 & 64 & 61 & 65 & 4 & 59 & 66 & 64 & 60 & 24 & 65 & 67 \\
 5 & 25 & 43 & 59 & 68 & 69 & 70 & 52 & 71 & 72 & 73 & 74 & 75 & 15 & 76 & 77 & 78 & 79 & 34 & 80 & 81 \\
 6 & 26 & 44 & 60 & 69 & 82 & 83 & 55 & 84 & 85 & 86 & 87 & 88 & 18 & 79 & 89 & 90 & 91 & 37 & 92 & 93 \\
 7 & 27 & 45 & 61 & 70 & 83 & 94 & 50 & 95 & 96 & 97 & 98 & 99 & 12 & 74 & 100 & 101 & 87 & 32 & 102 & 103 \\
 8 & 28 & 46 & 42 & 52 & 55 & 50 & 41 & 47 & 48 & 54 & 45 & 56 & 3 & 43 & 53 & 49 & 44 & 23 & 51 & 57 \\
 9 & 29 & 47 & 62 & 71 & 84 & 95 & 47 & 104 & 105 & 106 & 95 & 107 & 9 & 71 & 108 & 106 & 84 & 29 & 107 & 109 \\
 10 & 30 & 48 & 63 & 72 & 85 & 96 & 48 & 105 & 110 & 111 & 96 & 112 & 10 & 72 & 113 & 111 & 85 & 30 & 112 & 114 \\
 11 & 31 & 49 & 64 & 73 & 86 & 97 & 54 & 106 & 111 & 115 & 101 & 116 & 17 & 78 & 117 & 118 & 90 & 36 & 119 & 120 \\
 12 & 32 & 50 & 61 & 74 & 87 & 98 & 45 & 95 & 96 & 101 & 94 & 102 & 7 & 70 & 100 & 97 & 83 & 27 & 99 & 103 \\
 13 & 33 & 51 & 65 & 75 & 88 & 99 & 56 & 107 & 112 & 116 & 102 & 121 & 20 & 80 & 122 & 119 & 92 & 39 & 123 & 124 \\
 14 & 19 & 8 & 4 & 15 & 18 & 12 & 3 & 9 & 10 & 17 & 7 & 20 & 1 & 5 & 16 & 11 & 6 & 2 & 13 & 21 \\
 15 & 34 & 52 & 59 & 76 & 79 & 74 & 43 & 71 & 72 & 78 & 70 & 80 & 5 & 68 & 77 & 73 & 69 & 25 & 75 & 81 \\
 16 & 35 & 53 & 66 & 77 & 89 & 100 & 53 & 108 & 113 & 117 & 100 & 122 & 16 & 77 & 125 & 117 & 89 & 35 & 122 & 126 \\
 17 & 36 & 54 & 64 & 78 & 90 & 101 & 49 & 106 & 111 & 118 & 97 & 119 & 11 & 73 & 117 & 115 & 86 & 31 & 116 & 120 \\
 18 & 37 & 55 & 60 & 79 & 91 & 87 & 44 & 84 & 85 & 90 & 83 & 92 & 6 & 69 & 89 & 86 & 82 & 26 & 88 & 93 \\
 19 & 38 & 28 & 24 & 34 & 37 & 32 & 23 & 29 & 30 & 36 & 27 & 39 & 2 & 25 & 35 & 31 & 26 & 22 & 33 & 40 \\
 20 & 39 & 56 & 65 & 80 & 92 & 102 & 51 & 107 & 112 & 119 & 99 & 123 & 13 & 75 & 122 & 116 & 88 & 33 & 121 & 124 \\
 21 & 40 & 57 & 67 & 81 & 93 & 103 & 57 & 109 & 114 & 120 & 103 & 124 & 21 & 81 & 126 & 120 & 93 & 40 & 124 & 127
 \end{smallmatrix}\right)\,.\label{NuSaS}%
\end{align}%

\noindent\textbf{Suggestions to implementation of SaS process:} While SaS process is just a process of square-and-substitution process as shown in \eqref{eq:wlseq}, it is easy to implement with adjacency matrix to obtain a stable graph. Although it is theoretically using variables as labels, we suggest a procedure named as the numerical SaS (NuSaS for short) procedure.
In NuSaS procedure, use normal 0-1 adjacency matrix of a graph as the starting matrix. In the substitution part, one may use integers as labels to perform ff-equivalent number substitution (or any equivalent number substitution) to the square of a matrix. The stable graph obtained will be a pseudo-stable graph. In most of time, a pseudo-stable graph is with the same vertex partition as that from a genuine stable graph. It could happen a faulty as illustrated previously in this example.

In order to avoid the faulty stable (like $A_2$ as above in \eqref{a2}), whenever a stable graph is obtained by NuSaS procedure, perform another and a different equivalent number substitution to the pseudo-stable graph to obtain a graph, and compute and obtain the second pseudo-stable graph. If the vertex partition obtained from the second stable graph is the same as the one from the first stable graph, it is most \emph{possible} that you get a genuine stable partition. Another option is to perform a row-column permutation (and to permute back after computation) to the pseudo-stable graph and then to compute with NuSaS process, it will get a different partition with overwhelming probability if the partition from  previous pseudo-stable graph is not the correct one. The \emph{safe option} is to perform an equivalent variable substitution to the pseudo-stable graph, and then to evaluate and obtain the stable graph.\ep

2. The second example tends to illustrate the strongly equitable partition by the stable graph of a graph $X$ (cf. \eqref{gX}) obtained by SaS process or WL process.

      Graph $X$ is a graph of order $24$. The stable graph $\hat{X}$ obtained by SaS process and the stable graph $\WL(X)$ obtained by WL process are shown, respectively, in \eqref{sasX} and \eqref{WLX}.

      We remind that the stable graphs $\hat{X}$ and $\WL(X)$ do not give the automorphism partition of graph $X$ in the example. The stable graphs of the binding graph $B$ of graph $X$ successfully partition the vertices into orbits of the automorphism group of graph $X$.
      One may find the followings.

        (a) The stable graph $\hat{X}$ is a symmetric matrix (a labeled complete graph), while the stable graph $\WL(X)$ is not a symmetric one. However, $\WL(X)$ respects the converse equivalent property. (cf. Proposition \ref{prop:sqdiomnd} and the contexts above it)

        (b) One may find that the labels in each block (colored or not) appear only in the block or corresponding transposed part (in $\hat{X}$). And any two rows (or columns) within a block have the same labels as multisets. That indicts the strongly equitable partitions.

       (c) The vertex partitions by stable graph $\hat{X}$ and $\WL(L)$ are the same as $$\{\{1,2,3,4,5,6,7,8\},\{9,10,11,12,13,14,15,16\},\{17,18,19,20\},\{21,22,23,24\}\}$$
            (cf. Theorem \ref{thm:saswl})
            The $\WL(X)$ is finer in the partition of edges.

      (d) One may confirm that (cf. \eqref{eq:wdteq})
       $\hat{X}\approx\mathtt{WL}(X)+\big(\mathtt{WL}(X)\big)^\T\,.$

     (e) The labels of vertices appear only on the diagonal of $\hat{X}$ and $\WL(X)$, respectively. (This indicates the stable graph recognizes vertices.)

     (f) The matrix $X*\hat{X}$ is the Hadamard (components) product \eqref{XsasX} of $X$ and $\hat{X}$. The nonzero entries in this product are the labels to edges in $X$. One may check that these nonzero labels do not appear as other entries in $\hat{X}$. Similarly to $X*\WL(X)$ in \eqref{XWLX} and $\WL(X)$. Which means the stable graph recognizes edges.

  \begin{align}
  &X=\left(\begin{smallmatrix}
 0& 0& 0& 0& 1& 1& 0& 0& 1& 1& 0& 0& 1& 1& 0& 0& 1& 1& 0& 0& 1& 1& 0&
   0\\0& 0& 0& 0& 0& 0& 1& 1& 1& 1& 0& 0& 0& 0& 1& 1& 0& 0& 1& 1& 0&
   0& 1& 1\\0& 0& 0& 0& 0& 0& 1& 1& 0& 0& 1& 1& 1& 1& 0& 0& 1& 1& 0&
   0& 0& 0& 1& 1\\0& 0& 0& 0& 1& 1& 0& 0& 0& 0& 1& 1& 0& 0& 1& 1& 0&
   0& 1& 1& 1& 1& 0& 0\\1& 0& 0& 1& 0& 0& 0& 0& 1& 0& 1& 0& 1& 0& 0&
   1& 0& 1& 0& 1& 0& 0& 1& 1\\1& 0& 0& 1& 0& 0& 0& 0& 0& 1& 0& 1& 0&
   1& 1& 0& 1& 0& 1& 0& 0& 0& 1& 1\\0& 1& 1& 0& 0& 0& 0& 0& 1& 0& 1&
   0& 0& 1& 1& 0& 0& 1& 0& 1& 1& 1& 0& 0\\0& 1& 1& 0& 0& 0& 0& 0& 0&
   1& 0& 1& 1& 0& 0& 1& 1& 0& 1& 0& 1& 1& 0& 0\\1& 1& 0& 0& 1& 0& 1&
   0& 0& 1& 1& 0& 0& 0& 0& 0& 1& 0& 1& 0& 1& 0& 1& 0\\1& 1& 0& 0& 0&
   1& 0& 1& 1& 0& 0& 1& 0& 0& 0& 0& 0& 1& 0& 1& 0& 1& 0& 1\\0& 0& 1&
   1& 1& 0& 1& 0& 1& 0& 0& 1& 0& 0& 0& 0& 1& 0& 1& 0& 0& 1& 0& 1\\0&
   0& 1& 1& 0& 1& 0& 1& 0& 1& 1& 0& 0& 0& 0& 0& 0& 1& 0& 1& 1& 0& 1&
  0\\1& 0& 1& 0& 1& 0& 0& 1& 0& 0& 0& 0& 0& 1& 0& 1& 0& 0& 1& 1& 1&
  0& 1& 0\\1& 0& 1& 0& 0& 1& 1& 0& 0& 0& 0& 0& 1& 0& 1& 0& 0& 0& 1&
  1& 0& 1& 0& 1\\0& 1& 0& 1& 0& 1& 1& 0& 0& 0& 0& 0& 0& 1& 0& 1& 1&
  1& 0& 0& 1& 0& 1& 0\\0& 1& 0& 1& 1& 0& 0& 1& 0& 0& 0& 0& 1& 0& 1&
  0& 1& 1& 0& 0& 0& 1& 0& 1\\1& 0& 1& 0& 0& 1& 0& 1& 1& 0& 1& 0& 0&
  0& 1& 1& 0& 0& 0& 0& 1& 0& 0& 1\\1& 0& 1& 0& 1& 0& 1& 0& 0& 1& 0&
  1& 0& 0& 1& 1& 0& 0& 0& 0& 0& 1& 1& 0\\0& 1& 0& 1& 0& 1& 0& 1& 1&
  0& 1& 0& 1& 1& 0& 0& 0& 0& 0& 0& 0& 1& 1& 0\\0& 1& 0& 1& 1& 0& 1&
  0& 0& 1& 0& 1& 1& 1& 0& 0& 0& 0& 0& 0& 1& 0& 0& 1\\1& 0& 0& 1& 0&
  0& 1& 1& 1& 0& 0& 1& 1& 0& 1& 0& 1& 0& 0& 1& 0& 0& 0& 0\\1& 0& 0&
  1& 0& 0& 1& 1& 0& 1& 1& 0& 0& 1& 0& 1& 0& 1& 1& 0& 0& 0& 0& 0\\0&
  1& 1& 0& 1& 1& 0& 0& 1& 0& 0& 1& 1& 0& 1& 0& 0& 1& 1& 0& 0& 0& 0&
  0\\0& 1& 1& 0& 1& 1& 0& 0& 0& 1& 1& 0& 0& 1& 0& 1& 1& 0& 0& 1& 0&
  0& 0& 0
  \end{smallmatrix}\right),\label{gX}\\
 &\hat{X}=\left(\begin{smallmatrix}
 \tcb{1}& \tcb{2}& \tcb{3}& \tcb{4}& \tcb{5}& \tcb{5}& \tcb{6}& \tcb{6}& 7& 7& 8& 8& 9& 9& 10& 10& \tcr{11}& \tcr{11}& \tcr{12}& \tcr{12}&
  13& 13& 14& 14\\
  \tcb{2}& \tcb{1}& \tcb{4}& \tcb{3}& \tcb{6}& \tcb{6}& \tcb{5}& \tcb{5}& 7& 7& 8& 8& 10& 10& 9& 9&
  \tcr{12}& \tcr{12}& \tcr{11}& \tcr{11}& 14& 14& 13& 13\\
  \tcb{3}& \tcb{4}& \tcb{1}& \tcb{2}& \tcb{6}& \tcb{6}& \tcb{5}& \tcb{5}& 8& 8& 7&
  7& 9& 9& 10& 10& \tcr{11}& \tcr{11}& \tcr{12}& \tcr{12}& 14& 14& 13& 13\\
  \tcb{4}& \tcb{3}& \tcb{2}& \tcb{1}& \tcb{5}& \tcb{5}&
   \tcb{6}& \tcb{6}& 8& 8& 7& 7& 10& 10& 9& 9& \tcr{12}& \tcr{12}& \tcr{11}& \tcr{11}& 13& 13& 14&
  14\\\tcb{5}& \tcb{6}& \tcb{6}& \tcb{5}& \tcb{1}& \tcb{4}& \tcb{3}& \tcb{2}& 9& 10& 9& 10& 7& 8& 8& 7& \tcr{12}& \tcr{11}& \tcr{12}&
  \tcr{11}& 14& 14& 13& 13\\
  \tcb{5}& \tcb{6}& \tcb{6}& \tcb{5}& \tcb{4}& \tcb{1}& \tcb{2}& \tcb{3}& 10& 9& 10& 9& 8& 7& 7&
   8& \tcr{11}& \tcr{12}& \tcr{11}& \tcr{12}& 14& 14& 13& 13\\
   \tcb{6}& \tcb{5}& \tcb{5}& \tcb{6}& \tcb{3}& \tcb{2}& \tcb{1}& \tcb{4}& 9& 10&
   9& 10& 8& 7& 7& 8& \tcr{12}& \tcr{11}& \tcr{12}& \tcr{11}& 13& 13& 14& 14\\
   \tcb{6}& \tcb{5}& \tcb{5}& \tcb{6}& \tcb{2}&
   \tcb{3}& \tcb{4}& \tcb{1}& 10& 9& 10& 9& 7& 8& 8& 7& \tcr{11}& \tcr{12}& \tcr{11}& \tcr{12}& 13& 13& 14&
  14\\
  7& 7& 8& 8& 9& 10& 9& 10& \tcb{15}& \tcb{16}& \tcb{17}& \tcb{18}& \tcb{19}& \tcb{20}& \tcb{19}& \tcb{20}& 21&
  22& 21& 22& \tcr{23}& \tcr{24}& \tcr{23}& \tcr{24}\\
  7& 7& 8& 8& 10& 9& 10& 9& \tcb{16}& \tcb{15}& \tcb{18}&
  \tcb{17}& \tcb{20}& \tcb{19}& \tcb{20}& \tcb{19}& 22& 21& 22& 21& \tcr{24}& \tcr{23}& \tcr{24}& \tcr{23}\\
  8& 8& 7& 7& 9&
   10& 9& 10& \tcb{17}& \tcb{18}& \tcb{15}& \tcb{16}& \tcb{20}& \tcb{19}& \tcb{20}& \tcb{19}& 21& 22& 21& 22& \tcr{24}& \tcr{23}&
  \tcr{24}& \tcr{23}\\
  8& 8& 7& 7& 10& 9& 10& 9& \tcb{18}& \tcb{17}& \tcb{16}& \tcb{15}& \tcb{19}& \tcb{20}& \tcb{19}& \tcb{20}&
  22& 21& 22& 21& \tcr{23}& \tcr{24}& \tcr{23}& \tcr{24}\\
  9& 10& 9& 10& 7& 8& 8& 7& \tcb{19}& \tcb{20}&
  \tcb{20}& \tcb{19}& \tcb{15}& \tcb{17}& \tcb{18}& \tcb{16}& 22& 22& 21& 21& \tcr{23}& \tcr{24}& \tcr{23}& \tcr{24}\\
  9& 10& 9&
  10& 8& 7& 7& 8& \tcb{20}& \tcb{19}& \tcb{19}& \tcb{20}& \tcb{17}& \tcb{15}& \tcb{16}& \tcb{18}& 22& 22& 21& 21& \tcr{24}&
  \tcr{23}& \tcr{24}& \tcr{23}\\
  10& 9& 10& 9& 8& 7& 7& 8& \tcb{19}& \tcb{20}& \tcb{20}& \tcb{19}& \tcb{18}& \tcb{16}& \tcb{15}&
  \tcb{17}& 21& 21& 22& 22& \tcr{23}& \tcr{24}& \tcr{23}& \tcr{24}\\
  10& 9& 10& 9& 7& 8& 8& 7& \tcb{20}&
  \tcb{19}& \tcb{19}& \tcb{20}& \tcb{16}& \tcb{18}& \tcb{17}& \tcb{15}& 21& 21& 22& 22& \tcr{24}& \tcr{23}& \tcr{24}& \tcr{23}\\
  11&
  12& 11& 12& 12& 11& 12& 11& 21& 22& 21& 22& 22& 22& 21& 21& \tcb{25}& \tcb{26}&
  \tcb{26}& \tcb{27}& 28& 29& 29& 28\\
  11& 12& 11& 12& 11& 12& 11& 12& 22& 21&
  22& 21& 22& 22& 21& 21& \tcb{26}& \tcb{25}& \tcb{27}& \tcb{26}& 29& 28& 28& 29\\12& 11&
  12& 11& 12& 11& 12& 11& 21& 22& 21& 22& 21& 21& 22& 22& \tcb{26}& \tcb{27}& \tcb{25}&
  \tcb{26}& 29& 28& 28& 29\\12& 11& 12& 11& 11& 12& 11& 12& 22& 21& 22&
  21& 21& 21& 22& 22& \tcb{27}& \tcb{26}& \tcb{26}& \tcb{25}& 28& 29& 29& 28\\
  13& 14& 14&
  13& 14& 14& 13& 13& 23& 24& 24& 23& 23& 24& 23& 24& 28& 29& 29& 28&
  \tcb{30}& \tcb{31}& \tcb{32}& \tcb{33}\\13& 14& 14& 13& 14& 14& 13& 13& 24& 23& 23& 24&
  24& 23& 24& 23& 29& 28& 28& 29& \tcb{31}& \tcb{30}& \tcb{33}& \tcb{32}\\14& 13& 13& 14&
  13& 13& 14& 14& 23& 24& 24& 23& 23& 24& 23& 24& 29& 28& 28& 29& \tcb{32}&
  \tcb{33}& \tcb{30}& \tcb{31}\\14& 13& 13& 14& 13& 13& 14& 14& 24& 23& 23& 24& 24&
  23& 24& 23& 28& 29& 29& 28& \tcb{33}& \tcb{32}& \tcb{31}& \tcb{30}
   \end{smallmatrix}\right),\label{sasX}\\
 &\WL(X)=\left(\begin{smallmatrix}
 \tcb{1}&  \tcb{2}&  \tcb{3}&  \tcb{4}&  \tcb{5}&  \tcb{5}&  \tcb{6}&  \tcb{6}& 7& 7& 8& 8& 9& 9& 10& 10& \tcr{11}& \tcr{11}& \tcr{12}& \tcr{12}&
  13& 13& 14& 14\\
   \tcb{2}&  \tcb{1}&  \tcb{4}&  \tcb{3}&  \tcb{6}&  \tcb{6}&  \tcb{5}&  \tcb{5}& 7& 7& 8& 8& 10& 10& 9& 9&
  \tcr{12}& \tcr{12}& \tcr{11}& \tcr{11}& 14& 14& 13& 13\\
   \tcb{3}&  \tcb{4}&  \tcb{1}&  \tcb{2}&  \tcb{6}&  \tcb{6}&  \tcb{5}&  \tcb{5}& 8& 8& 7&
  7& 9& 9& 10& 10& \tcr{11}& \tcr{11}& \tcr{12}& \tcr{12}& 14& 14& 13& 13\\
   \tcb{4}&  \tcb{3}&  \tcb{2}&  \tcb{1}&  \tcb{5}&  \tcb{5}&
    \tcb{6}&  \tcb{6}& 8& 8& 7& 7& 10& 10& 9& 9& \tcr{12}& \tcr{12}& \tcr{11}& \tcr{11}& 13& 13& 14&
  14\\
   \tcb{5}&  \tcb{6}&  \tcb{6}&  \tcb{5}&  \tcb{1}&  \tcb{4}&  \tcb{3}&  \tcb{2}& 9& 10& 9& 10& 7& 8& 8& 7& \tcr{12}& \tcr{11}& \tcr{12}&
  \tcr{11}& 14& 14& 13& 13\\
   \tcb{5}&  \tcb{6}&  \tcb{6}&  \tcb{5}&  \tcb{4}&  \tcb{1}&  \tcb{2}&  \tcb{3}& 10& 9& 10& 9& 8& 7& 7&
   8& \tcr{11}& \tcr{12}& \tcr{11}& \tcr{12}& 14& 14& 13& 13\\
    \tcb{6}&  \tcb{5}&  \tcb{5}&  \tcb{6}&  \tcb{3}&  \tcb{2}&  \tcb{1}& \tcb{4}& 9& 10&
   9& 10& 8& 7& 7& 8& \tcr{12}& \tcr{11}& \tcr{12}& \tcr{11}& 13& 13& 14& 14\\
    \tcb{6}&  \tcb{5}&  \tcb{5}&  \tcb{6}&  \tcb{2}&
    \tcb{3}&  \tcb{4}&  \tcb{1}& 10& 9& 10& 9& 7& 8& 8& 7& \tcr{11}& \tcr{12}& \tcr{11}& \tcr{12}& 13& 13& 14&
  14\\15& 15& 16& 16& 17& 18& 17& 18&  \tcb{19}&  \tcb{20}&  \tcb{21}& \tcb{22}&  \tcb{23}&  \tcb{24}&  \tcb{23}&
   \tcb{24}& 25& 26& 25& 26& \tcr{27}& \tcr{28}& \tcr{27}& \tcr{28}\\15& 15& 16& 16& 18& 17& 18&
  17&  \tcb{20}&  \tcb{19}&  \tcb{22}&  \tcb{21}&  \tcb{24}&  \tcb{23}& \tcb{24}&  \tcb{23}& 26& 25& 26& 25& \tcr{28}& \tcr{27}& \tcr{28}&
  \tcr{27}\\
  16& 16& 15& 15& 17& 18& 17& 18&  \tcb{21}&  \tcb{22}&  \tcb{19}&  \tcb{20}&  \tcb{24}&  \tcb{23}&  \tcb{24}&
   \tcb{23}& 25& 26& 25& 26& \tcr{28}& \tcr{27}& \tcr{28}& \tcr{27}\\
   16& 16& 15& 15& 18& 17& 18&
  17&  \tcb{22}&  \tcb{21}&  \tcb{20}&  \tcb{19}&  \tcb{23}&  \tcb{24}& \tcb{23}&  \tcb{24}& 26& 25& 26& 25& \tcr{27}& \tcr{28}& \tcr{27}&
  \tcr{28}\\17& 18& 17& 18& 15& 16& 16& 15&  \tcb{23}&  \tcb{24}&  \tcb{24}&  \tcb{23}&  \tcb{19}&  \tcb{21}&  \tcb{22}&
   \tcb{20}& 26& 26& 25& 25& \tcr{27}& \tcr{28}& \tcr{27}& \tcr{28}\\17& 18& 17& 18& 16& 15& 15&
  16&  \tcb{24}&  \tcb{23}&  \tcb{23}&  \tcb{24}&  \tcb{21}&  \tcb{19}&  \tcb{20}&  \tcb{22}& 26& 26& 25& 25& \tcr{28}& \tcr{27}& \tcr{28}&
  \tcr{27}\\18& 17& 18& 17& 16& 15& 15& 16&  \tcb{23}&  \tcb{24}&  \tcb{24}&  \tcb{23}&  \tcb{22}&  \tcb{20}&  \tcb{19}&
   \tcb{21}& 25& 25& 26& 26& \tcr{27}& \tcr{28}& \tcr{27}& \tcr{28}\\18& 17& 18& 17& 15& 16& 16&
  15&  \tcb{24}&  \tcb{23}&  \tcb{23}&  \tcb{24}&  \tcb{20}&  \tcb{22}&  \tcb{21}&  \tcb{19}& 25& 25& 26& 26& \tcr{28}& \tcr{27}& \tcr{28}&
  \tcr{27}\\
  29& 30& 29& 30& 30& 29& 30& 29& 31& 32& 31& 32& 32& 32& 31&
  31&  \tcb{33}&  \tcb{34}&  \tcb{34}&  \tcb{35}& 36& 37& 37& 36\\29& 30& 29& 30& 29& 30& 29&
  30& 32& 31& 32& 31& 32& 32& 31& 31&  \tcb{34}&  \tcb{33}&  \tcb{35}&  \tcb{34}& 37& 36& 36&
  37\\30& 29& 30& 29& 30& 29& 30& 29& 31& 32& 31& 32& 31& 31& 32&
  32&  \tcb{34}&  \tcb{35}&  \tcb{33}&  \tcb{34}& 37& 36& 36& 37\\30& 29& 30& 29& 29& 30& 29&
  30& 32& 31& 32& 31& 31& 31& 32& 32&  \tcb{35}&  \tcb{34}&  \tcb{34}&  \tcb{33}& 36& 37& 37&
  36\\38& 39& 39& 38& 39& 39& 38& 38& 40& 41& 41& 40& 40& 41& 40&
  41& 42& 43& 43& 42&  \tcb{44}&  \tcb{45}&  \tcb{46}&  \tcb{47}\\
  38& 39& 39& 38& 39& 39& 38&
  38& 41& 40& 40& 41& 41& 40& 41& 40& 43& 42& 42& 43&  \tcb{45}& \tcb{44}&  \tcb{47}&
   \tcb{46}\\39& 38& 38& 39& 38& 38& 39& 39& 40& 41& 41& 40& 40& 41& 40&
  41& 43& 42& 42& 43&  \tcb{46}&  \tcb{47}&  \tcb{44}&  \tcb{45}\\39& 38& 38& 39& 38& 38& 39&
  39& 41& 40& 40& 41& 41& 40& 41& 40& 42& 43& 43& 42&  \tcb{47}& \tcb{46}&  \tcb{45}&  \tcb{44}
 \end{smallmatrix}\right),\label{WLX}\\
 &X*\hat{X}=\left(\begin{smallmatrix}
  0 & 0 & 0 & 0 & 5 & 5 & 0 & 0 & 7 & 7 & 0 & 0 & 9 & 9 & 0 & 0 & 11 & 11 & 0 & 0 & 13 & 13 & 0 & 0 \\
 0 & 0 & 0 & 0 & 0 & 0 & 5 & 5 & 7 & 7 & 0 & 0 & 0 & 0 & 9 & 9 & 0 & 0 & 11 & 11 & 0 & 0 & 13 & 13 \\
 0 & 0 & 0 & 0 & 0 & 0 & 5 & 5 & 0 & 0 & 7 & 7 & 9 & 9 & 0 & 0 & 11 & 11 & 0 & 0 & 0 & 0 & 13 & 13 \\
 0 & 0 & 0 & 0 & 5 & 5 & 0 & 0 & 0 & 0 & 7 & 7 & 0 & 0 & 9 & 9 & 0 & 0 & 11 & 11 & 13 & 13 & 0 & 0 \\
 5 & 0 & 0 & 5 & 0 & 0 & 0 & 0 & 9 & 0 & 9 & 0 & 7 & 0 & 0 & 7 & 0 & 11 & 0 & 11 & 0 & 0 & 13 & 13 \\
 5 & 0 & 0 & 5 & 0 & 0 & 0 & 0 & 0 & 9 & 0 & 9 & 0 & 7 & 7 & 0 & 11 & 0 & 11 & 0 & 0 & 0 & 13 & 13 \\
 0 & 5 & 5 & 0 & 0 & 0 & 0 & 0 & 9 & 0 & 9 & 0 & 0 & 7 & 7 & 0 & 0 & 11 & 0 & 11 & 13 & 13 & 0 & 0 \\
 0 & 5 & 5 & 0 & 0 & 0 & 0 & 0 & 0 & 9 & 0 & 9 & 7 & 0 & 0 & 7 & 11 & 0 & 11 & 0 & 13 & 13 & 0 & 0 \\
 7 & 7 & 0 & 0 & 9 & 0 & 9 & 0 & 0 & 16 & 17 & 0 & 0 & 0 & 0 & 0 & 21 & 0 & 21 & 0 & 23 & 0 & 23 & 0 \\
 7 & 7 & 0 & 0 & 0 & 9 & 0 & 9 & 16 & 0 & 0 & 17 & 0 & 0 & 0 & 0 & 0 & 21 & 0 & 21 & 0 & 23 & 0 & 23 \\
 0 & 0 & 7 & 7 & 9 & 0 & 9 & 0 & 17 & 0 & 0 & 16 & 0 & 0 & 0 & 0 & 21 & 0 & 21 & 0 & 0 & 23 & 0 & 23 \\
 0 & 0 & 7 & 7 & 0 & 9 & 0 & 9 & 0 & 17 & 16 & 0 & 0 & 0 & 0 & 0 & 0 & 21 & 0 & 21 & 23 & 0 & 23 & 0 \\
 9 & 0 & 9 & 0 & 7 & 0 & 0 & 7 & 0 & 0 & 0 & 0 & 0 & 17 & 0 & 16 & 0 & 0 & 21 & 21 & 23 & 0 & 23 & 0 \\
 9 & 0 & 9 & 0 & 0 & 7 & 7 & 0 & 0 & 0 & 0 & 0 & 17 & 0 & 16 & 0 & 0 & 0 & 21 & 21 & 0 & 23 & 0 & 23 \\
 0 & 9 & 0 & 9 & 0 & 7 & 7 & 0 & 0 & 0 & 0 & 0 & 0 & 16 & 0 & 17 & 21 & 21 & 0 & 0 & 23 & 0 & 23 & 0 \\
 0 & 9 & 0 & 9 & 7 & 0 & 0 & 7 & 0 & 0 & 0 & 0 & 16 & 0 & 17 & 0 & 21 & 21 & 0 & 0 & 0 & 23 & 0 & 23 \\
 11 & 0 & 11 & 0 & 0 & 11 & 0 & 11 & 21 & 0 & 21 & 0 & 0 & 0 & 21 & 21 & 0 & 0 & 0 & 0 & 28 & 0 & 0 & 28 \\
 11 & 0 & 11 & 0 & 11 & 0 & 11 & 0 & 0 & 21 & 0 & 21 & 0 & 0 & 21 & 21 & 0 & 0 & 0 & 0 & 0 & 28 & 28 & 0 \\
 0 & 11 & 0 & 11 & 0 & 11 & 0 & 11 & 21 & 0 & 21 & 0 & 21 & 21 & 0 & 0 & 0 & 0 & 0 & 0 & 0 & 28 & 28 & 0 \\
 0 & 11 & 0 & 11 & 11 & 0 & 11 & 0 & 0 & 21 & 0 & 21 & 21 & 21 & 0 & 0 & 0 & 0 & 0 & 0 & 28 & 0 & 0 & 28 \\
 13 & 0 & 0 & 13 & 0 & 0 & 13 & 13 & 23 & 0 & 0 & 23 & 23 & 0 & 23 & 0 & 28 & 0 & 0 & 28 & 0 & 0 & 0 & 0 \\
 13 & 0 & 0 & 13 & 0 & 0 & 13 & 13 & 0 & 23 & 23 & 0 & 0 & 23 & 0 & 23 & 0 & 28 & 28 & 0 & 0 & 0 & 0 & 0 \\
 0 & 13 & 13 & 0 & 13 & 13 & 0 & 0 & 23 & 0 & 0 & 23 & 23 & 0 & 23 & 0 & 0 & 28 & 28 & 0 & 0 & 0 & 0 & 0 \\
 0 & 13 & 13 & 0 & 13 & 13 & 0 & 0 & 0 & 23 & 23 & 0 & 0 & 23 & 0 & 23 & 28 & 0 & 0 & 28 & 0 & 0 & 0 & 0
 \end{smallmatrix}\right),\label{XsasX}\\
 &X*\WL(X)=\left(\begin{smallmatrix}
  0 & 0 & 0 & 0 & 5 & 5 & 0 & 0 & 7 & 7 & 0 & 0 & 9 & 9 & 0 & 0 & 11 & 11 & 0 & 0 & 13 & 13 & 0 & 0 \\
 0 & 0 & 0 & 0 & 0 & 0 & 5 & 5 & 7 & 7 & 0 & 0 & 0 & 0 & 9 & 9 & 0 & 0 & 11 & 11 & 0 & 0 & 13 & 13 \\
 0 & 0 & 0 & 0 & 0 & 0 & 5 & 5 & 0 & 0 & 7 & 7 & 9 & 9 & 0 & 0 & 11 & 11 & 0 & 0 & 0 & 0 & 13 & 13 \\
 0 & 0 & 0 & 0 & 5 & 5 & 0 & 0 & 0 & 0 & 7 & 7 & 0 & 0 & 9 & 9 & 0 & 0 & 11 & 11 & 13 & 13 & 0 & 0 \\
 5 & 0 & 0 & 5 & 0 & 0 & 0 & 0 & 9 & 0 & 9 & 0 & 7 & 0 & 0 & 7 & 0 & 11 & 0 & 11 & 0 & 0 & 13 & 13 \\
 5 & 0 & 0 & 5 & 0 & 0 & 0 & 0 & 0 & 9 & 0 & 9 & 0 & 7 & 7 & 0 & 11 & 0 & 11 & 0 & 0 & 0 & 13 & 13 \\
 0 & 5 & 5 & 0 & 0 & 0 & 0 & 0 & 9 & 0 & 9 & 0 & 0 & 7 & 7 & 0 & 0 & 11 & 0 & 11 & 13 & 13 & 0 & 0 \\
 0 & 5 & 5 & 0 & 0 & 0 & 0 & 0 & 0 & 9 & 0 & 9 & 7 & 0 & 0 & 7 & 11 & 0 & 11 & 0 & 13 & 13 & 0 & 0 \\
 15 & 15 & 0 & 0 & 17 & 0 & 17 & 0 & 0 & 20 & 21 & 0 & 0 & 0 & 0 & 0 & 25 & 0 & 25 & 0 & 27 & 0 & 27 & 0 \\
 15 & 15 & 0 & 0 & 0 & 17 & 0 & 17 & 20 & 0 & 0 & 21 & 0 & 0 & 0 & 0 & 0 & 25 & 0 & 25 & 0 & 27 & 0 & 27 \\
 0 & 0 & 15 & 15 & 17 & 0 & 17 & 0 & 21 & 0 & 0 & 20 & 0 & 0 & 0 & 0 & 25 & 0 & 25 & 0 & 0 & 27 & 0 & 27 \\
 0 & 0 & 15 & 15 & 0 & 17 & 0 & 17 & 0 & 21 & 20 & 0 & 0 & 0 & 0 & 0 & 0 & 25 & 0 & 25 & 27 & 0 & 27 & 0 \\
 17 & 0 & 17 & 0 & 15 & 0 & 0 & 15 & 0 & 0 & 0 & 0 & 0 & 21 & 0 & 20 & 0 & 0 & 25 & 25 & 27 & 0 & 27 & 0 \\
 17 & 0 & 17 & 0 & 0 & 15 & 15 & 0 & 0 & 0 & 0 & 0 & 21 & 0 & 20 & 0 & 0 & 0 & 25 & 25 & 0 & 27 & 0 & 27 \\
 0 & 17 & 0 & 17 & 0 & 15 & 15 & 0 & 0 & 0 & 0 & 0 & 0 & 20 & 0 & 21 & 25 & 25 & 0 & 0 & 27 & 0 & 27 & 0 \\
 0 & 17 & 0 & 17 & 15 & 0 & 0 & 15 & 0 & 0 & 0 & 0 & 20 & 0 & 21 & 0 & 25 & 25 & 0 & 0 & 0 & 27 & 0 & 27 \\
 29 & 0 & 29 & 0 & 0 & 29 & 0 & 29 & 31 & 0 & 31 & 0 & 0 & 0 & 31 & 31 & 0 & 0 & 0 & 0 & 36 & 0 & 0 & 36 \\
 29 & 0 & 29 & 0 & 29 & 0 & 29 & 0 & 0 & 31 & 0 & 31 & 0 & 0 & 31 & 31 & 0 & 0 & 0 & 0 & 0 & 36 & 36 & 0 \\
 0 & 29 & 0 & 29 & 0 & 29 & 0 & 29 & 31 & 0 & 31 & 0 & 31 & 31 & 0 & 0 & 0 & 0 & 0 & 0 & 0 & 36 & 36 & 0 \\
 0 & 29 & 0 & 29 & 29 & 0 & 29 & 0 & 0 & 31 & 0 & 31 & 31 & 31 & 0 & 0 & 0 & 0 & 0 & 0 & 36 & 0 & 0 & 36 \\
 38 & 0 & 0 & 38 & 0 & 0 & 38 & 38 & 40 & 0 & 0 & 40 & 40 & 0 & 40 & 0 & 42 & 0 & 0 & 42 & 0 & 0 & 0 & 0 \\
 38 & 0 & 0 & 38 & 0 & 0 & 38 & 38 & 0 & 40 & 40 & 0 & 0 & 40 & 0 & 40 & 0 & 42 & 42 & 0 & 0 & 0 & 0 & 0 \\
 0 & 38 & 38 & 0 & 38 & 38 & 0 & 0 & 40 & 0 & 0 & 40 & 40 & 0 & 40 & 0 & 0 & 42 & 42 & 0 & 0 & 0 & 0 & 0 \\
 0 & 38 & 38 & 0 & 38 & 38 & 0 & 0 & 0 & 40 & 40 & 0 & 0 & 40 & 0 & 40 & 42 & 0 & 0 & 42 & 0 & 0 & 0 & 0
 \end{smallmatrix}\right),\label{XWLX}\\
&\begin{smallmatrix}
\tcr{1} & \tcb{2} & \tcb{3} & \tcb{4} & 5 & 5 & 6 & 6 & 7 & 7 & 8 & 8 & 9 & 9 & 10 & 10 & 11 & 11 & 12 & 12 & 13 & 13 & 14 & 14 \\
 \tcb{2} & \tcr{1} & \tcb{4} & \tcb{3} & 6 & 6 & 5 & 5 & 7 & 7 & 8 & 8 & 10 & 10 & 9 & 9 & 12 & 12 & 11 & 11 & 14 & 14 & 13 & 13 \\
 \tcb{3} & \tcb{4} & \tcr{1} & \tcb{2} & 6 & 6 & 5 & 5 & 8 & 8 & 7 & 7 & 9 & 9 & 10 & 10 & 11 & 11 & 12 & 12 & 14 & 14 & 13 & 13 \\
 \tcb{4} & \tcb{3} & \tcb{2} & \tcr{1} & 5 & 5 & 6 & 6 & 8 & 8 & 7 & 7 & 10 & 10 & 9 & 9 & 12 & 12 & 11 & 11 & 13 & 13 & 14 & 14 \\
 5 & 6 & 6 & 5 & \tcr{161} & \tcb{162} & \tcb{163} & \tcb{164} & 165 & 166 & 165 & 166 & 167 & 168 & 168 & 167 & 169 & 170 & 169 & 170 & 171 & 171 & 172 & 172 \\
 5 & 6 & 6 & 5 & \tcb{162} & \tcr{161} & \tcb{164} & \tcb{163} & 166 & 165 & 166 & 165 & 168 & 167 & 167 & 168 & 170 & 169 & 170 & 169 & 171 & 171 & 172 & 172 \\
 6 & 5 & 5 & 6 & \tcb{163} & \tcb{164} & \tcr{161} & \tcb{162} & 165 & 166 & 165 & 166 & 168 & 167 & 167 & 168 & 169 & 170 & 169 & 170 & 172 & 172 & 171 & 171 \\
 6 & 5 & 5 & 6 & \tcb{164} & \tcb{163} & \tcb{162} & \tcr{161} & 166 & 165 & 166 & 165 & 167 & 168 & 168 & 167 & 170 & 169 & 170 & 169 & 172 & 172 & 171 & 171 \\
 7 & 7 & 8 & 8 & 165 & 166 & 165 & 166 & \tcr{319} & \tcb{320} & \tcb{321} & \tcb{322} & 323 & 324 & 323 & 324 & 325 & 326 & 325 & 326 & 327 & 328 & 327 & 328 \\
 7 & 7 & 8 & 8 & 166 & 165 & 166 & 165 & \tcb{320} & \tcr{319} & \tcb{322} & \tcb{321} & 324 & 323 & 324 & 323 & 326 & 325 & 326 & 325 & 328 & 327 & 328 & 327 \\
 8 & 8 & 7 & 7 & 165 & 166 & 165 & 166 & \tcb{321} & \tcb{322} & \tcr{319} & \tcb{320} & 324 & 323 & 324 & 323 & 325 & 326 & 325 & 326 & 328 & 327 & 328 & 327 \\
 8 & 8 & 7 & 7 & 166 & 165 & 166 & 165 & \tcb{322} & \tcb{321} & \tcb{320} & \tcr{319} & 323 & 324 & 323 & 324 & 326 & 325 & 326 & 325 & 327 & 328 & 327 & 328 \\
 9 & 10 & 9 & 10 & 167 & 168 & 168 & 167 & 323 & 324 & 324 & 323 & \tcr{475} & \tcb{476} & \tcb{477} & \tcb{478} & 479 & 479 & 480 & 480 & 481 & 482 & 481 & 482 \\
 9 & 10 & 9 & 10 & 168 & 167 & 167 & 168 & 324 & 323 & 323 & 324 & \tcb{476} & \tcr{475} & \tcb{478} & \tcb{477} & 479 & 479 & 480 & 480 & 482 & 481 & 482 & 481 \\
 10 & 9 & 10 & 9 & 168 & 167 & 167 & 168 & 323 & 324 & 324 & 323 & \tcb{477} & \tcb{478} & \tcr{475} & \tcb{476} & 480 & 480 & 479 & 479 & 481 & 482 & 481 & 482 \\
 10 & 9 & 10 & 9 & 167 & 168 & 168 & 167 & 324 & 323 & 323 & 324 & \tcb{478} & \tcb{477} & \tcb{476} & \tcr{475} & 480 & 480 & 479 & 479 & 482 & 481 & 482 & 481 \\
 11 & 12 & 11 & 12 & 169 & 170 & 169 & 170 & 325 & 326 & 325 & 326 & 479 & 479 & 480 & 480 & \tcr{629} & \tcb{630} & \tcb{631} & \tcb{632} & 633 & 634 & 634 & 633 \\
 11 & 12 & 11 & 12 & 170 & 169 & 170 & 169 & 326 & 325 & 326 & 325 & 479 & 479 & 480 & 480 & \tcb{630} & \tcr{629} & \tcb{632} & \tcb{631} & 634 & 633 & 633 & 634 \\
 12 & 11 & 12 & 11 & 169 & 170 & 169 & 170 & 325 & 326 & 325 & 326 & 480 & 480 & 479 & 479 & \tcb{631} & \tcb{632} & \tcr{629} & \tcb{630} & 634 & 633 & 633 & 634 \\
 12 & 11 & 12 & 11 & 170 & 169 & 170 & 169 & 326 & 325 & 326 & 325 & 480 & 480 & 479 & 479 & \tcb{632} & \tcb{631} & \tcb{630} & \tcr{629} & 633 & 634 & 634 & 633 \\
 13 & 14 & 14 & 13 & 171 & 171 & 172 & 172 & 327 & 328 & 328 & 327 & 481 & 482 & 481 & 482 & 633 & 634 & 634 & 633 & \tcr{781} & \tcb{782} & \tcb{783} & \tcb{784} \\
 13 & 14 & 14 & 13 & 171 & 171 & 172 & 172 & 328 & 327 & 327 & 328 & 482 & 481 & 482 & 481 & 634 & 633 & 633 & 634 & \tcb{782} & \tcr{781} & \tcb{784} & \tcb{783} \\
 14 & 13 & 13 & 14 & 172 & 172 & 171 & 171 & 327 & 328 & 328 & 327 & 481 & 482 & 481 & 482 & 634 & 633 & 633 & 634 & \tcb{783} & \tcb{784} & \tcr{781} & \tcb{782} \\
 14 & 13 & 13 & 14 & 172 & 172 & 171 & 171 & 328 & 327 & 327 & 328 & 482 & 481 & 482 & 481 & 633 & 634 & 634 & 633 & \tcb{784} & \tcb{783} & \tcb{782} & \tcr{781}
\end{smallmatrix}\label{XB}
\end{align}
The last matrix in \eqref{XB} is the induced subgraph of the stable graph $\hat{B}$ by vertices $[24]$, where $B$ is the binding graph of  graph $X$. One may see that the basic cells in the vertex partition is the automorphism partition of graph $X$, shown as follows. $$\{\{1,2,3,4\},\{5,6,7,8\},\{9,10,11,12\},\{13,14,15,16\},\{17,18,19,20\},\{21,22,23,24\}\}.$$
\ep

3. We now present the third example of a small graph to illustrate the graphs $\Phi$, $\Theta$ and so on appeared in the article. The graph $X$, the stable graph $\hat{X}$ obtained by SaS process, the stable graph $\WL(X)$ obtained by WL process and a binding graph $\mathtt{bi}(X)$ of graph $X$ are listed in the followings.

The stable graph of binding graph $\mathtt{bi}(X)$ obtained by SaS process is given on page \pageref{page1}. On the same page the bipartite graph $\Phi$ constructed from the stable graph is also listed. The numbers in blue and red are the number of column and row, respectively.

The bipartite graph $\Theta$ constructed from $\Phi$ is given on page \pageref{page2}. On the page \pageref{page3} of this example, we show the stable graph of binding graph $\mathtt{bi}(X)$ obtained by WL process.

One may verify (cf. Theorem \ref{thm:theta} and Theorem \ref{thm:xphi}) that $\hat{\Theta}\approx\hat{\Phi}$ and they are equivalent to the stable graph of $\mathtt{bi}(X)$ obtained by SaS process on page \pageref{page1}, and $\WL(\Theta)\approx \WL(\Phi)$ and they are equivalent to the stable graph of $\mathtt{bi}(X)$ obtained by WL process shown on page \pageref{page3}.

\begin{align}
X=\left(\begin{smallmatrix}
 0 & 1 & 1 & 1 & 0 & 0 & 0 & 0 \\
 1 & 0 & 1 & 0 & 1 & 0 & 0 & 0 \\
 1 & 1 & 0 & 0 & 0 & 1 & 0 & 0 \\
 1 & 0 & 0 & 0 & 0 & 0 & 1 & 1 \\
 0 & 1 & 0 & 0 & 0 & 0 & 1 & 1 \\
 0 & 0 & 1 & 0 & 0 & 0 & 1 & 1 \\
 0 & 0 & 0 & 1 & 1 & 1 & 0 & 0 \\
 0 & 0 & 0 & 1 & 1 & 1 & 0 & 0
\end{smallmatrix}\right),\quad
\hat{X}=\left(\begin{smallmatrix}
 1 & 2 & 2 & 3 & 4 & 4 & 5 & 5 \\
 2 & 1 & 2 & 4 & 3 & 4 & 5 & 5 \\
 2 & 2 & 1 & 4 & 4 & 3 & 5 & 5 \\
 3 & 4 & 4 & 6 & 7 & 7 & 8 & 8 \\
 4 & 3 & 4 & 7 & 6 & 7 & 8 & 8 \\
 4 & 4 & 3 & 7 & 7 & 6 & 8 & 8 \\
 5 & 5 & 5 & 8 & 8 & 8 & 9 & 10 \\
 5 & 5 & 5 & 8 & 8 & 8 & 10 & 9
\end{smallmatrix}\right),\quad
\WL(X)=\left(\begin{smallmatrix}
 1 & 2 & 2 & 3 & 4 & 4 & 5 & 5 \\
 2 & 1 & 2 & 4 & 3 & 4 & 5 & 5 \\
 2 & 2 & 1 & 4 & 4 & 3 & 5 & 5 \\
 6 & 7 & 7 & 8 & 9 & 9 & 10 & 10 \\
 7 & 6 & 7 & 9 & 8 & 9 & 10 & 10 \\
 7 & 7 & 6 & 9 & 9 & 8 & 10 & 10 \\
 11 & 11 & 11 & 12 & 12 & 12 & 13 & 14 \\
 11 & 11 & 11 & 12 & 12 & 12 & 14 & 13
\end{smallmatrix}\right).\label{xphi}\notag{}
\end{align}
\begin{align}
\mathtt{bi}(X)=\left(\begin{smallmatrix}
 0 & \tcb{1} & \tcb{2} & \tcb{3} & \tcb{4} & \tcb{5} & \tcb{6} & \tcb{7} & \tcb{8} & \tcb{9} & \tcb{10} & \tcb{11} & \tcb{12} & \tcb{13} & \tcb{14} & \tcb{15} & \tcb{16} & \tcb{17} & \tcb{18} & \tcb{19} & \tcb{20} & \tcb{21} & \tcb{22} & \tcb{23} & \tcb{24} & \tcb{25} & \tcb{26} & \tcb{27} & \tcb{28} & \tcb{29} & \tcb{30} & \tcb{31} & \tcb{32} & \tcb{33} & \tcb{34} & \tcb{35} & \tcb{36} \\
\tcr{1} & 0 & 1 & 1 & 1 & 0 & 0 & 0 & 0 & 1 & 1 & 0 & 1 & 0 & 0 & 1 & 1 & 0 & 0 & 0 & 0 & 1 & 1 & 0 & 0 & 0 & 0 & 0 & 0 & 0 & 0 & 0 & 0 & 0 & 0 & 0 & 0 \\
\tcr{2} & 1 & 0 & 1 & 0 & 1 & 0 & 0 & 0 & 1 & 0 & 1 & 0 & 1 & 0 & 0 & 0 & 1 & 1 & 0 & 0 & 0 & 0 & 1 & 1 & 0 & 0 & 0 & 0 & 0 & 0 & 0 & 0 & 0 & 0 & 0 & 0 \\
 \tcr{3} & 1 & 1 & 0 & 0 & 0 & 1 & 0 & 0 & 0 & 1 & 1 & 0 & 0 & 1 & 0 & 0 & 0 & 0 & 1 & 1 & 0 & 0 & 0 & 0 & 1 & 1 & 0 & 0 & 0 & 0 & 0 & 0 & 0 & 0 & 0 & 0 \\
 \tcr{4} & 1 & 0 & 0 & 0 & 0 & 0 & 1 & 1 & 0 & 0 & 0 & 1 & 0 & 0 & 0 & 0 & 1 & 0 & 1 & 0 & 0 & 0 & 0 & 0 & 0 & 0 & 1 & 1 & 0 & 1 & 1 & 0 & 0 & 0 & 0 & 0 \\
 \tcr{5} & 0 & 1 & 0 & 0 & 0 & 0 & 1 & 1 & 0 & 0 & 0 & 0 & 1 & 0 & 1 & 0 & 0 & 0 & 0 & 1 & 0 & 0 & 0 & 0 & 0 & 0 & 1 & 0 & 1 & 0 & 0 & 1 & 1 & 0 & 0 & 0 \\
 \tcr{6} & 0 & 0 & 1 & 0 & 0 & 0 & 1 & 1 & 0 & 0 & 0 & 0 & 0 & 1 & 0 & 1 & 0 & 1 & 0 & 0 & 0 & 0 & 0 & 0 & 0 & 0 & 0 & 1 & 1 & 0 & 0 & 0 & 0 & 1 & 1 & 0 \\
 \tcr{7} & 0 & 0 & 0 & 1 & 1 & 1 & 0 & 0 & 0 & 0 & 0 & 0 & 0 & 0 & 0 & 0 & 0 & 0 & 0 & 0 & 1 & 0 & 1 & 0 & 1 & 0 & 0 & 0 & 0 & 1 & 0 & 1 & 0 & 1 & 0 & 1 \\
 \tcr{8} & 0 & 0 & 0 & 1 & 1 & 1 & 0 & 0 & 0 & 0 & 0 & 0 & 0 & 0 & 0 & 0 & 0 & 0 & 0 & 0 & 0 & 1 & 0 & 1 & 0 & 1 & 0 & 0 & 0 & 0 & 1 & 0 & 1 & 0 & 1 & 1 \\
 \tcr{9} & 1 & 1 & 0 & 0 & 0 & 0 & 0 & 0 & 0 & 0 & 0 & 0 & 0 & 0 & 0 & 0 & 0 & 0 & 0 & 0 & 0 & 0 & 0 & 0 & 0 & 0 & 0 & 0 & 0 & 0 & 0 & 0 & 0 & 0 & 0 & 0 \\
 \tcr{10} & 1 & 0 & 1 & 0 & 0 & 0 & 0 & 0 & 0 & 0 & 0 & 0 & 0 & 0 & 0 & 0 & 0 & 0 & 0 & 0 & 0 & 0 & 0 & 0 & 0 & 0 & 0 & 0 & 0 & 0 & 0 & 0 & 0 & 0 & 0 & 0 \\
 \tcr{11} & 0 & 1 & 1 & 0 & 0 & 0 & 0 & 0 & 0 & 0 & 0 & 0 & 0 & 0 & 0 & 0 & 0 & 0 & 0 & 0 & 0 & 0 & 0 & 0 & 0 & 0 & 0 & 0 & 0 & 0 & 0 & 0 & 0 & 0 & 0 & 0 \\
 \tcr{12} & 1 & 0 & 0 & 1 & 0 & 0 & 0 & 0 & 0 & 0 & 0 & 0 & 0 & 0 & 0 & 0 & 0 & 0 & 0 & 0 & 0 & 0 & 0 & 0 & 0 & 0 & 0 & 0 & 0 & 0 & 0 & 0 & 0 & 0 & 0 & 0 \\
 \tcr{13} & 0 & 1 & 0 & 0 & 1 & 0 & 0 & 0 & 0 & 0 & 0 & 0 & 0 & 0 & 0 & 0 & 0 & 0 & 0 & 0 & 0 & 0 & 0 & 0 & 0 & 0 & 0 & 0 & 0 & 0 & 0 & 0 & 0 & 0 & 0 & 0 \\
 \tcr{14} & 0 & 0 & 1 & 0 & 0 & 1 & 0 & 0 & 0 & 0 & 0 & 0 & 0 & 0 & 0 & 0 & 0 & 0 & 0 & 0 & 0 & 0 & 0 & 0 & 0 & 0 & 0 & 0 & 0 & 0 & 0 & 0 & 0 & 0 & 0 & 0 \\
 \tcr{15} & 1 & 0 & 0 & 0 & 1 & 0 & 0 & 0 & 0 & 0 & 0 & 0 & 0 & 0 & 0 & 0 & 0 & 0 & 0 & 0 & 0 & 0 & 0 & 0 & 0 & 0 & 0 & 0 & 0 & 0 & 0 & 0 & 0 & 0 & 0 & 0 \\
 \tcr{16} & 1 & 0 & 0 & 0 & 0 & 1 & 0 & 0 & 0 & 0 & 0 & 0 & 0 & 0 & 0 & 0 & 0 & 0 & 0 & 0 & 0 & 0 & 0 & 0 & 0 & 0 & 0 & 0 & 0 & 0 & 0 & 0 & 0 & 0 & 0 & 0 \\
 \tcr{17} & 0 & 1 & 0 & 1 & 0 & 0 & 0 & 0 & 0 & 0 & 0 & 0 & 0 & 0 & 0 & 0 & 0 & 0 & 0 & 0 & 0 & 0 & 0 & 0 & 0 & 0 & 0 & 0 & 0 & 0 & 0 & 0 & 0 & 0 & 0 & 0 \\
 \tcr{18} & 0 & 1 & 0 & 0 & 0 & 1 & 0 & 0 & 0 & 0 & 0 & 0 & 0 & 0 & 0 & 0 & 0 & 0 & 0 & 0 & 0 & 0 & 0 & 0 & 0 & 0 & 0 & 0 & 0 & 0 & 0 & 0 & 0 & 0 & 0 & 0 \\
 \tcr{19} & 0 & 0 & 1 & 1 & 0 & 0 & 0 & 0 & 0 & 0 & 0 & 0 & 0 & 0 & 0 & 0 & 0 & 0 & 0 & 0 & 0 & 0 & 0 & 0 & 0 & 0 & 0 & 0 & 0 & 0 & 0 & 0 & 0 & 0 & 0 & 0 \\
 \tcr{20} & 0 & 0 & 1 & 0 & 1 & 0 & 0 & 0 & 0 & 0 & 0 & 0 & 0 & 0 & 0 & 0 & 0 & 0 & 0 & 0 & 0 & 0 & 0 & 0 & 0 & 0 & 0 & 0 & 0 & 0 & 0 & 0 & 0 & 0 & 0 & 0 \\
 \tcr{21} & 1 & 0 & 0 & 0 & 0 & 0 & 1 & 0 & 0 & 0 & 0 & 0 & 0 & 0 & 0 & 0 & 0 & 0 & 0 & 0 & 0 & 0 & 0 & 0 & 0 & 0 & 0 & 0 & 0 & 0 & 0 & 0 & 0 & 0 & 0 & 0 \\
 \tcr{22} & 1 & 0 & 0 & 0 & 0 & 0 & 0 & 1 & 0 & 0 & 0 & 0 & 0 & 0 & 0 & 0 & 0 & 0 & 0 & 0 & 0 & 0 & 0 & 0 & 0 & 0 & 0 & 0 & 0 & 0 & 0 & 0 & 0 & 0 & 0 & 0 \\
 \tcr{23} & 0 & 1 & 0 & 0 & 0 & 0 & 1 & 0 & 0 & 0 & 0 & 0 & 0 & 0 & 0 & 0 & 0 & 0 & 0 & 0 & 0 & 0 & 0 & 0 & 0 & 0 & 0 & 0 & 0 & 0 & 0 & 0 & 0 & 0 & 0 & 0 \\
 \tcr{24} & 0 & 1 & 0 & 0 & 0 & 0 & 0 & 1 & 0 & 0 & 0 & 0 & 0 & 0 & 0 & 0 & 0 & 0 & 0 & 0 & 0 & 0 & 0 & 0 & 0 & 0 & 0 & 0 & 0 & 0 & 0 & 0 & 0 & 0 & 0 & 0 \\
 \tcr{25} & 0 & 0 & 1 & 0 & 0 & 0 & 1 & 0 & 0 & 0 & 0 & 0 & 0 & 0 & 0 & 0 & 0 & 0 & 0 & 0 & 0 & 0 & 0 & 0 & 0 & 0 & 0 & 0 & 0 & 0 & 0 & 0 & 0 & 0 & 0 & 0 \\
 \tcr{26} & 0 & 0 & 1 & 0 & 0 & 0 & 0 & 1 & 0 & 0 & 0 & 0 & 0 & 0 & 0 & 0 & 0 & 0 & 0 & 0 & 0 & 0 & 0 & 0 & 0 & 0 & 0 & 0 & 0 & 0 & 0 & 0 & 0 & 0 & 0 & 0 \\
 \tcr{27} & 0 & 0 & 0 & 1 & 1 & 0 & 0 & 0 & 0 & 0 & 0 & 0 & 0 & 0 & 0 & 0 & 0 & 0 & 0 & 0 & 0 & 0 & 0 & 0 & 0 & 0 & 0 & 0 & 0 & 0 & 0 & 0 & 0 & 0 & 0 & 0 \\
 \tcr{28} & 0 & 0 & 0 & 1 & 0 & 1 & 0 & 0 & 0 & 0 & 0 & 0 & 0 & 0 & 0 & 0 & 0 & 0 & 0 & 0 & 0 & 0 & 0 & 0 & 0 & 0 & 0 & 0 & 0 & 0 & 0 & 0 & 0 & 0 & 0 & 0 \\
 \tcr{29} & 0 & 0 & 0 & 0 & 1 & 1 & 0 & 0 & 0 & 0 & 0 & 0 & 0 & 0 & 0 & 0 & 0 & 0 & 0 & 0 & 0 & 0 & 0 & 0 & 0 & 0 & 0 & 0 & 0 & 0 & 0 & 0 & 0 & 0 & 0 & 0 \\
 \tcr{30} & 0 & 0 & 0 & 1 & 0 & 0 & 1 & 0 & 0 & 0 & 0 & 0 & 0 & 0 & 0 & 0 & 0 & 0 & 0 & 0 & 0 & 0 & 0 & 0 & 0 & 0 & 0 & 0 & 0 & 0 & 0 & 0 & 0 & 0 & 0 & 0 \\
 \tcr{31} & 0 & 0 & 0 & 1 & 0 & 0 & 0 & 1 & 0 & 0 & 0 & 0 & 0 & 0 & 0 & 0 & 0 & 0 & 0 & 0 & 0 & 0 & 0 & 0 & 0 & 0 & 0 & 0 & 0 & 0 & 0 & 0 & 0 & 0 & 0 & 0 \\
 \tcr{32} & 0 & 0 & 0 & 0 & 1 & 0 & 1 & 0 & 0 & 0 & 0 & 0 & 0 & 0 & 0 & 0 & 0 & 0 & 0 & 0 & 0 & 0 & 0 & 0 & 0 & 0 & 0 & 0 & 0 & 0 & 0 & 0 & 0 & 0 & 0 & 0 \\
 \tcr{33} & 0 & 0 & 0 & 0 & 1 & 0 & 0 & 1 & 0 & 0 & 0 & 0 & 0 & 0 & 0 & 0 & 0 & 0 & 0 & 0 & 0 & 0 & 0 & 0 & 0 & 0 & 0 & 0 & 0 & 0 & 0 & 0 & 0 & 0 & 0 & 0 \\
 \tcr{34} & 0 & 0 & 0 & 0 & 0 & 1 & 1 & 0 & 0 & 0 & 0 & 0 & 0 & 0 & 0 & 0 & 0 & 0 & 0 & 0 & 0 & 0 & 0 & 0 & 0 & 0 & 0 & 0 & 0 & 0 & 0 & 0 & 0 & 0 & 0 & 0 \\
 \tcr{35} & 0 & 0 & 0 & 0 & 0 & 1 & 0 & 1 & 0 & 0 & 0 & 0 & 0 & 0 & 0 & 0 & 0 & 0 & 0 & 0 & 0 & 0 & 0 & 0 & 0 & 0 & 0 & 0 & 0 & 0 & 0 & 0 & 0 & 0 & 0 & 0 \\
 \tcr{36} & 0 & 0 & 0 & 0 & 0 & 0 & 1 & 1 & 0 & 0 & 0 & 0 & 0 & 0 & 0 & 0 & 0 & 0 & 0 & 0 & 0 & 0 & 0 & 0 & 0 & 0 & 0 & 0 & 0 & 0 & 0 & 0 & 0 & 0 & 0 & 0
\end{smallmatrix}\right).\tag{This is the binding graph of $X$}
\end{align}

\begin{align}
\begin{smallmatrix}
 0 & \tcb{1} & \tcb{2} & \tcb{3} & \tcb{4} & \tcb{5} & \tcb{6} & \tcb{7} & \tcb{8} & \tcb{9} & \tcb{10} & \tcb{11} & \tcb{12} & \tcb{13} & \tcb{14} & \tcb{15} & \tcb{16} & \tcb{17} & \tcb{18} & \tcb{19} & \tcb{20} & \tcb{21} & \tcb{22} & \tcb{23} & \tcb{24} & \tcb{25} & \tcb{26} & \tcb{27} & \tcb{28} & \tcb{29} & \tcb{30} & \tcb{31} & \tcb{32} & \tcb{33} & \tcb{34} & \tcb{35} & \tcb{36} \\
 \tcr{1} & 1 & 2 & 2 & 3 & 4 & 4 & 5 & 5 & 6 & 6 & 7 & 8 & 9 & 9 & 10 & 10 & 11 & 12 & 11 & 12 & 13 & 13 & 14 & 14 & 14 & 14 & 15 & 15 & 16 & 17 & 17 & 18 & 18 & 18 & 18 & 19 \\
 \tcr{2} & 2 & 1 & 2 & 4 & 3 & 4 & 5 & 5 & 6 & 7 & 6 & 9 & 8 & 9 & 11 & 12 & 10 & 10 & 12 & 11 & 14 & 14 & 13 & 13 & 14 & 14 & 15 & 16 & 15 & 18 & 18 & 17 & 17 & 18 & 18 & 19 \\
 \tcr{3} & 2 & 2 & 1 & 4 & 4 & 3 & 5 & 5 & 7 & 6 & 6 & 9 & 9 & 8 & 12 & 11 & 12 & 11 & 10 & 10 & 14 & 14 & 14 & 14 & 13 & 13 & 16 & 15 & 15 & 18 & 18 & 18 & 18 & 17 & 17 & 19 \\
 \tcr{4} & 3 & 4 & 4 & 20 & 21 & 21 & 22 & 22 & 23 & 23 & 24 & 25 & 26 & 26 & 27 & 27 & 28 & 29 & 28 & 29 & 30 & 30 & 31 & 31 & 31 & 31 & 32 & 32 & 33 & 34 & 34 & 35 & 35 & 35 & 35 & 36 \\
 \tcr{5} & 4 & 3 & 4 & 21 & 20 & 21 & 22 & 22 & 23 & 24 & 23 & 26 & 25 & 26 & 28 & 29 & 27 & 27 & 29 & 28 & 31 & 31 & 30 & 30 & 31 & 31 & 32 & 33 & 32 & 35 & 35 & 34 & 34 & 35 & 35 & 36 \\
 \tcr{6} & 4 & 4 & 3 & 21 & 21 & 20 & 22 & 22 & 24 & 23 & 23 & 26 & 26 & 25 & 29 & 28 & 29 & 28 & 27 & 27 & 31 & 31 & 31 & 31 & 30 & 30 & 33 & 32 & 32 & 35 & 35 & 35 & 35 & 34 & 34 & 36 \\
 \tcr{7} & 5 & 5 & 5 & 22 & 22 & 22 & 37 & 38 & 39 & 39 & 39 & 40 & 40 & 40 & 41 & 41 & 41 & 41 & 41 & 41 & 42 & 43 & 42 & 43 & 42 & 43 & 44 & 44 & 44 & 45 & 46 & 45 & 46 & 45 & 46 & 47 \\
 \tcr{8} & 5 & 5 & 5 & 22 & 22 & 22 & 38 & 37 & 39 & 39 & 39 & 40 & 40 & 40 & 41 & 41 & 41 & 41 & 41 & 41 & 43 & 42 & 43 & 42 & 43 & 42 & 44 & 44 & 44 & 46 & 45 & 46 & 45 & 46 & 45 & 47 \\
 \tcr{9} & 6 & 6 & 7 & 23 & 23 & 24 & 39 & 39 & 48 & 49 & 49 & 50 & 50 & 51 & 52 & 53 & 52 & 53 & 54 & 54 & 55 & 55 & 55 & 55 & 56 & 56 & 57 & 58 & 58 & 59 & 59 & 59 & 59 & 60 & 60 & 61 \\
 \tcr{10} & 6 & 7 & 6 & 23 & 24 & 23 & 39 & 39 & 49 & 48 & 49 & 50 & 51 & 50 & 53 & 52 & 54 & 54 & 52 & 53 & 55 & 55 & 56 & 56 & 55 & 55 & 58 & 57 & 58 & 59 & 59 & 60 & 60 & 59 & 59 & 61 \\
 \tcr{11} & 7 & 6 & 6 & 24 & 23 & 23 & 39 & 39 & 49 & 49 & 48 & 51 & 50 & 50 & 54 & 54 & 53 & 52 & 53 & 52 & 56 & 56 & 55 & 55 & 55 & 55 & 58 & 58 & 57 & 60 & 60 & 59 & 59 & 59 & 59 & 61 \\
 \tcr{12} & 8 & 9 & 9 & 25 & 26 & 26 & 40 & 40 & 50 & 50 & 51 & 62 & 63 & 63 & 64 & 64 & 65 & 66 & 65 & 66 & 67 & 67 & 68 & 68 & 68 & 68 & 69 & 69 & 70 & 71 & 71 & 72 & 72 & 72 & 72 & 73 \\
 \tcr{13} & 9 & 8 & 9 & 26 & 25 & 26 & 40 & 40 & 50 & 51 & 50 & 63 & 62 & 63 & 65 & 66 & 64 & 64 & 66 & 65 & 68 & 68 & 67 & 67 & 68 & 68 & 69 & 70 & 69 & 72 & 72 & 71 & 71 & 72 & 72 & 73 \\
 \tcr{14} & 9 & 9 & 8 & 26 & 26 & 25 & 40 & 40 & 51 & 50 & 50 & 63 & 63 & 62 & 66 & 65 & 66 & 65 & 64 & 64 & 68 & 68 & 68 & 68 & 67 & 67 & 70 & 69 & 69 & 72 & 72 & 72 & 72 & 71 & 71 & 73 \\
 \tcr{15} & 10 & 11 & 12 & 27 & 28 & 29 & 41 & 41 & 52 & 53 & 54 & 64 & 65 & 66 & 74 & 75 & 76 & 77 & 77 & 78 & 79 & 79 & 80 & 80 & 81 & 81 & 82 & 83 & 84 & 85 & 85 & 86 & 86 & 87 & 87 & 88 \\
 \tcr{16} & 10 & 12 & 11 & 27 & 29 & 28 & 41 & 41 & 53 & 52 & 54 & 64 & 66 & 65 & 75 & 74 & 77 & 78 & 76 & 77 & 79 & 79 & 81 & 81 & 80 & 80 & 83 & 82 & 84 & 85 & 85 & 87 & 87 & 86 & 86 & 88 \\
 \tcr{17} & 11 & 10 & 12 & 28 & 27 & 29 & 41 & 41 & 52 & 54 & 53 & 65 & 64 & 66 & 76 & 77 & 74 & 75 & 78 & 77 & 80 & 80 & 79 & 79 & 81 & 81 & 82 & 84 & 83 & 86 & 86 & 85 & 85 & 87 & 87 & 88 \\
 \tcr{18} & 12 & 10 & 11 & 29 & 27 & 28 & 41 & 41 & 53 & 54 & 52 & 66 & 64 & 65 & 77 & 78 & 75 & 74 & 77 & 76 & 81 & 81 & 79 & 79 & 80 & 80 & 83 & 84 & 82 & 87 & 87 & 85 & 85 & 86 & 86 & 88 \\
 \tcr{19} & 11 & 12 & 10 & 28 & 29 & 27 & 41 & 41 & 54 & 52 & 53 & 65 & 66 & 64 & 77 & 76 & 78 & 77 & 74 & 75 & 80 & 80 & 81 & 81 & 79 & 79 & 84 & 82 & 83 & 86 & 86 & 87 & 87 & 85 & 85 & 88 \\
 \tcr{20} & 12 & 11 & 10 & 29 & 28 & 27 & 41 & 41 & 54 & 53 & 52 & 66 & 65 & 64 & 78 & 77 & 77 & 76 & 75 & 74 & 81 & 81 & 80 & 80 & 79 & 79 & 84 & 83 & 82 & 87 & 87 & 86 & 86 & 85 & 85 & 88 \\
 \tcr{21} & 13 & 14 & 14 & 30 & 31 & 31 & 42 & 43 & 55 & 55 & 56 & 67 & 68 & 68 & 79 & 79 & 80 & 81 & 80 & 81 & 89 & 90 & 91 & 92 & 91 & 92 & 93 & 93 & 94 & 95 & 96 & 97 & 98 & 97 & 98 & 99 \\
 \tcr{22} & 13 & 14 & 14 & 30 & 31 & 31 & 43 & 42 & 55 & 55 & 56 & 67 & 68 & 68 & 79 & 79 & 80 & 81 & 80 & 81 & 90 & 89 & 92 & 91 & 92 & 91 & 93 & 93 & 94 & 96 & 95 & 98 & 97 & 98 & 97 & 99 \\
 \tcr{23} & 14 & 13 & 14 & 31 & 30 & 31 & 42 & 43 & 55 & 56 & 55 & 68 & 67 & 68 & 80 & 81 & 79 & 79 & 81 & 80 & 91 & 92 & 89 & 90 & 91 & 92 & 93 & 94 & 93 & 97 & 98 & 95 & 96 & 97 & 98 & 99 \\
 \tcr{24} & 14 & 13 & 14 & 31 & 30 & 31 & 43 & 42 & 55 & 56 & 55 & 68 & 67 & 68 & 80 & 81 & 79 & 79 & 81 & 80 & 92 & 91 & 90 & 89 & 92 & 91 & 93 & 94 & 93 & 98 & 97 & 96 & 95 & 98 & 97 & 99 \\
 \tcr{25} & 14 & 14 & 13 & 31 & 31 & 30 & 42 & 43 & 56 & 55 & 55 & 68 & 68 & 67 & 81 & 80 & 81 & 80 & 79 & 79 & 91 & 92 & 91 & 92 & 89 & 90 & 94 & 93 & 93 & 97 & 98 & 97 & 98 & 95 & 96 & 99 \\
 \tcr{26} & 14 & 14 & 13 & 31 & 31 & 30 & 43 & 42 & 56 & 55 & 55 & 68 & 68 & 67 & 81 & 80 & 81 & 80 & 79 & 79 & 92 & 91 & 92 & 91 & 90 & 89 & 94 & 93 & 93 & 98 & 97 & 98 & 97 & 96 & 95 & 99 \\
 \tcr{27} & 15 & 15 & 16 & 32 & 32 & 33 & 44 & 44 & 57 & 58 & 58 & 69 & 69 & 70 & 82 & 83 & 82 & 83 & 84 & 84 & 93 & 93 & 93 & 93 & 94 & 94 & 100 & 101 & 101 & 102 & 102 & 102 & 102 & 103 & 103 & 104 \\
 \tcr{28} & 15 & 16 & 15 & 32 & 33 & 32 & 44 & 44 & 58 & 57 & 58 & 69 & 70 & 69 & 83 & 82 & 84 & 84 & 82 & 83 & 93 & 93 & 94 & 94 & 93 & 93 & 101 & 100 & 101 & 102 & 102 & 103 & 103 & 102 & 102 & 104 \\
 \tcr{29} & 16 & 15 & 15 & 33 & 32 & 32 & 44 & 44 & 58 & 58 & 57 & 70 & 69 & 69 & 84 & 84 & 83 & 82 & 83 & 82 & 94 & 94 & 93 & 93 & 93 & 93 & 101 & 101 & 100 & 103 & 103 & 102 & 102 & 102 & 102 & 104 \\
 \tcr{30} & 17 & 18 & 18 & 34 & 35 & 35 & 45 & 46 & 59 & 59 & 60 & 71 & 72 & 72 & 85 & 85 & 86 & 87 & 86 & 87 & 95 & 96 & 97 & 98 & 97 & 98 & 102 & 102 & 103 & 105 & 106 & 107 & 108 & 107 & 108 & 109 \\
 \tcr{31} & 17 & 18 & 18 & 34 & 35 & 35 & 46 & 45 & 59 & 59 & 60 & 71 & 72 & 72 & 85 & 85 & 86 & 87 & 86 & 87 & 96 & 95 & 98 & 97 & 98 & 97 & 102 & 102 & 103 & 106 & 105 & 108 & 107 & 108 & 107 & 109 \\
 \tcr{32} & 18 & 17 & 18 & 35 & 34 & 35 & 45 & 46 & 59 & 60 & 59 & 72 & 71 & 72 & 86 & 87 & 85 & 85 & 87 & 86 & 97 & 98 & 95 & 96 & 97 & 98 & 102 & 103 & 102 & 107 & 108 & 105 & 106 & 107 & 108 & 109 \\
 \tcr{33} & 18 & 17 & 18 & 35 & 34 & 35 & 46 & 45 & 59 & 60 & 59 & 72 & 71 & 72 & 86 & 87 & 85 & 85 & 87 & 86 & 98 & 97 & 96 & 95 & 98 & 97 & 102 & 103 & 102 & 108 & 107 & 106 & 105 & 108 & 107 & 109 \\
 \tcr{34} & 18 & 18 & 17 & 35 & 35 & 34 & 45 & 46 & 60 & 59 & 59 & 72 & 72 & 71 & 87 & 86 & 87 & 86 & 85 & 85 & 97 & 98 & 97 & 98 & 95 & 96 & 103 & 102 & 102 & 107 & 108 & 107 & 108 & 105 & 106 & 109 \\
 \tcr{35} & 18 & 18 & 17 & 35 & 35 & 34 & 46 & 45 & 60 & 59 & 59 & 72 & 72 & 71 & 87 & 86 & 87 & 86 & 85 & 85 & 98 & 97 & 98 & 97 & 96 & 95 & 103 & 102 & 102 & 108 & 107 & 108 & 107 & 106 & 105 & 109 \\
 \tcr{36} & 19 & 19 & 19 & 36 & 36 & 36 & 47 & 47 & 61 & 61 & 61 & 73 & 73 & 73 & 88 & 88 & 88 & 88 & 88 & 88 & 99 & 99 & 99 & 99 & 99 & 99 & 104 & 104 & 104 & 109 & 109 & 109 & 109 & 109 & 109 & 110
\end{smallmatrix}\tag{This is the stable graph of $\mathtt{bi}(X)$ obtained by SaS process.}\label{page1}\\[.3cm]
\begin{smallmatrix}
0 & \tcb{1} & \tcb{2} & \tcb{3} & \tcb{4} & \tcb{5} & \tcb{6} & \tcb{7} & \tcb{8} & \tcb{9} & \tcb{10} & \tcb{11} & \tcb{12} & \tcb{13} & \tcb{14} & \tcb{15} & \tcb{16} & \tcb{17} & \tcb{18} & \tcb{19} & \tcb{20} & \tcb{21} & \tcb{22} & \tcb{23} & \tcb{24} & \tcb{25} & \tcb{26} & \tcb{27} & \tcb{28} & \tcb{29} & \tcb{30} & \tcb{31} & \tcb{32} & \tcb{33} & \tcb{34} & \tcb{35} & \tcb{36} \\
 \tcr{1} & 1 & 0 & 0 & 0 & 0 & 0 & 0 & 0 & 6 & 6 & 0 & 8 & 0 & 0 & 10 & 10 & 0 & 0 & 0 & 0 & 13 & 13 & 0 & 0 & 0 & 0 & 0 & 0 & 0 & 0 & 0 & 0 & 0 & 0 & 0 & 0 \\
 \tcr{2} & 0 & 1 & 0 & 0 & 0 & 0 & 0 & 0 & 6 & 0 & 6 & 0 & 8 & 0 & 0 & 0 & 10 & 10 & 0 & 0 & 0 & 0 & 13 & 13 & 0 & 0 & 0 & 0 & 0 & 0 & 0 & 0 & 0 & 0 & 0 & 0 \\
 \tcr{3} & 0 & 0 & 1 & 0 & 0 & 0 & 0 & 0 & 0 & 6 & 6 & 0 & 0 & 8 & 0 & 0 & 0 & 0 & 10 & 10 & 0 & 0 & 0 & 0 & 13 & 13 & 0 & 0 & 0 & 0 & 0 & 0 & 0 & 0 & 0 & 0 \\
 \tcr{4} & 0 & 0 & 0 & 20 & 0 & 0 & 0 & 0 & 0 & 0 & 0 & 25 & 0 & 0 & 0 & 0 & 28 & 0 & 28 & 0 & 0 & 0 & 0 & 0 & 0 & 0 & 32 & 32 & 0 & 34 & 34 & 0 & 0 & 0 & 0 & 0 \\
 \tcr{5} & 0 & 0 & 0 & 0 & 20 & 0 & 0 & 0 & 0 & 0 & 0 & 0 & 25 & 0 & 28 & 0 & 0 & 0 & 0 & 28 & 0 & 0 & 0 & 0 & 0 & 0 & 32 & 0 & 32 & 0 & 0 & 34 & 34 & 0 & 0 & 0 \\
 \tcr{6} & 0 & 0 & 0 & 0 & 0 & 20 & 0 & 0 & 0 & 0 & 0 & 0 & 0 & 25 & 0 & 28 & 0 & 28 & 0 & 0 & 0 & 0 & 0 & 0 & 0 & 0 & 0 & 32 & 32 & 0 & 0 & 0 & 0 & 34 & 34 & 0 \\
 \tcr{7} & 0 & 0 & 0 & 0 & 0 & 0 & 37 & 0 & 0 & 0 & 0 & 0 & 0 & 0 & 0 & 0 & 0 & 0 & 0 & 0 & 42 & 0 & 42 & 0 & 42 & 0 & 0 & 0 & 0 & 45 & 0 & 45 & 0 & 45 & 0 & 47 \\
 \tcr{8} & 0 & 0 & 0 & 0 & 0 & 0 & 0 & 37 & 0 & 0 & 0 & 0 & 0 & 0 & 0 & 0 & 0 & 0 & 0 & 0 & 0 & 42 & 0 & 42 & 0 & 42 & 0 & 0 & 0 & 0 & 45 & 0 & 45 & 0 & 45 & 47 \\
 \tcr{9} & 6 & 6 & 0 & 0 & 0 & 0 & 0 & 0 & 48 & 0 & 0 & 0 & 0 & 0 & 0 & 0 & 0 & 0 & 0 & 0 & 0 & 0 & 0 & 0 & 0 & 0 & 0 & 0 & 0 & 0 & 0 & 0 & 0 & 0 & 0 & 0 \\
 \tcr{10} & 6 & 0 & 6 & 0 & 0 & 0 & 0 & 0 & 0 & 48 & 0 & 0 & 0 & 0 & 0 & 0 & 0 & 0 & 0 & 0 & 0 & 0 & 0 & 0 & 0 & 0 & 0 & 0 & 0 & 0 & 0 & 0 & 0 & 0 & 0 & 0 \\
 \tcr{11} & 0 & 6 & 6 & 0 & 0 & 0 & 0 & 0 & 0 & 0 & 48 & 0 & 0 & 0 & 0 & 0 & 0 & 0 & 0 & 0 & 0 & 0 & 0 & 0 & 0 & 0 & 0 & 0 & 0 & 0 & 0 & 0 & 0 & 0 & 0 & 0 \\
 \tcr{12} & 8 & 0 & 0 & 25 & 0 & 0 & 0 & 0 & 0 & 0 & 0 & 62 & 0 & 0 & 0 & 0 & 0 & 0 & 0 & 0 & 0 & 0 & 0 & 0 & 0 & 0 & 0 & 0 & 0 & 0 & 0 & 0 & 0 & 0 & 0 & 0 \\
 \tcr{13} & 0 & 8 & 0 & 0 & 25 & 0 & 0 & 0 & 0 & 0 & 0 & 0 & 62 & 0 & 0 & 0 & 0 & 0 & 0 & 0 & 0 & 0 & 0 & 0 & 0 & 0 & 0 & 0 & 0 & 0 & 0 & 0 & 0 & 0 & 0 & 0 \\
 \tcr{14} & 0 & 0 & 8 & 0 & 0 & 25 & 0 & 0 & 0 & 0 & 0 & 0 & 0 & 62 & 0 & 0 & 0 & 0 & 0 & 0 & 0 & 0 & 0 & 0 & 0 & 0 & 0 & 0 & 0 & 0 & 0 & 0 & 0 & 0 & 0 & 0 \\
 \tcr{15} & 10 & 0 & 0 & 0 & 28 & 0 & 0 & 0 & 0 & 0 & 0 & 0 & 0 & 0 & 74 & 0 & 0 & 0 & 0 & 0 & 0 & 0 & 0 & 0 & 0 & 0 & 0 & 0 & 0 & 0 & 0 & 0 & 0 & 0 & 0 & 0 \\
 \tcr{16} & 10 & 0 & 0 & 0 & 0 & 28 & 0 & 0 & 0 & 0 & 0 & 0 & 0 & 0 & 0 & 74 & 0 & 0 & 0 & 0 & 0 & 0 & 0 & 0 & 0 & 0 & 0 & 0 & 0 & 0 & 0 & 0 & 0 & 0 & 0 & 0 \\
 \tcr{17} & 0 & 10 & 0 & 28 & 0 & 0 & 0 & 0 & 0 & 0 & 0 & 0 & 0 & 0 & 0 & 0 & 74 & 0 & 0 & 0 & 0 & 0 & 0 & 0 & 0 & 0 & 0 & 0 & 0 & 0 & 0 & 0 & 0 & 0 & 0 & 0 \\
 \tcr{18} & 0 & 10 & 0 & 0 & 0 & 28 & 0 & 0 & 0 & 0 & 0 & 0 & 0 & 0 & 0 & 0 & 0 & 74 & 0 & 0 & 0 & 0 & 0 & 0 & 0 & 0 & 0 & 0 & 0 & 0 & 0 & 0 & 0 & 0 & 0 & 0 \\
 \tcr{19} & 0 & 0 & 10 & 28 & 0 & 0 & 0 & 0 & 0 & 0 & 0 & 0 & 0 & 0 & 0 & 0 & 0 & 0 & 74 & 0 & 0 & 0 & 0 & 0 & 0 & 0 & 0 & 0 & 0 & 0 & 0 & 0 & 0 & 0 & 0 & 0 \\
 \tcr{20} & 0 & 0 & 10 & 0 & 28 & 0 & 0 & 0 & 0 & 0 & 0 & 0 & 0 & 0 & 0 & 0 & 0 & 0 & 0 & 74 & 0 & 0 & 0 & 0 & 0 & 0 & 0 & 0 & 0 & 0 & 0 & 0 & 0 & 0 & 0 & 0 \\
 \tcr{21} & 13 & 0 & 0 & 0 & 0 & 0 & 42 & 0 & 0 & 0 & 0 & 0 & 0 & 0 & 0 & 0 & 0 & 0 & 0 & 0 & 89 & 0 & 0 & 0 & 0 & 0 & 0 & 0 & 0 & 0 & 0 & 0 & 0 & 0 & 0 & 0 \\
 \tcr{22} & 13 & 0 & 0 & 0 & 0 & 0 & 0 & 42 & 0 & 0 & 0 & 0 & 0 & 0 & 0 & 0 & 0 & 0 & 0 & 0 & 0 & 89 & 0 & 0 & 0 & 0 & 0 & 0 & 0 & 0 & 0 & 0 & 0 & 0 & 0 & 0 \\
 \tcr{23} & 0 & 13 & 0 & 0 & 0 & 0 & 42 & 0 & 0 & 0 & 0 & 0 & 0 & 0 & 0 & 0 & 0 & 0 & 0 & 0 & 0 & 0 & 89 & 0 & 0 & 0 & 0 & 0 & 0 & 0 & 0 & 0 & 0 & 0 & 0 & 0 \\
 \tcr{24} & 0 & 13 & 0 & 0 & 0 & 0 & 0 & 42 & 0 & 0 & 0 & 0 & 0 & 0 & 0 & 0 & 0 & 0 & 0 & 0 & 0 & 0 & 0 & 89 & 0 & 0 & 0 & 0 & 0 & 0 & 0 & 0 & 0 & 0 & 0 & 0 \\
 \tcr{25} & 0 & 0 & 13 & 0 & 0 & 0 & 42 & 0 & 0 & 0 & 0 & 0 & 0 & 0 & 0 & 0 & 0 & 0 & 0 & 0 & 0 & 0 & 0 & 0 & 89 & 0 & 0 & 0 & 0 & 0 & 0 & 0 & 0 & 0 & 0 & 0 \\
 \tcr{26} & 0 & 0 & 13 & 0 & 0 & 0 & 0 & 42 & 0 & 0 & 0 & 0 & 0 & 0 & 0 & 0 & 0 & 0 & 0 & 0 & 0 & 0 & 0 & 0 & 0 & 89 & 0 & 0 & 0 & 0 & 0 & 0 & 0 & 0 & 0 & 0 \\
 \tcr{27} & 0 & 0 & 0 & 32 & 32 & 0 & 0 & 0 & 0 & 0 & 0 & 0 & 0 & 0 & 0 & 0 & 0 & 0 & 0 & 0 & 0 & 0 & 0 & 0 & 0 & 0 & 100 & 0 & 0 & 0 & 0 & 0 & 0 & 0 & 0 & 0 \\
 \tcr{28} & 0 & 0 & 0 & 32 & 0 & 32 & 0 & 0 & 0 & 0 & 0 & 0 & 0 & 0 & 0 & 0 & 0 & 0 & 0 & 0 & 0 & 0 & 0 & 0 & 0 & 0 & 0 & 100 & 0 & 0 & 0 & 0 & 0 & 0 & 0 & 0 \\
 \tcr{29} & 0 & 0 & 0 & 0 & 32 & 32 & 0 & 0 & 0 & 0 & 0 & 0 & 0 & 0 & 0 & 0 & 0 & 0 & 0 & 0 & 0 & 0 & 0 & 0 & 0 & 0 & 0 & 0 & 100 & 0 & 0 & 0 & 0 & 0 & 0 & 0 \\
 \tcr{30} & 0 & 0 & 0 & 34 & 0 & 0 & 45 & 0 & 0 & 0 & 0 & 0 & 0 & 0 & 0 & 0 & 0 & 0 & 0 & 0 & 0 & 0 & 0 & 0 & 0 & 0 & 0 & 0 & 0 & 105 & 0 & 0 & 0 & 0 & 0 & 0 \\
 \tcr{31} & 0 & 0 & 0 & 34 & 0 & 0 & 0 & 45 & 0 & 0 & 0 & 0 & 0 & 0 & 0 & 0 & 0 & 0 & 0 & 0 & 0 & 0 & 0 & 0 & 0 & 0 & 0 & 0 & 0 & 0 & 105 & 0 & 0 & 0 & 0 & 0 \\
 \tcr{32} & 0 & 0 & 0 & 0 & 34 & 0 & 45 & 0 & 0 & 0 & 0 & 0 & 0 & 0 & 0 & 0 & 0 & 0 & 0 & 0 & 0 & 0 & 0 & 0 & 0 & 0 & 0 & 0 & 0 & 0 & 0 & 105 & 0 & 0 & 0 & 0 \\
 \tcr{33} & 0 & 0 & 0 & 0 & 34 & 0 & 0 & 45 & 0 & 0 & 0 & 0 & 0 & 0 & 0 & 0 & 0 & 0 & 0 & 0 & 0 & 0 & 0 & 0 & 0 & 0 & 0 & 0 & 0 & 0 & 0 & 0 & 105 & 0 & 0 & 0 \\
 \tcr{34} & 0 & 0 & 0 & 0 & 0 & 34 & 45 & 0 & 0 & 0 & 0 & 0 & 0 & 0 & 0 & 0 & 0 & 0 & 0 & 0 & 0 & 0 & 0 & 0 & 0 & 0 & 0 & 0 & 0 & 0 & 0 & 0 & 0 & 105 & 0 & 0 \\
 \tcr{35} & 0 & 0 & 0 & 0 & 0 & 34 & 0 & 45 & 0 & 0 & 0 & 0 & 0 & 0 & 0 & 0 & 0 & 0 & 0 & 0 & 0 & 0 & 0 & 0 & 0 & 0 & 0 & 0 & 0 & 0 & 0 & 0 & 0 & 0 & 105 & 0 \\
 \tcr{36} & 0 & 0 & 0 & 0 & 0 & 0 & 47 & 47 & 0 & 0 & 0 & 0 & 0 & 0 & 0 & 0 & 0 & 0 & 0 & 0 & 0 & 0 & 0 & 0 & 0 & 0 & 0 & 0 & 0 & 0 & 0 & 0 & 0 & 0 & 0 & 110
  \end{smallmatrix}\tag{This is graph $\Phi$ obtained from the stable graph above}
 \end{align}
 \begin{align}
   \begin{smallmatrix}
 0 & \tcb{1} & \tcb{2} & \tcb{3} & \tcb{4} & \tcb{5} & \tcb{6} & \tcb{7} & \tcb{8} & \tcb{9} & \tcb{10} & \tcb{11} & \tcb{12} & \tcb{13} & \tcb{14} & \tcb{15} & \tcb{16} & \tcb{17} & \tcb{18} & \tcb{19} & \tcb{20} & \tcb{21} & \tcb{22} & \tcb{23} & \tcb{24} & \tcb{25} & \tcb{26} & \tcb{27} & \tcb{28} & \tcb{29} & \tcb{30} & \tcb{31} & \tcb{32} & \tcb{33} & \tcb{34} & \tcb{35} & \tcb{36} \\
 \tcr{1} & 0 & 0 & 0 & 0 & 0 & 0 & 0 & 0 & 1 & 1 & 0 & 1 & 0 & 0 & 1 & 1 & 0 & 0 & 0 & 0 & 1 & 1 & 0 & 0 & 0 & 0 & 0 & 0 & 0 & 0 & 0 & 0 & 0 & 0 & 0 & 0 \\
 \tcr{2} & 0 & 0 & 0 & 0 & 0 & 0 & 0 & 0 & 1 & 0 & 1 & 0 & 1 & 0 & 0 & 0 & 1 & 1 & 0 & 0 & 0 & 0 & 1 & 1 & 0 & 0 & 0 & 0 & 0 & 0 & 0 & 0 & 0 & 0 & 0 & 0 \\
 \tcr{3} & 0 & 0 & 0 & 0 & 0 & 0 & 0 & 0 & 0 & 1 & 1 & 0 & 0 & 1 & 0 & 0 & 0 & 0 & 1 & 1 & 0 & 0 & 0 & 0 & 1 & 1 & 0 & 0 & 0 & 0 & 0 & 0 & 0 & 0 & 0 & 0 \\
 \tcr{4} & 0 & 0 & 0 & 0 & 0 & 0 & 0 & 0 & 0 & 0 & 0 & 1 & 0 & 0 & 0 & 0 & 1 & 0 & 1 & 0 & 0 & 0 & 0 & 0 & 0 & 0 & 1 & 1 & 0 & 1 & 1 & 0 & 0 & 0 & 0 & 0 \\
 \tcr{5} & 0 & 0 & 0 & 0 & 0 & 0 & 0 & 0 & 0 & 0 & 0 & 0 & 1 & 0 & 1 & 0 & 0 & 0 & 0 & 1 & 0 & 0 & 0 & 0 & 0 & 0 & 1 & 0 & 1 & 0 & 0 & 1 & 1 & 0 & 0 & 0 \\
 \tcr{6} & 0 & 0 & 0 & 0 & 0 & 0 & 0 & 0 & 0 & 0 & 0 & 0 & 0 & 1 & 0 & 1 & 0 & 1 & 0 & 0 & 0 & 0 & 0 & 0 & 0 & 0 & 0 & 1 & 1 & 0 & 0 & 0 & 0 & 1 & 1 & 0 \\
 \tcr{7} & 0 & 0 & 0 & 0 & 0 & 0 & 0 & 0 & 0 & 0 & 0 & 0 & 0 & 0 & 0 & 0 & 0 & 0 & 0 & 0 & 1 & 0 & 1 & 0 & 1 & 0 & 0 & 0 & 0 & 1 & 0 & 1 & 0 & 1 & 0 & 1 \\
 \tcr{8} & 0 & 0 & 0 & 0 & 0 & 0 & 0 & 0 & 0 & 0 & 0 & 0 & 0 & 0 & 0 & 0 & 0 & 0 & 0 & 0 & 0 & 1 & 0 & 1 & 0 & 1 & 0 & 0 & 0 & 0 & 1 & 0 & 1 & 0 & 1 & 1 \\
 \tcr{9} & 1 & 1 & 0 & 0 & 0 & 0 & 0 & 0 & 48 & 0 & 0 & 0 & 0 & 0 & 0 & 0 & 0 & 0 & 0 & 0 & 0 & 0 & 0 & 0 & 0 & 0 & 0 & 0 & 0 & 0 & 0 & 0 & 0 & 0 & 0 & 0 \\
 \tcr{10} & 1 & 0 & 1 & 0 & 0 & 0 & 0 & 0 & 0 & 48 & 0 & 0 & 0 & 0 & 0 & 0 & 0 & 0 & 0 & 0 & 0 & 0 & 0 & 0 & 0 & 0 & 0 & 0 & 0 & 0 & 0 & 0 & 0 & 0 & 0 & 0 \\
 \tcr{11} & 0 & 1 & 1 & 0 & 0 & 0 & 0 & 0 & 0 & 0 & 48 & 0 & 0 & 0 & 0 & 0 & 0 & 0 & 0 & 0 & 0 & 0 & 0 & 0 & 0 & 0 & 0 & 0 & 0 & 0 & 0 & 0 & 0 & 0 & 0 & 0 \\
 \tcr{12} & 1 & 0 & 0 & 1 & 0 & 0 & 0 & 0 & 0 & 0 & 0 & 62 & 0 & 0 & 0 & 0 & 0 & 0 & 0 & 0 & 0 & 0 & 0 & 0 & 0 & 0 & 0 & 0 & 0 & 0 & 0 & 0 & 0 & 0 & 0 & 0 \\
 \tcr{13} & 0 & 1 & 0 & 0 & 1 & 0 & 0 & 0 & 0 & 0 & 0 & 0 & 62 & 0 & 0 & 0 & 0 & 0 & 0 & 0 & 0 & 0 & 0 & 0 & 0 & 0 & 0 & 0 & 0 & 0 & 0 & 0 & 0 & 0 & 0 & 0 \\
 \tcr{14} & 0 & 0 & 1 & 0 & 0 & 1 & 0 & 0 & 0 & 0 & 0 & 0 & 0 & 62 & 0 & 0 & 0 & 0 & 0 & 0 & 0 & 0 & 0 & 0 & 0 & 0 & 0 & 0 & 0 & 0 & 0 & 0 & 0 & 0 & 0 & 0 \\
 \tcr{15} & 1 & 0 & 0 & 0 & 1 & 0 & 0 & 0 & 0 & 0 & 0 & 0 & 0 & 0 & 74 & 0 & 0 & 0 & 0 & 0 & 0 & 0 & 0 & 0 & 0 & 0 & 0 & 0 & 0 & 0 & 0 & 0 & 0 & 0 & 0 & 0 \\
 \tcr{16} & 1 & 0 & 0 & 0 & 0 & 1 & 0 & 0 & 0 & 0 & 0 & 0 & 0 & 0 & 0 & 74 & 0 & 0 & 0 & 0 & 0 & 0 & 0 & 0 & 0 & 0 & 0 & 0 & 0 & 0 & 0 & 0 & 0 & 0 & 0 & 0 \\
 \tcr{17} & 0 & 1 & 0 & 1 & 0 & 0 & 0 & 0 & 0 & 0 & 0 & 0 & 0 & 0 & 0 & 0 & 74 & 0 & 0 & 0 & 0 & 0 & 0 & 0 & 0 & 0 & 0 & 0 & 0 & 0 & 0 & 0 & 0 & 0 & 0 & 0 \\
 \tcr{18} & 0 & 1 & 0 & 0 & 0 & 1 & 0 & 0 & 0 & 0 & 0 & 0 & 0 & 0 & 0 & 0 & 0 & 74 & 0 & 0 & 0 & 0 & 0 & 0 & 0 & 0 & 0 & 0 & 0 & 0 & 0 & 0 & 0 & 0 & 0 & 0 \\
 \tcr{19} & 0 & 0 & 1 & 1 & 0 & 0 & 0 & 0 & 0 & 0 & 0 & 0 & 0 & 0 & 0 & 0 & 0 & 0 & 74 & 0 & 0 & 0 & 0 & 0 & 0 & 0 & 0 & 0 & 0 & 0 & 0 & 0 & 0 & 0 & 0 & 0 \\
 \tcr{20} & 0 & 0 & 1 & 0 & 1 & 0 & 0 & 0 & 0 & 0 & 0 & 0 & 0 & 0 & 0 & 0 & 0 & 0 & 0 & 74 & 0 & 0 & 0 & 0 & 0 & 0 & 0 & 0 & 0 & 0 & 0 & 0 & 0 & 0 & 0 & 0 \\
 \tcr{21} & 1 & 0 & 0 & 0 & 0 & 0 & 1 & 0 & 0 & 0 & 0 & 0 & 0 & 0 & 0 & 0 & 0 & 0 & 0 & 0 & 89 & 0 & 0 & 0 & 0 & 0 & 0 & 0 & 0 & 0 & 0 & 0 & 0 & 0 & 0 & 0 \\
 \tcr{22} & 1 & 0 & 0 & 0 & 0 & 0 & 0 & 1 & 0 & 0 & 0 & 0 & 0 & 0 & 0 & 0 & 0 & 0 & 0 & 0 & 0 & 89 & 0 & 0 & 0 & 0 & 0 & 0 & 0 & 0 & 0 & 0 & 0 & 0 & 0 & 0 \\
 \tcr{23} & 0 & 1 & 0 & 0 & 0 & 0 & 1 & 0 & 0 & 0 & 0 & 0 & 0 & 0 & 0 & 0 & 0 & 0 & 0 & 0 & 0 & 0 & 89 & 0 & 0 & 0 & 0 & 0 & 0 & 0 & 0 & 0 & 0 & 0 & 0 & 0 \\
 \tcr{24} & 0 & 1 & 0 & 0 & 0 & 0 & 0 & 1 & 0 & 0 & 0 & 0 & 0 & 0 & 0 & 0 & 0 & 0 & 0 & 0 & 0 & 0 & 0 & 89 & 0 & 0 & 0 & 0 & 0 & 0 & 0 & 0 & 0 & 0 & 0 & 0 \\
 \tcr{25} & 0 & 0 & 1 & 0 & 0 & 0 & 1 & 0 & 0 & 0 & 0 & 0 & 0 & 0 & 0 & 0 & 0 & 0 & 0 & 0 & 0 & 0 & 0 & 0 & 89 & 0 & 0 & 0 & 0 & 0 & 0 & 0 & 0 & 0 & 0 & 0 \\
 \tcr{26} & 0 & 0 & 1 & 0 & 0 & 0 & 0 & 1 & 0 & 0 & 0 & 0 & 0 & 0 & 0 & 0 & 0 & 0 & 0 & 0 & 0 & 0 & 0 & 0 & 0 & 89 & 0 & 0 & 0 & 0 & 0 & 0 & 0 & 0 & 0 & 0 \\
 \tcr{27} & 0 & 0 & 0 & 1 & 1 & 0 & 0 & 0 & 0 & 0 & 0 & 0 & 0 & 0 & 0 & 0 & 0 & 0 & 0 & 0 & 0 & 0 & 0 & 0 & 0 & 0 & 100 & 0 & 0 & 0 & 0 & 0 & 0 & 0 & 0 & 0 \\
 \tcr{28} & 0 & 0 & 0 & 1 & 0 & 1 & 0 & 0 & 0 & 0 & 0 & 0 & 0 & 0 & 0 & 0 & 0 & 0 & 0 & 0 & 0 & 0 & 0 & 0 & 0 & 0 & 0 & 100 & 0 & 0 & 0 & 0 & 0 & 0 & 0 & 0 \\
 \tcr{29} & 0 & 0 & 0 & 0 & 1 & 1 & 0 & 0 & 0 & 0 & 0 & 0 & 0 & 0 & 0 & 0 & 0 & 0 & 0 & 0 & 0 & 0 & 0 & 0 & 0 & 0 & 0 & 0 & 100 & 0 & 0 & 0 & 0 & 0 & 0 & 0 \\
 \tcr{30} & 0 & 0 & 0 & 1 & 0 & 0 & 1 & 0 & 0 & 0 & 0 & 0 & 0 & 0 & 0 & 0 & 0 & 0 & 0 & 0 & 0 & 0 & 0 & 0 & 0 & 0 & 0 & 0 & 0 & 105 & 0 & 0 & 0 & 0 & 0 & 0 \\
 \tcr{31} & 0 & 0 & 0 & 1 & 0 & 0 & 0 & 1 & 0 & 0 & 0 & 0 & 0 & 0 & 0 & 0 & 0 & 0 & 0 & 0 & 0 & 0 & 0 & 0 & 0 & 0 & 0 & 0 & 0 & 0 & 105 & 0 & 0 & 0 & 0 & 0 \\
 \tcr{32} & 0 & 0 & 0 & 0 & 1 & 0 & 1 & 0 & 0 & 0 & 0 & 0 & 0 & 0 & 0 & 0 & 0 & 0 & 0 & 0 & 0 & 0 & 0 & 0 & 0 & 0 & 0 & 0 & 0 & 0 & 0 & 105 & 0 & 0 & 0 & 0 \\
 \tcr{33} & 0 & 0 & 0 & 0 & 1 & 0 & 0 & 1 & 0 & 0 & 0 & 0 & 0 & 0 & 0 & 0 & 0 & 0 & 0 & 0 & 0 & 0 & 0 & 0 & 0 & 0 & 0 & 0 & 0 & 0 & 0 & 0 & 105 & 0 & 0 & 0 \\
 \tcr{34} & 0 & 0 & 0 & 0 & 0 & 1 & 1 & 0 & 0 & 0 & 0 & 0 & 0 & 0 & 0 & 0 & 0 & 0 & 0 & 0 & 0 & 0 & 0 & 0 & 0 & 0 & 0 & 0 & 0 & 0 & 0 & 0 & 0 & 105 & 0 & 0 \\
 \tcr{35} & 0 & 0 & 0 & 0 & 0 & 1 & 0 & 1 & 0 & 0 & 0 & 0 & 0 & 0 & 0 & 0 & 0 & 0 & 0 & 0 & 0 & 0 & 0 & 0 & 0 & 0 & 0 & 0 & 0 & 0 & 0 & 0 & 0 & 0 & 105 & 0 \\
 \tcr{36} & 0 & 0 & 0 & 0 & 0 & 0 & 1 & 1 & 0 & 0 & 0 & 0 & 0 & 0 & 0 & 0 & 0 & 0 & 0 & 0 & 0 & 0 & 0 & 0 & 0 & 0 & 0 & 0 & 0 & 0 & 0 & 0 & 0 & 0 & 0 & 110
   \end{smallmatrix}\tag{This is the graph $\Theta$ obtained from $\Phi$ in last page}\label{page2}
\end{align}
\begin{landscape}
\begin{align}\begin{smallmatrix}
 0 & \tcb{1} & \tcb{2} & \tcb{3} & \tcb{4} & \tcb{5} & \tcb{6} & \tcb{7} & \tcb{8} & \tcb{9} & \tcb{10} & \tcb{11} & \tcb{12} & \tcb{13} & \tcb{14} & \tcb{15} & \tcb{16} & \tcb{17} & \tcb{18} & \tcb{19} & \tcb{20} & \tcb{21} & \tcb{22} & \tcb{23} & \tcb{24} & \tcb{25} & \tcb{26} & \tcb{27} & \tcb{28} & \tcb{29} & \tcb{30} & \tcb{31} & \tcb{32} & \tcb{33} & \tcb{34} & \tcb{35} & \tcb{36} \\
 \tcr{1} & 1 & 2 & 2 & 3 & 4 & 4 & 5 & 5 & 6 & 6 & 7 & 8 & 9 & 9 & 10 & 10 & 11 & 12 & 11 & 12 & 13 & 13 & 14 & 14 & 14 & 14 & 15 & 15 & 16 & 17 & 17 & 18 & 18 & 18 & 18 & 19 \\
 \tcr{2} & 2 & 1 & 2 & 4 & 3 & 4 & 5 & 5 & 6 & 7 & 6 & 9 & 8 & 9 & 11 & 12 & 10 & 10 & 12 & 11 & 14 & 14 & 13 & 13 & 14 & 14 & 15 & 16 & 15 & 18 & 18 & 17 & 17 & 18 & 18 & 19 \\
 \tcr{3} & 2 & 2 & 1 & 4 & 4 & 3 & 5 & 5 & 7 & 6 & 6 & 9 & 9 & 8 & 12 & 11 & 12 & 11 & 10 & 10 & 14 & 14 & 14 & 14 & 13 & 13 & 16 & 15 & 15 & 18 & 18 & 18 & 18 & 17 & 17 & 19 \\
 \tcr{4} & 20 & 21 & 21 & 22 & 23 & 23 & 24 & 24 & 25 & 25 & 26 & 27 & 28 & 28 & 29 & 29 & 30 & 31 & 30 & 31 & 32 & 32 & 33 & 33 & 33 & 33 & 34 & 34 & 35 & 36 & 36 & 37 & 37 & 37 & 37 & 38 \\
 \tcr{5} & 21 & 20 & 21 & 23 & 22 & 23 & 24 & 24 & 25 & 26 & 25 & 28 & 27 & 28 & 30 & 31 & 29 & 29 & 31 & 30 & 33 & 33 & 32 & 32 & 33 & 33 & 34 & 35 & 34 & 37 & 37 & 36 & 36 & 37 & 37 & 38 \\
 \tcr{6} & 21 & 21 & 20 & 23 & 23 & 22 & 24 & 24 & 26 & 25 & 25 & 28 & 28 & 27 & 31 & 30 & 31 & 30 & 29 & 29 & 33 & 33 & 33 & 33 & 32 & 32 & 35 & 34 & 34 & 37 & 37 & 37 & 37 & 36 & 36 & 38 \\
 \tcr{7} & 39 & 39 & 39 & 40 & 40 & 40 & 41 & 42 & 43 & 43 & 43 & 44 & 44 & 44 & 45 & 45 & 45 & 45 & 45 & 45 & 46 & 47 & 46 & 47 & 46 & 47 & 48 & 48 & 48 & 49 & 50 & 49 & 50 & 49 & 50 & 51 \\
 \tcr{8} & 39 & 39 & 39 & 40 & 40 & 40 & 42 & 41 & 43 & 43 & 43 & 44 & 44 & 44 & 45 & 45 & 45 & 45 & 45 & 45 & 47 & 46 & 47 & 46 & 47 & 46 & 48 & 48 & 48 & 50 & 49 & 50 & 49 & 50 & 49 & 51 \\
 \tcr{9} & 52 & 52 & 53 & 54 & 54 & 55 & 56 & 56 & 57 & 58 & 58 & 59 & 59 & 60 & 61 & 62 & 61 & 62 & 63 & 63 & 64 & 64 & 64 & 64 & 65 & 65 & 66 & 67 & 67 & 68 & 68 & 68 & 68 & 69 & 69 & 70 \\
 \tcr{10} & 52 & 53 & 52 & 54 & 55 & 54 & 56 & 56 & 58 & 57 & 58 & 59 & 60 & 59 & 62 & 61 & 63 & 63 & 61 & 62 & 64 & 64 & 65 & 65 & 64 & 64 & 67 & 66 & 67 & 68 & 68 & 69 & 69 & 68 & 68 & 70 \\
 \tcr{11} & 53 & 52 & 52 & 55 & 54 & 54 & 56 & 56 & 58 & 58 & 57 & 60 & 59 & 59 & 63 & 63 & 62 & 61 & 62 & 61 & 65 & 65 & 64 & 64 & 64 & 64 & 67 & 67 & 66 & 69 & 69 & 68 & 68 & 68 & 68 & 70 \\
 \tcr{12} & 71 & 72 & 72 & 73 & 74 & 74 & 75 & 75 & 76 & 76 & 77 & 78 & 79 & 79 & 80 & 80 & 81 & 82 & 81 & 82 & 83 & 83 & 84 & 84 & 84 & 84 & 85 & 85 & 86 & 87 & 87 & 88 & 88 & 88 & 88 & 89 \\
 \tcr{13} & 72 & 71 & 72 & 74 & 73 & 74 & 75 & 75 & 76 & 77 & 76 & 79 & 78 & 79 & 81 & 82 & 80 & 80 & 82 & 81 & 84 & 84 & 83 & 83 & 84 & 84 & 85 & 86 & 85 & 88 & 88 & 87 & 87 & 88 & 88 & 89 \\
 \tcr{14} & 72 & 72 & 71 & 74 & 74 & 73 & 75 & 75 & 77 & 76 & 76 & 79 & 79 & 78 & 82 & 81 & 82 & 81 & 80 & 80 & 84 & 84 & 84 & 84 & 83 & 83 & 86 & 85 & 85 & 88 & 88 & 88 & 88 & 87 & 87 & 89 \\
 \tcr{15} & 90 & 91 & 92 & 93 & 94 & 95 & 96 & 96 & 97 & 98 & 99 & 100 & 101 & 102 & 103 & 104 & 105 & 106 & 107 & 108 & 109 & 109 & 110 & 110 & 111 & 111 & 112 & 113 & 114 & 115 & 115 & 116 & 116 & 117 & 117 & 118 \\
 \tcr{16} & 90 & 92 & 91 & 93 & 95 & 94 & 96 & 96 & 98 & 97 & 99 & 100 & 102 & 101 & 104 & 103 & 107 & 108 & 105 & 106 & 109 & 109 & 111 & 111 & 110 & 110 & 113 & 112 & 114 & 115 & 115 & 117 & 117 & 116 & 116 & 118 \\
 \tcr{17} & 91 & 90 & 92 & 94 & 93 & 95 & 96 & 96 & 97 & 99 & 98 & 101 & 100 & 102 & 105 & 106 & 103 & 104 & 108 & 107 & 110 & 110 & 109 & 109 & 111 & 111 & 112 & 114 & 113 & 116 & 116 & 115 & 115 & 117 & 117 & 118 \\
 \tcr{18} & 92 & 90 & 91 & 95 & 93 & 94 & 96 & 96 & 98 & 99 & 97 & 102 & 100 & 101 & 107 & 108 & 104 & 103 & 106 & 105 & 111 & 111 & 109 & 109 & 110 & 110 & 113 & 114 & 112 & 117 & 117 & 115 & 115 & 116 & 116 & 118 \\
 \tcr{19} & 91 & 92 & 90 & 94 & 95 & 93 & 96 & 96 & 99 & 97 & 98 & 101 & 102 & 100 & 106 & 105 & 108 & 107 & 103 & 104 & 110 & 110 & 111 & 111 & 109 & 109 & 114 & 112 & 113 & 116 & 116 & 117 & 117 & 115 & 115 & 118 \\
 \tcr{20} & 92 & 91 & 90 & 95 & 94 & 93 & 96 & 96 & 99 & 98 & 97 & 102 & 101 & 100 & 108 & 107 & 106 & 105 & 104 & 103 & 111 & 111 & 110 & 110 & 109 & 109 & 114 & 113 & 112 & 117 & 117 & 116 & 116 & 115 & 115 & 118 \\
 \tcr{21} & 119 & 120 & 120 & 121 & 122 & 122 & 123 & 124 & 125 & 125 & 126 & 127 & 128 & 128 & 129 & 129 & 130 & 131 & 130 & 131 & 132 & 133 & 134 & 135 & 134 & 135 & 136 & 136 & 137 & 138 & 139 & 140 & 141 & 140 & 141 & 142 \\
 \tcr{22} & 119 & 120 & 120 & 121 & 122 & 122 & 124 & 123 & 125 & 125 & 126 & 127 & 128 & 128 & 129 & 129 & 130 & 131 & 130 & 131 & 133 & 132 & 135 & 134 & 135 & 134 & 136 & 136 & 137 & 139 & 138 & 141 & 140 & 141 & 140 & 142 \\
 \tcr{23} & 120 & 119 & 120 & 122 & 121 & 122 & 123 & 124 & 125 & 126 & 125 & 128 & 127 & 128 & 130 & 131 & 129 & 129 & 131 & 130 & 134 & 135 & 132 & 133 & 134 & 135 & 136 & 137 & 136 & 140 & 141 & 138 & 139 & 140 & 141 & 142 \\
 \tcr{24} & 120 & 119 & 120 & 122 & 121 & 122 & 124 & 123 & 125 & 126 & 125 & 128 & 127 & 128 & 130 & 131 & 129 & 129 & 131 & 130 & 135 & 134 & 133 & 132 & 135 & 134 & 136 & 137 & 136 & 141 & 140 & 139 & 138 & 141 & 140 & 142 \\
 \tcr{25} & 120 & 120 & 119 & 122 & 122 & 121 & 123 & 124 & 126 & 125 & 125 & 128 & 128 & 127 & 131 & 130 & 131 & 130 & 129 & 129 & 134 & 135 & 134 & 135 & 132 & 133 & 137 & 136 & 136 & 140 & 141 & 140 & 141 & 138 & 139 & 142 \\
 \tcr{26} & 120 & 120 & 119 & 122 & 122 & 121 & 124 & 123 & 126 & 125 & 125 & 128 & 128 & 127 & 131 & 130 & 131 & 130 & 129 & 129 & 135 & 134 & 135 & 134 & 133 & 132 & 137 & 136 & 136 & 141 & 140 & 141 & 140 & 139 & 138 & 142 \\
 \tcr{27} & 143 & 143 & 144 & 145 & 145 & 146 & 147 & 147 & 148 & 149 & 149 & 150 & 150 & 151 & 152 & 153 & 152 & 153 & 154 & 154 & 155 & 155 & 155 & 155 & 156 & 156 & 157 & 158 & 158 & 159 & 159 & 159 & 159 & 160 & 160 & 161 \\
 \tcr{28} & 143 & 144 & 143 & 145 & 146 & 145 & 147 & 147 & 149 & 148 & 149 & 150 & 151 & 150 & 153 & 152 & 154 & 154 & 152 & 153 & 155 & 155 & 156 & 156 & 155 & 155 & 158 & 157 & 158 & 159 & 159 & 160 & 160 & 159 & 159 & 161 \\
 \tcr{29} & 144 & 143 & 143 & 146 & 145 & 145 & 147 & 147 & 149 & 149 & 148 & 151 & 150 & 150 & 154 & 154 & 153 & 152 & 153 & 152 & 156 & 156 & 155 & 155 & 155 & 155 & 158 & 158 & 157 & 160 & 160 & 159 & 159 & 159 & 159 & 161 \\
 \tcr{30} & 162 & 163 & 163 & 164 & 165 & 165 & 166 & 167 & 168 & 168 & 169 & 170 & 171 & 171 & 172 & 172 & 173 & 174 & 173 & 174 & 175 & 176 & 177 & 178 & 177 & 178 & 179 & 179 & 180 & 181 & 182 & 183 & 184 & 183 & 184 & 185 \\
 \tcr{31} & 162 & 163 & 163 & 164 & 165 & 165 & 167 & 166 & 168 & 168 & 169 & 170 & 171 & 171 & 172 & 172 & 173 & 174 & 173 & 174 & 176 & 175 & 178 & 177 & 178 & 177 & 179 & 179 & 180 & 182 & 181 & 184 & 183 & 184 & 183 & 185 \\
 \tcr{32} & 163 & 162 & 163 & 165 & 164 & 165 & 166 & 167 & 168 & 169 & 168 & 171 & 170 & 171 & 173 & 174 & 172 & 172 & 174 & 173 & 177 & 178 & 175 & 176 & 177 & 178 & 179 & 180 & 179 & 183 & 184 & 181 & 182 & 183 & 184 & 185 \\
 \tcr{33} & 163 & 162 & 163 & 165 & 164 & 165 & 167 & 166 & 168 & 169 & 168 & 171 & 170 & 171 & 173 & 174 & 172 & 172 & 174 & 173 & 178 & 177 & 176 & 175 & 178 & 177 & 179 & 180 & 179 & 184 & 183 & 182 & 181 & 184 & 183 & 185 \\
 \tcr{34} & 163 & 163 & 162 & 165 & 165 & 164 & 166 & 167 & 169 & 168 & 168 & 171 & 171 & 170 & 174 & 173 & 174 & 173 & 172 & 172 & 177 & 178 & 177 & 178 & 175 & 176 & 180 & 179 & 179 & 183 & 184 & 183 & 184 & 181 & 182 & 185 \\
 \tcr{35} & 163 & 163 & 162 & 165 & 165 & 164 & 167 & 166 & 169 & 168 & 168 & 171 & 171 & 170 & 174 & 173 & 174 & 173 & 172 & 172 & 178 & 177 & 178 & 177 & 176 & 175 & 180 & 179 & 179 & 184 & 183 & 184 & 183 & 182 & 181 & 185 \\
 \tcr{36} & 186 & 186 & 186 & 187 & 187 & 187 & 188 & 188 & 189 & 189 & 189 & 190 & 190 & 190 & 191 & 191 & 191 & 191 & 191 & 191 & 192 & 192 & 192 & 192 & 192 & 192 & 193 & 193 & 193 & 194 & 194 & 194 & 194 & 194 & 194 & 195
 \end{smallmatrix}\tag{This is the stable graph $\WL(\mathtt{bi}(X))$ of binding graph $\mathtt{bi}(X)$ obtained by WL process.}\label{page3}
 \end{align}
\end{landscape}
\section{An Explicit Proof (Sketch) of $\hat{\Phi}\approx\hat{\Theta}$ in Theorem \ref{thm:theta}}\label{sec:proof}

\noindent\emph{Proof of $\hat{\Phi}\approx\hat{\Theta}$ in Theorem \ref{thm:theta}}. Since $\Theta\rightarrowtail\Phi$
 by the definition of $\Theta$, it holds that $\hat{\Theta}\rightarrowtail\hat{\Phi}$.

To show the other side, let's look at the process to evaluate the stable graph $\hat{\Theta}$ with SaS stabilization. The first step is to make $\Theta$ recognize vertices. From definition of $\Theta$, we only has to rename the labels $x_0$ of all basic vertices as a variable $y\in\Var$ with $y\notin\Theta$. The result graph is denoted as $\Theta_1=(\theta_{ij})$.

In this setting, for any $i\ne j$ and $i,j\in[n_1]$, $\theta_{ii}$ is $y$ iff $i$ is a basic vertex; $\phi_{ii}$ otherwise. $\theta_{ij}$ is $x$ iff one of $i,j$ is a basic vertex and the other is its binding vertex; $x_0$ otherwise.

Now let $\Theta^2_1:=(y_{ij})$, where  $y_{ij}=\sum_{k=1}^{n_1}\theta_{ik}\theta_{jk}$. Recall that vetices in $[n]$ are all basic vertices and $n_1=n(n+1)/2$. We consider the different cases separately.
\begin{itemize}
  \item Let $i\in[n]$ be any basic vertex. We have
\begin{align}\label{eq:thetabv1}
y_{ii}&=\sum_{k=1}^{n_1}\theta_{ik}^2=y^2+(n-1)x^2+(n_1-n)x_0^2.
\end{align}

  \item Let $i\in[n+1..n_1]$ be any binding vertex. We have
\begin{align}\label{eq:thetabv2}
y_{ii}&=\sum_{k=1}^{n_1}\theta_{ik}^2=\phi_{ii}^2+2x^2+(n_1-3)x_0^2.
\end{align}

  \item Let $i,j\in[n]$ be two different basic vertices. We have
\begin{align}\label{eq:thetabv3}
  y_{ij}&=\sum_{k=1}^{n_1}\theta_{ik}\theta_{jk}=\theta_{ii}\theta_{ij}+\theta_{ij}\theta_{jj}+\sum_{k\in[n]\backslash\{i,j\}}\theta_{ik}\theta_{jk}+\sum_{k=n+1}^{n_1}\theta_{ik}\theta_{jk}\notag\\
   &=\big(\,2 x_0y+(n-2)x_0^2\,\big)+\big(\,2(n-2)xx_0+x^2+(n_1-n-2(n-2)-1)x_0^2\,\big)\notag\\
&=2x_0y+x^2+2(n-2)xx_0+(n_1-2n+1)x_0^2.
\end{align}

\item Let $i,j\in[n+1..n_1]$ be two different binding vertices. If $i$ and $j$ bind a common basic vertex, they are called binding siblings. We have
\begin{align}\label{eq:thetabv4}
  y_{ij}&=\sum_{k=1}^{n_1}\theta_{ik}\theta_{jk}=\theta_{ii}\theta_{ij}+
\theta_{ij}\theta_{jj}+\sum_{k=1}^{n}\theta_{ik}\theta_{jk}+
\sum_{\stackrel{k\in[n+1..n_1]}{k\notin\{i,j\}}}\theta_{ik}\theta_{jk}\notag\\
   &=\left\{
       \begin{array}{ll}
        x^2+2xx_0+(n-3)x_0^2+(\phi_{ii}+\phi_{jj})x_0+(n_1-n-2)x_0^2 , & \hbox{if $i,j$ are binding siblings;} \\
          4xx_0+(n-4)x_0^2+(\phi_{ii}+\phi_{jj})x_0+(n_1-n-2)x_0^2, & \hbox{otherwise.}
       \end{array}
     \right.\notag\\
&=\left\{
       \begin{array}{ll}
        x^2+2xx_0+(\phi_{ii}+\phi_{jj})x_0+(n_1-5)x_0^2 , & \hbox{if $i$ and  $j$ are binding siblings;} \\
          4xx_0+(\phi_{ii}+\phi_{jj})x_0+(n_1-6)x_0^2, & \hbox{otherwise.}
       \end{array}
     \right.
\end{align}

\item Let $i\in[n]$ be a basic vertex and $j\in[n+1..n_1]$ be a  binding vertex. We have
\begin{align}\label{eq:thetabv5}
  y_{ij}&=\sum_{k=1}^{n_1}\theta_{ik}\theta_{jk}=\theta_{ii}\theta_{ij}+
\theta_{ij}\theta_{jj}+\sum_{\stackrel{k\in[n]}{k\ne i}}\theta_{ik}\theta_{jk}+
\sum_{\stackrel{k\in[n+1..n_1]}{k\ne j}}\theta_{ik}\theta_{jk}\notag\\
   &=\left\{
       \begin{array}{ll}
         xy+x\phi_{jj}+(n-2)x_0^2+xx_0+(n-2)xx_0+(n_1-n-(n-1))x_0^2, & \hbox{if $j$ binds $i$;} \\
         x_0y+x_0\phi_{jj}+(n-3)x_0^2+2xx_0+(n-1)xx_0+(n_1-2n)x_0^2, & \hbox{otherwise.}
       \end{array}
     \right.\notag\\
&=\left\{
       \begin{array}{ll}
         xy+x\phi_{jj}+(n-1)xx_0+(n_1-n-1)x_0^2, & \hbox{if $j$ binds $i$;} \\
         x_0y+x_0\phi_{jj}+(n+1)xx_0+(n_1-n-3)x_0^2, & \hbox{otherwise.}
       \end{array}
     \right.
\end{align}
\end{itemize}
Let $\Theta_2:=(z_{ij})$ is the graph after equivalent variable substitution to $\Theta_1^2$. In the graph $\Theta_2$, we have the following observations.
\begin{itemize}
  \item[(a)] For any basic vertices $i,j\in[n]$, we have $z_{ii}=z_{jj}$ from \eqref{eq:thetabv1}, and they are independent of labels of binding vertices of $\Phi$.

 \item[(b)] For any two binding vertices $i,j\in[n+1..n_1]$, then $z_{ii}=z_{jj}$ if and only if $\theta_{ii}=\theta_{jj}$ according to \eqref{eq:thetabv2}.

  \item[(c)]  For all basic vertices $i,j,u,v$ and $i\ne j, u\ne v$ we have $z_{ij}=z_{uv}$ by \eqref{eq:thetabv3}, and they are independent of labels of binding vertices of $\Phi$.
  \item[(d)] The labels on edges between binding siblings do not overlap with those labels on edges between binding non-siblings by  \eqref{eq:thetabv4}.  The labels on both vertices will contribute to those labels.
\item[(e)] The label on a binding edge dependents on the label to binding vertex. The labels on binding edges do not overlap with labels on non-binding edges by  \eqref{eq:thetabv5}.%
\end{itemize}
We now look at $\Theta_2^2:=(w_{ij})=\big(\sum_{k=1}^{n_1}z_{ik}z_{kj}\big)$. Recall that the first $n$ vertices are basic vertices. We justify the labels $w_{ii}, w_{jj}$ for all $i,j\in[n]$. Since
$$
w_{ii}=\sum_{k=1}^{n_1}z_{ik}^2=z_{ii}^2+
\sum_{\stackrel{k\in[n]}{k\ne i}}z_{ik}^2+\sum_{\stackrel{u\in[n]\backslash\{i\}}{k=i\dot{\wedge}u}}z_{ik}^2+
\sum_{\stackrel{ u,v\in[n]\backslash\{i\}}{k=u\dot{\wedge}v}}z_{ik}^2$$
and
$$w_{jj}=\sum_{k=1}^{n_1}z_{jk}^2=z_{jj}^2+
\sum_{\stackrel{k\in[n]}{k\ne i}}z_{jk}^2+\sum_{\stackrel{u\in[n]\backslash\{j\}}{k=j\dot{\wedge}u}}z_{jk}^2+
\sum_{\stackrel{u,v\in[n]\backslash\{j\}}{k=u\dot{\wedge}v}}z_{jk}^2\,,$$
we have, for any $i,j\in[n]$, $z^2_{ii}=z^2_{jj}$ and $\sum_{\stackrel{k\in[n]}{k\ne i}}z_{ik}^2=\sum_{\stackrel{k\in[n]}{k\ne j}}z_{jk}^2$ by the observations (a) and (c).

Careful inspection from \eqref{eq:thetabv5} will gives:
\begin{align}\label{eq:thetabv7}
  \mset{\theta_{kk}\mid k=i\dot{\wedge}u,\,u\in[n]\backslash\{i\}}\equiv\mset{\theta_{kk}\mid k=j\dot{\wedge}u,\,u\in[n]\backslash\{j\}}
\end{align}
 if and only if $$\sum_{\stackrel{u\in[n]\backslash\{i\}}{k=i\dot{\wedge}u}}z_{ik}^2=\sum_{\stackrel{u\in[n]\backslash\{j\}}{k=j\dot{\wedge}u}}z_{jk}^2 \quad\text{and}\quad\sum_{\stackrel{u,v\in[n]\backslash\{i\}}{k=u\dot{\wedge}v}}z_{ik}^2=\sum_{\stackrel{u,v\in[n]\backslash\{j\}}{k=u\dot{\wedge}v}}z_{jk}^2\,.$$

That is equivalent to $w_{ii}=w_{jj}$. We get $w_{ii}=w_{jj}$ if and only if \eqref{eq:thetabv7} holds for basic vertices $i,j\in[n]$ in $\Phi$. However, \eqref{eq:thetabv7} holds iff
\begin{align}\label{eq:thetabv8}
  \mset{\phi_{kk}\mid k=i\dot{\wedge}u,\,u\in[n]\backslash\{i\}}\equiv\mset{\phi_{kk}\mid k=j\dot{\wedge}u,\,u\in[n]\backslash\{j\}}
\end{align}

We let readers convince themselves that for any binding vertices $i,j\in[n+1,n_1]$, we have $w_{ii}=w_{jj}$ if and only if $\phi_{ii}=\phi_{jj}$ holds in $\Phi$.

 Let the graph $\Theta_3$ be graph after the equivalent variable substitution to $\Theta_2^2$. That shows, the labels to basic vertices in $\Theta_3$ are equivalent to those in $\Phi$ and labels to binding vertices keep ``stead still'' and are equivalent to those in $\Theta$ and hence in $\Phi$. That is, the diagonal of $\Phi$ is recreated equivalently in $\Theta_3$.

Since $\Theta_3\rightarrowtail\hat{\Theta}$ and stable graph $\hat{\Theta}$ recognizes binding edges and induces strongly equitable partition. The analysis above together with Proposition \ref{prop:dbqd} implies $\Phi\rightarrowtail\hat{\Theta}$, and hence $\hat{\Phi}\rightarrowtail\hat{\Theta}$.

We thus have $\hat{\Theta}\approx\hat{\Phi}$. That finishes the proof.
\endproof

\section{An Explicit Proof (Sketch) of $\WL(\Phi)\approx\WL(\Theta)$ in Theorem \ref{thm:theta}}\label{sec:proof1}

\noindent\emph{Proof of\quad $\WL(\Phi)\approx\WL(\Theta)$ in Theorem \ref{thm:theta}}. Since $\Theta\rightarrowtail\Phi$
 by the definition of $\Theta$, it holds that $\WL(\Theta)\rightarrowtail\WL(\Phi)$.

To show the other side, let's look at the process to evaluate the stable graph $\WL(\Theta)$ with WL process. The first step is to make $\Theta$ recognize vertices. From definition of $\Theta$, we only has to rename the labels $x_0$ of all basic vertices as a variable $y\in\Var$ with $y\notin\Theta$. The result graph is denoted as $\Theta_1=(\theta_{ij})$.

In this setting, for any $i\ne j$ and $i,j\in[n_1]$, $\theta_{ii}$ is $y$ iff $i$ is a basic vertex; $\phi_{ii}$ otherwise. $\theta_{ij}$ is $x$ iff one of $i,j$ is a basic vertex and the other is its binding vertex; $x_0$ otherwise.

Now let $\Theta_1\dd\Theta_1:=(y_{ij})$, where  $y_{ij}=\sum_{k=1}^{n_1}\theta_{ik}\dd\theta_{jk}$. Recall that vetices in $[n]$ are all basic vertices and $n_1=n(n+1)/2$. We consider the different cases separately as follows.
\begin{itemize}
  \item Let $i\in[n]$ be any basic vertex. We have
\begin{align}\label{eq:thetabv1d}
y_{ii}&=\sum_{k=1}^{n_1}\theta_{ik}\dd\theta_{ki}=y\dd y+(n-1)x\dd x+(n_1-n)x_0\dd x_0.
\end{align}

  \item Let $i\in[n+1..n_1]$ be any binding vertex. We have
\begin{align}\label{eq:thetabv2d}
y_{ii}&=\sum_{k=1}^{n_1}\theta_{ik}\dd\theta_{ik}=\phi_{ii}\dd\phi_{ii}+2x\dd x+(n_1-3)x_0\dd x_0.
\end{align}

  \item Let $i,j\in[n]$ be two different basic vertices. We have
\begin{align}\label{eq:thetabv3d}
  y_{ij}&=\sum_{k=1}^{n_1}\theta_{ik}\dd\theta_{kj}
  =\theta_{ii}\dd\theta_{ij}+\theta_{ij}\dd\theta_{jj}+\sum_{k\in[n]\backslash\{i,j\}}\theta_{ik}\dd\theta_{kj}+\sum_{k=n+1}^{n_1}\theta_{ik}\dd\theta_{kj}\notag\\
   &=\,y\dd x_0+x_0\dd y+ \sum_{k\in[n]\backslash\{i,j\}}x_0\dd x_0+\sum_{k=n+1}^{n_1}\theta_{ik}\dd\theta_{kj}\notag\\
    &=\,y\dd x_0+x_0\dd y+ (n-2)x_0\dd x_0+\notag\\
    &\hspace{1cm}+x\dd x+(n-2)(x\dd x_0+x_0\dd x)+(n_1-n-2(n-2)-1)x_0\dd x_0\notag\\
&=\,y\dd x_0+x_0\dd y+ x\dd x+(n-2)(x\dd x_0+x_0\dd x)+(n_1-2n+1)x_0\dd x_0.
\end{align}

\item Let $i,j\in[n+1..n_1]$ be two different binding vertices. If $i$ and $j$ bind a common basic vertex, they are called binding siblings. We have
\begin{align}
  y_{ij}&=\sum_{k=1}^{n_1}\theta_{ik}\dd\theta_{kj}=\theta_{ii}\dd\theta_{ij}+
\theta_{ij}\dd\theta_{jj}+\sum_{k=1}^{n}\theta_{ik}\dd\theta_{kj}+
\sum_{\stackrel{k\in[n+1..n_1]}{k\notin\{i,j\}}}\theta_{ik}\dd\theta_{kj}\notag\\
&=\phi_{ii}\dd x_0+
x_0\dd\phi_{jj}+(n_1-n-2)\,(x_0\dd x_0)+
\sum_{k=1}^{n}\theta_{ik}\dd\theta_{kj}\notag
\end{align}
Notice that if $i,j$ are binding siblings, then
$$\sum_{k=1}^{n}\theta_{ik}\dd\theta_{kj}=x\dd x+x_0\dd x+x\dd x_0+(n-3)x_0\dd x_0.$$ And if $i,j$ are not binding siblings, then $$\sum_{k=1}^{n}\theta_{ik}\dd\theta_{kj}=2(x_0\dd x+x\dd x_0)+(n-4)(x_0\dd x_0)\,.$$
These give the following.
\begin{align}
y_{ij}&=\left\{
       \begin{array}{ll}
       \phi_{ii}\dd x_0+
x_0\dd\phi_{jj}+
x\dd x+x_0\dd x+x\dd x_0+(n_1-5)x_0\dd x_0, & \hbox{if $i,j$ are binding siblings;} \\
          \phi_{ii}\dd x_0+
x_0\dd\phi_{jj}+
2(x_0\dd x+x\dd x_0)+(n-6)(x_0\dd x_0), & \hbox{otherwise.}
       \end{array}
     \right.\label{eq:thetabv4d}
\end{align}

\item Let $i\in[n]$ be a basic vertex and $j\in[n+1..n_1]$ be a  binding vertex. We have
\begin{align}
  y_{ij}&=\sum_{k=1}^{n_1}\theta_{ik}\dd\theta_{kj}=\theta_{ii}\dd\theta_{ij}+
\theta_{ij}\dd\theta_{jj}+\sum_{\stackrel{k\in[n]}{k\ne i}}\theta_{ik}\dd\theta_{kj}+
\sum_{\stackrel{k\in[n+1..n_1]}{k\ne j}}\theta_{ik}\dd\theta_{kj}\notag\\
&=y\dd\theta_{ij}+
\theta_{ij}\dd\phi_{jj}+\sum_{\stackrel{k\in[n]}{k\ne i}}\theta_{ik}\dd\theta_{kj}+
\sum_{\stackrel{k\in[n+1..n_1]}{k\ne j}}\theta_{ik}\dd\theta_{kj}.\notag
\end{align}
If $j$ binds $i$, then
\begin{align}
 y_{ij}&=y\dd x+
x\dd\phi_{jj}+x_0\dd x+(n-2)(x_0\dd x_0)+(n-2)(x\dd x_0)+
(n_1-n-(n-1))(x_0\dd x_0)\notag\\
&=y\dd x+
x\dd\phi_{jj}+x_0\dd x+(n-2)(x\dd x_0)+
(n_1-n-1))(x_0\dd x_0)\,.\label{eq:thetabv5d}
\end{align}
If $j$ does not bind $i$, then
\begin{align}
 y_{ij}&=y\dd x_0+
x_0\dd\phi_{jj}+(n-3)(x_0\dd x_0)+2(x_0\dd x)+\notag\\
&\hspace{2cm}+(n-1)(x\dd x_0)+(n_1-2n)(x_0\dd x_0)\notag\\
&=y\dd x_0+
x_0\dd\phi_{jj}+2(x_0\dd x)+(n-1)(x\dd x_0)+
(n_1-n-3)(x_0\dd x_0)\,. \label{eq:thetabv6d}
\end{align}
\end{itemize}
Let $\Theta_2:=(z_{ij})$ is the graph after equivalent variable substitution to $\Theta_1\dd\Theta_1$. In the graph $\Theta_2$, we have the following observations.
\begin{itemize}
  \item[(a)] For any basic vertices $i,j\in[n]$, we have $z_{ii}=z_{jj}$ from \eqref{eq:thetabv1d}, and they are independent of labels of binding vertices of $\Phi$.

 \item[(b)] For any two binding vertices $i,j\in[n+1..n_1]$, then $z_{ii}=z_{jj}$ if and only if $\theta_{ii}=\theta_{jj}$ according to \eqref{eq:thetabv2d}.

  \item[(c)]  For all basic vertices $i,j,u,v$ and $i\ne j, u\ne v$ we have $z_{ij}=z_{uv}$ by \eqref{eq:thetabv3d}, and they are independent of labels of binding vertices of $\Phi$.
  \item[(d)] The labels on edges between binding siblings do not overlap with those labels on edges between binding non-siblings by  \eqref{eq:thetabv4d}.  The labels on both vertices will contribute to those labels.
\item[(e)] The label on a binding edge dependents on the label to binding vertex. The labels on binding edges do not overlap with labels on non-binding edges by  \eqref{eq:thetabv5d},\eqref{eq:thetabv6d}.%
\end{itemize}
We now look at $\Theta_2\dd \Theta_2:=(w_{ij})=\big(\sum_{k=1}^{n_1}z_{ik}\dd z_{kj}\big)$. Recall that the first $n$ vertices are basic vertices. We justify the labels $w_{ii}, w_{jj}$ for all $i,j\in[n]$. Since
$$
w_{ii}=\sum_{k=1}^{n_1}z_{ik}\dd z_{ki}=z_{ii}\dd z_{ii}+
\sum_{\stackrel{k\in[n]}{k\ne i}}z_{ik}\dd z_{ki}+\sum_{\stackrel{u\in[n]\backslash\{i\}}{k=i\dot{\wedge}u}}z_{ik}\dd z_{ki}+
\sum_{\stackrel{ u,v\in[n]\backslash\{i\}}{k=u\dot{\wedge}v}}z_{ik}\dd z_{ki}$$
and
$$w_{jj}=\sum_{k=1}^{n_1}z_{jk}\dd z_{kj}=z_{jj}\dd z_{jj}+
\sum_{\stackrel{k\in[n]}{k\ne i}}z_{jk}\dd z_{kj}+\sum_{\stackrel{u\in[n]\backslash\{j\}}{k=j\dot{\wedge}u}}z_{jk}\dd z_{kj}+
\sum_{\stackrel{u,v\in[n]\backslash\{j\}}{k=u\dot{\wedge}v}}z_{jk}\dd z_{kj}\,,$$
we have, for any $i,j\in[n]$, $z_{ii}\dd z_{ii}=z_{jj}\dd z_{jj}$ and $\sum_{\stackrel{k\in[n]}{k\ne i}}z_{ik}\dd z_{ki}=\sum_{\stackrel{k\in[n]}{k\ne j}}z_{jk}\dd z_{kj}$ by the observations (a) and (c).

Careful inspection from \eqref{eq:thetabv5d} and \eqref{eq:thetabv6d} will gives:
\begin{align}\label{eq:thetabv7d}
 \mset{\theta_{kk}\mid k=i\dot{\wedge}u,\,u\in[n]\backslash\{i\}}\equiv\mset{\theta_{kk}\mid k=j\dot{\wedge}u,\,u\in[n]\backslash\{j\}}
\end{align}
 if and only if $$\sum_{\stackrel{u\in[n]\backslash\{i\}}{k=i\dot{\wedge}u}}z_{ik}\dd z_{ki}=\sum_{\stackrel{u\in[n]\backslash\{j\}}{k=j\dot{\wedge}u}}z_{jk}\dd z_{kj} \quad\text{and}\quad\sum_{\stackrel{u,v\in[n]\backslash\{i\}}{k=u\dot{\wedge}v}}z_{ik}\dd z_{ki}=\sum_{\stackrel{u,v\in[n]\backslash\{j\}}{k=u\dot{\wedge}v}}z_{jk}\dd z_{kj}\,.$$

That is equivalent to $w_{ii}=w_{jj}$. We get $w_{ii}=w_{jj}$ if and only if \eqref{eq:thetabv7d} holds for basic vertices $i,j\in[n]$ in $\Phi$. However,
equation \eqref{eq:thetabv7d} is equivalent to
\begin{align}\label{eq:thetabv8d}
 \mset{\phi_{kk}\mid k=i\dot{\wedge}u,\,u\in[n]\backslash\{i\}}\equiv\mset{\phi_{kk}\mid k=j\dot{\wedge}u,\,u\in[n]\backslash\{j\}}
\end{align}

We let readers convince themselves that for any binding vertices $i,j\in[n+1,n_1]$, we have $w_{ii}=w_{jj}$ if and only if $\phi_{ii}=\phi_{jj}$ in $\Phi$.

 Let the graph $\Theta_3$ be graph after the equivalent variable substitution to $\Theta_2\dd\Theta_2$. That shows, the labels to basic vertices in $\Theta_3$ are equivalent to those in $\Phi$ and labels to binding vertices keep ``stead still'' and are equivalent to those in $\Theta$ and hence in $\Phi$. That is, the diagonal of $\Phi$ is recreated equivalently in $\Theta_3$. That shows $\Phi\rightarrowtail\Theta_3\rightarrowtail\WL(\Theta)$.
 That implies $\Phi\rightarrowtail\WL(\Theta)$. We finally get  $\WL(\Phi)\rightarrowtail\WL(\Theta)$.

We thus have $\WL(\Theta)\approx\WL(\Phi)$. That finishes the proof.
\endproof

\end{document}